\documentclass[10pt]{article}
\usepackage{graphicx,epsfig,amsmath,amsthm,amssymb}
\usepackage{enumerate}
\usepackage{multirow}
\usepackage{bbm}
\usepackage{balance}
\usepackage{footnote}
\usepackage{tablefootnote}
\usepackage{xr}
\newcommand{\ex}[1]{\mathbb{E}\left[ #1 \right] }
\newcommand{\norm}[1]{\left\lVert #1 \right\rVert}
\newcommand{\Q}{ {\mathbf Q}}

\newcommand{\Qc}{ {\mathbf Q}_{\perp}}
\newcommand{\Qp}{ {\mathbf Q}_{\parallel}}
\newcommand{\A}{ {\mathbf A}}
\newcommand{\s}{ {\mathbf S}}
\newcommand{\UU}{ {\mathbf U}}
\newcommand{\inner}[2]{\langle #1, #2 \rangle}
\newcommand{\lep}[1]{\mathop  \le \limits^{(#1)}}
\newcommand{\gep}[1]{\mathop  \ge \limits^{(#1)}}
\newcommand{\gp}[1]{\mathop  > \limits^{(#1)}}
\newcommand{\ep}[1]{\mathop  = \limits^{(#1)}}

\newtheorem{definition}{Definition}
\newtheorem{lemma}{Lemma}

\newtheorem{proposition}{Proposition}
\newtheorem{assumption}{Assumption}
\newtheorem{theorem}{Theorem}

\newtheorem{remark}{Remark}

\begin{document}

\title{Designing Low-Complexity Heavy-Traffic Delay-Optimal Load Balancing Schemes: Theory to Algorithms}  %
\author{Xingyu Zhou\\Department of ECE\\The Ohio State University\\zhou.2055@osu.edu \and Fei Wu\\
Department of CSE \\The Ohio State University\\wu.1973@osu.edu \and Jian Tan\\Department of ECE\\The Ohio State University\\tan.252@osu.edu\\
\and Yin Sun\\Department of ECE\\Auburn University\\yzs0078@auburn.edu\\ \and Ness Shroff\\Department of ECE and CSE\\The Ohio State University\\
shroff.11@osu.edu }

\maketitle
\begin{abstract}

We establish a unified analytical framework for load balancing systems, which allows us to construct a general class $\Pi$ of policies that are both throughput optimal and heavy-traffic delay optimal.
This general class $\Pi$ includes as special cases  popular policies such as join-shortest-queue and power-of-$d$, but not the join-idle-queue (JIQ) policy.  
In fact we show that JIQ, which is not in $\Pi$, is actually not heavy-traffic delay optimal.
Owing to the significant flexibility offered by class $\Pi$, we are able to design a new policy called join-below-threshold (JBT-d), which maintains  the simplicity of pull-based policies such as JIQ, but updates its threshold dynamically.  We prove that JBT-$d$ belongs to the class $\Pi$ when the threshold is picked appropriately and thus it is heavy-traffic delay optimal.
Extensive simulations show that the new policy not only has a low complexity in message rates, but also achieves excellent delay performance, comparable to the optimal join-shortest-queue in various system settings.
\end{abstract}



\section{Introduction}

	Load balancing, which is responsible for dispatching jobs on parallel servers, is a key component in computer networks and distributed computing systems. For a large number of practical applications, such as, Web service~\cite{gupta2007analysis}, distributed caching systems (e.g., Memcached~\cite{nishtala2013scaling}), large data stores (e.g., HBase~\cite{george2011hbase}), embarrassingly parallel computations~\cite{dean2008mapreduce} and grid computing~\cite{foster2008cloud}, the system performance critically depends on the load balancing algorithm it employs. 


In a load balancing system, there are two directions of message flows: push messages (from the dispatcher to the servers) and pull messages (from the servers to the dispatcher).   In a push-based policy, the dispatcher actively sends query messages to the servers and waits for their responses; 
In a pull-based policy,  the dispatcher passively listens to the report  from the servers. The job dispatching decision is conducted at the dispatcher based on the pull-messages sent from the servers. 
Push-based policies (e.g., the join-shortest-queue (JSQ) policy~\cite{weber1978optimal},~\cite{eryilmaz2012asymptotically} and the power-of-$d$ policy~\cite{mitzenmacher2001power},~\cite{vvedenskaya1996queueing}) have been shown to be delay optimal in the heavy-traffic regime~\cite{eryilmaz2012asymptotically},~\cite{maguluri2014heavy}.  
Recently, the pulled-based policies such as join-idle-queue (JIQ)~\cite{lu2011join} and the equivalent one in~\cite{stolyar2015pull}, have been proposed. Compared with the push-based policies, these pull-based policies not only achieve good delay performance, but also have some nice features, such as, lower message overhead, lower computational complexity, and zero dispatching delay.
However, as shown in the simulations of~\cite{lu2011join}, the delay performance of existing pull-based polices will degrade substantially as the load gets higher. In fact, as shown in Theorem~\ref{THM:JIQ} of this paper, JIQ is not heavy-traffic delay optimal even for homogeneous servers. Therefore, one key question is how to design load balancing policies that are heavy-traffic delay optimal and meanwhile possess all the nice features of pull-based policies such as zero dispatching delay, low message overhead and low computational complexity.

In this paper, we take a systematic approach to address this question. To that end, the main contributions of this paper are summarized as follows: 

\begin{itemize}

	\item  We derive inner-product based sufficient conditions for proving that a load-balancing policy is throughput optimal and heavy-traffic delay optimal. Using these sufficient conditions, we obtain a general class $\Pi$ of load balancing policies that are throughput optimal and heavy-traffic delay optimal. This class of load balancing policies contains the famous (push-based) JSQ and the power-of-d policies as special cases, but not the (pull-based) JIQ policy.

	\item On the other hand, we show that JIQ, which is not in $\Pi$, is not heavy-traffic delay optimal even for homogeneous servers. While it has been empirically shown in the past that the delay using JIQ is quite bad at high loads, the question of whether it was heavy-traffic delay optimal in homogeneous servers has been previously left unsolved. Furthermore, our novel Lyapunov-drift approach offers a new avenue to show a policy is not heavy-traffic delay optimal. 

	\item Owing to the significant flexibility offered by class $\Pi$, we are able to design a new policy called Join-Below-Threshold (JBT-$d$).
	To the best of our knowledge, this is the first load balancing policy that guarantees heavy-traffic delay optimality while enjoying nice features of pull-based policy, e.g., zero dispatching delay, low message overhead and low computational complexity. Through extensive simulations, we demonstrate that JBT-$d$ has excellent delay performance for different system sizes and various arrival and service processes over a large range of traffic loads.

\end{itemize}

	The rest of the paper is organized as follows. Section \ref{sec:relatedwork} reviews the related work on load balancing schemes. Section~\ref{sec:notation} introduces the necessary 
	notation in the paper.  Section~\ref{sec:model} describes the system model and the related definitions. 
	Section \ref{sec:mainresult} presents the main results of the paper. In particular, a class $\Pi$ of flexible 
	load balancing policies are introduced, containing as special cases the popular existing ones and motivating new ones.  
	Sufficient conditions are derived to guarantee throughput and heavy-traffic delay optimality.  Section~\ref{sec:simulation} contains the simulation results on comparing 
	different policies, demonstrating the performance and simplicity of our new policy.
	Section \ref{sec:proofMain} contains the proof of the main results.

\subsection{Related work: push versus pull}
\label{sec:relatedwork}
	This section reviews state-of-the-art load balancing policies with a focus on the system performance in heavy traffic.  We group these policies mainly into two categories: push-based and pull-based as shown in Fig. \ref{fig:model}.  

	\textbf{Push-based policy:} Under a push-based policy, the dispatcher tries to ``push'' jobs to servers. More specifically, upon each job arrival, the dispatcher sends probing messages to the servers, which feed back the required information for dispatching decisions, e.g., queue lengths. After receiving the feedback, the dispatcher sends the incoming jobs to servers based on a dispatching distribution. 
	A classical example in this category is the JSQ policy, under which the dispatcher queries the queue length information of each server upon new job arrivals, 
	and sends the incoming jobs to the server with the shortest queue, with ties broken randomly. It has been shown~\cite{weber1978optimal} that for homogeneous servers this policy is delay optimal in a stochastic ordering sense under the assumption of renewal arrival and non-decreasing failure rate service. In the heavy-traffic regime, it has been proved that it is heavy-traffic delay optimal for both heterogeneous and homogeneous servers~\cite{eryilmaz2012asymptotically}. Nevertheless, the performance of this policy comes at the cost of substantial overhead as it has to sample the queue lengths of all the servers, which is undesirable in large-scale systems. To overcome this problem, an alternative load balancing policy called power-of-$d$ has been introduced~\cite{mitzenmacher2001power},~\cite{vvedenskaya1996queueing}; see also related works~\cite{ying2015power},~\cite{tsitsiklis2013queueing}. Under this policy, the dispatcher routes all the incoming jobs to the server that has the shortest queue length, with ties broken randomly,  out of the $d$ servers sampled uniformly at random. This policy has also been shown to be heavy-traffic optimal for homogeneous servers~\cite{maguluri2014heavy}. However, for heterogeneous servers, the power-of-$d$ policy is neither throughput optimal, nor delay optimal in heavy traffic. 

	\textbf{Pull-based policy:} Under a pull-based policy, the servers spontaneously send messages to  ``pull''  jobs from the dispatcher according to a fixed policy. One illustrative example is the JIQ policy~\cite{lu2011join} and the equivalent one in~\cite{stolyar2015pull}.  Under the JIQ policy, each server sends a pull message to the dispatcher whenever it becomes idle. Upon job arrivals, the dispatcher checks the available pull messages in memory.  
	If such messages exist,  it removes one uniformly at random, and sends the jobs to the corresponding server.  Otherwise, the new jobs will be dispatched uniformly at random to one of the servers in the system. 
	This policy has several favorable properties. The most important property is that the required number of messages in steady-state is at most one for each job arrival, which is smaller than the $2d$ of the power-of-$d$-choices ($d$ for query and $d$ for response per job).  However, as already shown in~\cite{lu2011join}, when the load becomes heavy, the performance of JIQ keeps empirically degrades substantially, and in fact, in Theorem~\ref{THM:JIQ} we show that it is not heavy traffic delay optimal even for homogeneous servers. 



	\subsection{Notations}
	\label{sec:notation}
	We use boldface letters to denote vectors in $\mathbb{R}^N$ and ordinary letters for scalers. Denote by $\overline{\Q}$ the random vector whose probability distribution is the same as the steady-state distribution of $\{\Q(t), t\ge 0\}$. The dot product in $\mathbb{R}^N$ is denoted by $\inner{\mathbf{x} }{ {\mathbf{y} } } := \sum_{i=1}^N x_iy_i$. For any $\mathbf{x} \in \mathbb{R}^N$, the $l_1$ norm is denoted by $\norm{\mathbf{x} }_1 := \sum_{n=1}^N |x_n|$ and $l_2$ norm is denoted by $\norm{\mathbf{x} } := \sqrt{\inner{\mathbf{x} }{\mathbf{x} } }$. The parallel and perpendicular component of the queue length vector $\Q$ with respect to a vector $\mathbf{c}$ with unit norm is denoted by $\Qp:=\inner{\mathbf{c}}{\Q}\mathbf{c}$  and $\Qc:=\Q - \Qp$, respectively.

	\begin{figure}[t]\vspace{-0.0cm}
	\graphicspath{{./Figures/}}
	\centering
	\includegraphics[width=10cm]{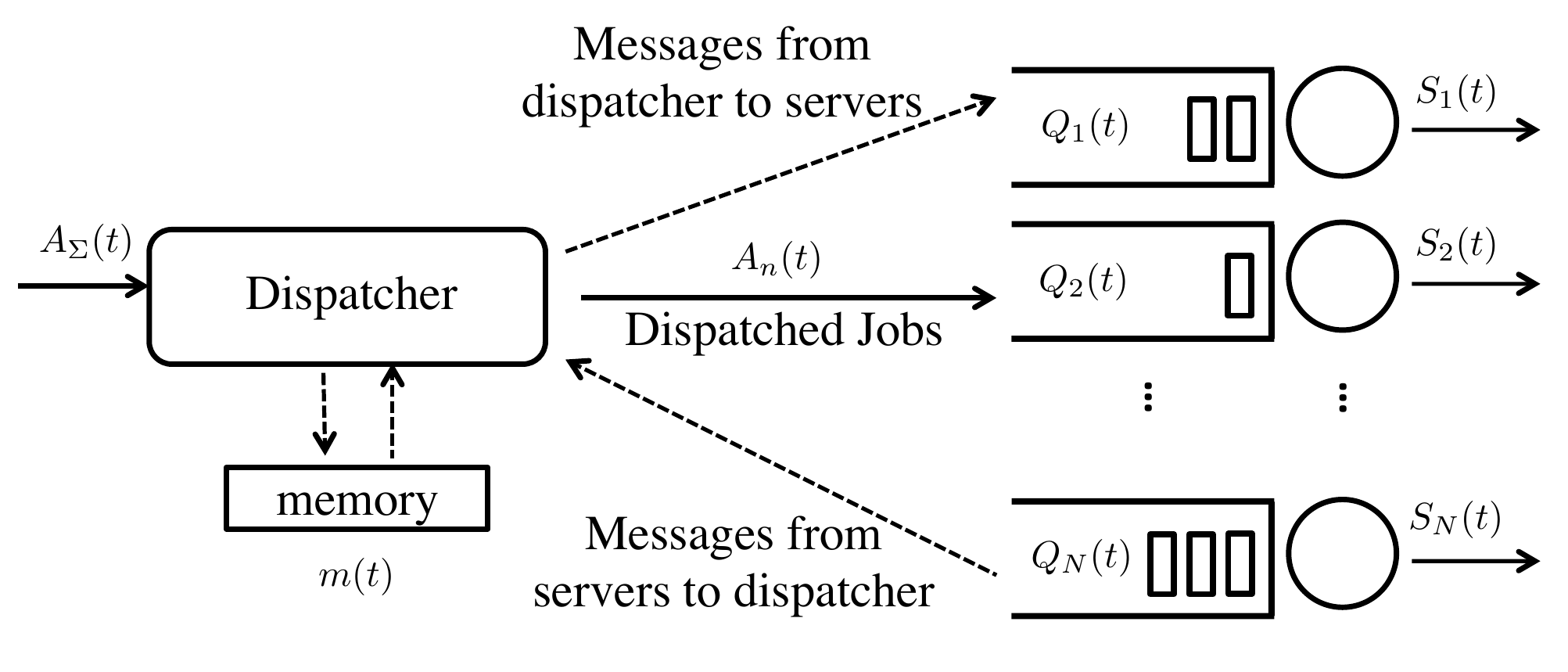}
	\caption{System model of general load balancing. \normalfont{(a) For push-based policy, we have $m(t) = \emptyset$ for all $t$ since it does not require any memory. The message exchange is bidirectional: probing from the dispatcher and feedback messages from servers. (b) For pull-based policy, $m(t)$ stores the ID of the servers that satisfy a certain condition at time $t$. The message exchange is unidirectional, i.e., only the pull-message is sent from the servers to the dispatcher.}}\label{fig:model}
	\vspace{-0.4cm}
	\end{figure}  

\section{Model and Definitions}
\label{sec:model}
	This section describes a general model for the load balance systems as shown in Fig. \ref{fig:model}, and introduces necessary definitions.   
	  \subsection{Model Description}
	Consider a time-slotted load balancing system,  with one central dispatcher and $N$ parallel servers.  These servers are
	  indexed by its ID $n = 1,2,\ldots,N$. Each server $n$ is associated with a FIFO (first-in, first-out) queue of length $Q_n(t)$ at the beginning of time slot $t$, $t=0,1, 2, \ldots$. Thus, we use index $n$ to represent both the server and the associated queue.  Once a job joins a queue, it will remain in that queue until its service is completed.

	  \begin{assumption}[Arrival Process]
	  \label{assume:arrival}
	   	 Let $A_\Sigma(t)$ and $A_n(t)$ denote the number of exogenous job arrivals and the number of arrivals routed to queue $n$ at time slot $t$, respectively.   We assume that all the exogenous arrivals at time $t$ are routed to one selected queue $s$, using the standard model as in \cite{eryilmaz2012asymptotically},~\cite{maguluri2014heavy}, i.e., $A_s(t) = A_\Sigma(t), s\in \mathcal{N} = \{1,2,\ldots, N\} $ and $A_i(t) = 0,$ for all $i \in \mathcal{N}\setminus\{s\}$.  The job arrival process $\{ A_\Sigma(t), t \ge 0\}$ is a nonnegative integer valued stochastic process that is \emph{i.i.d} across time $t$, with mean $\ex{A_\Sigma(t)} = \lambda_\Sigma$ and variance $\text{Var} (A_\Sigma(t))  = \sigma_\Sigma^2$.  
	  We further assume that the number of exogenous arrivals at each time slot is bounded by a constant, i.e., $A_\Sigma(t) \le A_{\max} < \infty$ for all $t \ge 0$.
	  \end{assumption}

	  \begin{assumption}[Service Process]
	  \label{assume:service}
	  	  Let $S_n(t)$ denote the potential service offered to queue $n$ at time $t$, which represents the maximum number of jobs that can be served in time slot $t$. Therefore, if the offered service $S_n(t)$ is larger than the number of pending jobs in queue $n$ at time slot~$t$,  it will cause an unused service
	 $U_n(t)$, as defined in~(\ref{eq:q_dynamic}). 
	     For each $n$,  the process $\{S_n(t), t\ge 0\}$ is a nonnegative integer valued \emph{i.i.d.} stochastic process with mean $\ex{S_n(t)} = \mu_n$ and variance $\text{Var}(S_n(t)) = \nu_n^2$. Moreover, $\lambda_{\Sigma} < \sum_{i=1}^N \mu_n$. Furthermore,  the processes $\{S_n(t), t\ge 0\}, n \in \mathcal{N}$ are mutually independent across different queues, which are also independent of the arrival processes. 
	 The offered service $S_n(t)$ to each queue is uniformly bounded by a constant, i.e., $S_n(t) \le S_{\max} < \infty$ for all $t \ge 0$ and all $n \in \mathcal{N}$.
	  \end{assumption}

	Let $\Q(t) = \{Q_1(t), \ldots, Q_N(t)\}$ be the queue lengths observed at the beginning of time $t$.  
	Define $m(t)$ to be the set of server IDs stored in the dispatcher at the beginning of time slot $t$.  
	  In general, the dispatcher makes the decision of $A_n(t)$ based on $(\Q(t), m(t))$ for each time slot $t$. This includes the cases that the dispatching decision depends only on $\Q(t)$ (e.g., JSQ), partial information of $\Q(t)$ (e.g., power-of-$d$) or only on $m(t)$ (e.g., JIQ).
	  In each time slot,  the queueing dynamics evolves according to the following
	  procedure. 
	  The job arrivals occur at the beginning of time slot $t$. Then, 
	   the dispatching decision $A_n(t)$ is selected 
	   based on $(\Q(t), m(t))$. Further, the routed jobs are processed by the allocated servers.  
	    Thus,  the queueing dynamics is given by the following equation, 

	    \begin{equation}
	    \label{eq:q_dynamic}
	    	\begin{split}
	    		Q_n(t+1) &= \left[ Q_n(t) + A_n(t) - S_n(t) \right]\\
	    		& = Q_n(t) + A_n(t) - S_n(t) + U_n(t),
	    	\end{split}
	    \end{equation}
	    where $[x]^+ = \max(0,x)$, $U_n(t) = \max(S_n(t) - Q_n(t) - A_n(t),0)$ denotes the unused service of queue $n$.

\subsection{Definitions}
	The load balancing system is modeled as a discrete-time Markov chain $\{Z(t) = (\Q(t), m(t)), t\ge0 \}$ with state space $\mathcal{Z}$, using queue length vector $\Q(t)$ together with the memory state $m(t)$. We consider a system  $\{Z^{(\epsilon)}(t),t\ge0 \}$ parameterized by $\epsilon$, i.e., the exogenous arrival process is  $\{ A_\Sigma^{(\epsilon)}(t), t \ge 0\}$ with $\lambda_\Sigma^{(\epsilon)} = \mu_\Sigma - \epsilon = \sum_n \mu_n - \epsilon$. That is, we use $\epsilon$ to indicate the distance of arrival rate to the capacity boundary, and it is also adopted as a superscript to represent the corresponding random variables and processes. 
	\begin{definition}[Stability] 
	$\{Z^{(\epsilon)}(t),t\ge0 \}$ is said to be stable if we have 
	      \begin{displaymath}
	        \limsup_{C \to \infty}\limsup_{t\to\infty}\mathbb{P}\left(\sum_n Q_n^{(\epsilon)}(t) > C\right) = 0.
	      \end{displaymath}
	\end{definition}

	A load balancing policy is said to be throughput optimal if it stabilizes the system under any arrival rate in the capacity region. Since the capacity region in our model is simply $\lambda_\Sigma < \mu_\Sigma$, the definition of throughput optimality is given as follows.
	\begin{definition}[Throughput Optimality]
	      A load balancing policy is said to be throughput optimal if it stabilizes $\{Z^{(\epsilon)}(t),t\ge0 \}$ for any $\epsilon > 0$.
	    \end{definition}

	For the definition of heavy-traffic delay optimality, we need the following definition and property.

	\begin{definition}[Resource-pooled System]
	    A single-server FCFS (first-come, first-serve) system $\{q^{(\epsilon)}(t),t\ge0\}$ is said to be the resource-pooled system with respect to $\{Z^{(\epsilon)}(t),t\ge0 \}$, if its arrival and service process satisfy $a^{(\epsilon)}(t) = A_\Sigma^{(\epsilon)}(t)$ and $s(t) = \sum S_n(t)$ for all $t\ge0$. Then, we have 
	      \begin{equation}
	      \label{eq:stochastic}
	      	\ex{q^{(\epsilon)}(t)} \le \ex{\sum Q_n^{(\epsilon)}(t)},
	      \end{equation}
	      for all $t\ge 0$ and $\epsilon > 0$.
	    \end{definition}

	     In words, a resource-pooled system is a system that merges the total resource of $N$ servers and queues to a single server with a single queue. 
	     Eq.~\eqref{eq:stochastic} holds due to the fact for any $t$, the overall arrivals to the resource-pooled system and to  load balancing system are the same, and the overall service in the resource-pooled system is stochastically larger than the overall service in the load balancing system. This is due to the fact that the jobs in load balancing system cannot be moved from one queue to another, which often results in a strict inequality in Eq.~\eqref{eq:stochastic}. However, in the heavy-traffic regime, this lower bound can be achieved under some policy in an asymptotic sense as defined in the next definition, and hence based on Little's law this policy achieves the minimum average delay of the system.

	    \begin{definition}[Heavy-traffic Delay Optimality]
	    \label{def:heavy-opt}
	       A load balancing policy is said to be heavy-traffic delay optimal if the stationary workload of  $\{Z^{(\epsilon)}(t),t\ge0 \}$ under all the arrival and service processes in Assumptions \ref{assume:arrival} and \ref{assume:service}, satisfies \footnote{Assume $(\sigma_\Sigma^{(\epsilon)})^2$ converges to a constant.}

	      \begin{equation}
	      \label{eq:heavy_traffic_def}
	        \lim_{\epsilon \downarrow 0} \epsilon\ex{\sum_n\overline{Q}_n^{(\epsilon)}} = \lim_{\epsilon \downarrow 0} \epsilon \ex{\overline{q}^{(\epsilon)}},
	      \end{equation}
	      where $\overline{\Q}$ is the random vector whose probability distribution is the same as the steady-state distribution of $\{\Q(t),t\ge 0\}$.
	    \end{definition}


	    \begin{remark}
	    \label{rem:not_optimal}
	    	Based on the definition above, in order to show a policy, say $\mathcal{P}_1$, is not heavy-traffic delay optimal, it is sufficient to find a class of $\{A_\Sigma^{(\epsilon)}(t)\}$ and $\{S_n(t)\}$ such that Eq. \eqref{eq:heavy_traffic_def} does not hold. In other words, there exists a class of arrival and service processes for which policy $\mathcal{P}_1$ cannot achieve the lower bound (i.e., the resource-pooled system) while JSQ can (since it is heavy-traffic delay optimal).
	    \end{remark}

\section{Main Results}
\label{sec:mainresult}
In this section, we introduce a class $\Pi$ of load balancing policies which are proven to be delay-optimal in the heavy-traffic regime.
Popular load balancing policies, such as JSQ and power-of-$d$, are special cases in $\Pi$; but the JIQ policy does not belong to $\Pi$ as we will show in Theorem~\ref{THM:JIQ} that it is not heavy-traffic delay optimal. In order to improve the delay performance of JIQ while maintaining its low message overhead and simplicity, we develop a new load balancing policy named join-below-threshold (JBT-$d$), which is  heavy-traffic delay-optimal as we can show JBT-$d$ is in $\Pi$ and has a low message overhead similar to JIQ.
 
\subsection{The Class of Load Balancing Policies {\large $\Pi$} }
\label{sec:pi}

Let us denote $\mathbf{p}(t)=(p_1(t),\ldots,p_N(t))$, where $p_n(t)$ is the probability that the new arrivals  in time slot $t$ are dispatched to queue $n$ such that $\sum_{n=1}^N p_n(t) = 1$.
We consider a class of load balancing policies in which $\mathbf{p}(t)$ is a function of the system state $Z(t)=\{\Q(t),m(t)\}$.  
Consider a permutation $\sigma_t(\cdot)$ of $(1,2,\ldots, N)$ that satisfies $Q_{\sigma_t(1)}(t)\le Q_{\sigma_t(2)}(t)\le \ldots \le Q_{\sigma_t(N)}(t)$ for all $t$, i.e.,  the queues are sorted according to an increasing order of the queue lengths in time slot $t$ with ties broken randomly. Define $\mathbf{P}(t)=(P_1(t),\ldots,P_N(t))$ such that $\mathbf{P}(t)$ is a permutation of $\mathbf{p}(t)$ with $P_n(t) = p_{\sigma_t(n)}(t)$. Let
	\begin{align}
		\Delta_n(t) & = p_{\sigma_t(n)}(t) - {\mu_{\sigma_t(n)}}/{\mu_\Sigma}\nonumber\\
		&=P_n(t) - {\mu_{\sigma_t(n)}}/{\mu_\Sigma}.\label{eq_delta}
	\end{align}
	



	\begin{definition}[Equivalence in inner-product]
		  A dispatching distribution $\mathbf{\hat{P}}(t) $ is said to be  \emph{equivalent to another dispatching distribution $\mathbf{P}(t)$ in inner product}, if 
		\begin{equation}
			\sum_n Q_{\sigma_t(n)}\Delta_n(t) = \sum_n Q_{\sigma_t(n)}\hat{\Delta}_n(t), \label{eq_equivalent1}
		\end{equation}
		or equivalently, if 
				\begin{equation}
			\sum_n Q_{\sigma_t(n)}P_n(t) = \sum_n Q_{\sigma_t(n)}\hat{P}_n(t).\label{eq_equivalent2}
		\end{equation}
	\end{definition}
	
	The equivalence between \eqref{eq_equivalent1} and \eqref{eq_equivalent2} follows immediately from \eqref{eq_delta}.
Intuitively speaking, a load-balancing policy is `good' if the inner product between $\Q_{\sigma_t}(t)$ and $\mathbf{P}(t)$ is as small as possible such that more packets are dispatched to shorter queues. If $\mathbf{P}(t)$ is  equivalent to $\mathbf{\hat{P}}(t)$ in inner-product, we can replace $\mathbf{\hat{P}}(t)$ by $\mathbf{P}(t)$ without affecting the property of the policy in heavy-traffic regime, which will be explained in details later.


	The following definitions enable us to distinguish different load balancing policies based on $\mathbf{P}(t)$ or equivalently $\Delta(t)$: 


	\begin{definition}[Tilted distribution]
		 A dispatching distribution $\mathbf{P}(t)$ is said to be \emph{tilted}, if there exists $k\in\{2,\ldots,N\}$ such that $\Delta_n(t) \ge 0$ for all $n<k$ and  $\Delta_n(t) \le 0$ for all $n\ge k$.
	\end{definition}

	\begin{definition}[$\delta$-tilted distribution]
		  A dispatching distribution $\mathbf{P}(t)$ is said to be \emph{$\delta$-tilted}, if (i) $\mathbf{P}(t)$ is tilted and (ii) these exists a constant $\delta>0$ such that $\Delta_1(t) \ge \delta$ and $\Delta_N(t)\le -\delta$.
	\end{definition}

	Some examples are presented in Fig.~\ref{fig:dispatch} to facilitate the understanding of tilted distribution, $\delta$-tilted distribution, and equivalence in inner-product. Fig.~\ref{fig:dispatch} (a)-(f) illustrate six dispatching distributions $\mathbf{P}(t)$. The queue state $\Q(t)$ is given by  (i) or (ii). The service rates are $\mu_A = \mu_B = \mu_C = \mu_D = 1$ such that $\mu_{i}/{\mu_\Sigma} =1/4$ for $i=A,B,C,D$. By direct computation, one can obtain that $P_n(t)$ is tilted in scenario (a), (b), (d), (e), and (f), and is $\delta$-tilted in scenario (d), (e), and (f). If $\Q(t)$ is in the State (i), there is no tie in the queue length and hence the permutation $\sigma_t(\cdot)$ is unique, which means that $\mathbf{P}(t)$ is fully determined by $\mathbf{p}(t)$. If 
$\Q(t)$ is in the State (ii), all queue lengths are equal and hence the permutation $\sigma_t(\cdot)$ is non-unique, which means that $\mathbf{P}(t)$ is determined by both $\mathbf{p}(t)$ and $\sigma_t(\cdot)$. In this case, however, the  inner product between $\Q_{\sigma_t}(t)$ and $\mathbf{P}(t)$ is one in all (a)-(f), and hence the dispatching distributions $\mathbf{P}(t)$ in (a)-(f) are mutually equivalent in inner product. For example, in this case even though $\mathbf{P}(t)$ in (c) is neither tilted nor $\delta$-tilted, it is equivalent in inner product to $\mathbf{P}(t)$ in (d) which is both tilted and $\delta$-tilted.

	From the perspective of heavy-traffic delay performance, tilted distribution is a dispatching distribution that is not worse than random routing and $\delta$-tilted distribution is a dispatching distribution that is strictly better than random routing. In addition, the equivalence in inner-product allows us to 
transfer a tilted dispatching distribution to a $\delta$-tilted dispatching distribution when there are ties in queue lengths, that is, it allows to merges probability in $\mathbf{P}(t)$ from longer queues to shorter queues without changing the inner product.

	\begin{figure}[t]\vspace{-0.0cm}
	\graphicspath{{./Figures/}}
	\centering
	\includegraphics[width=10cm]{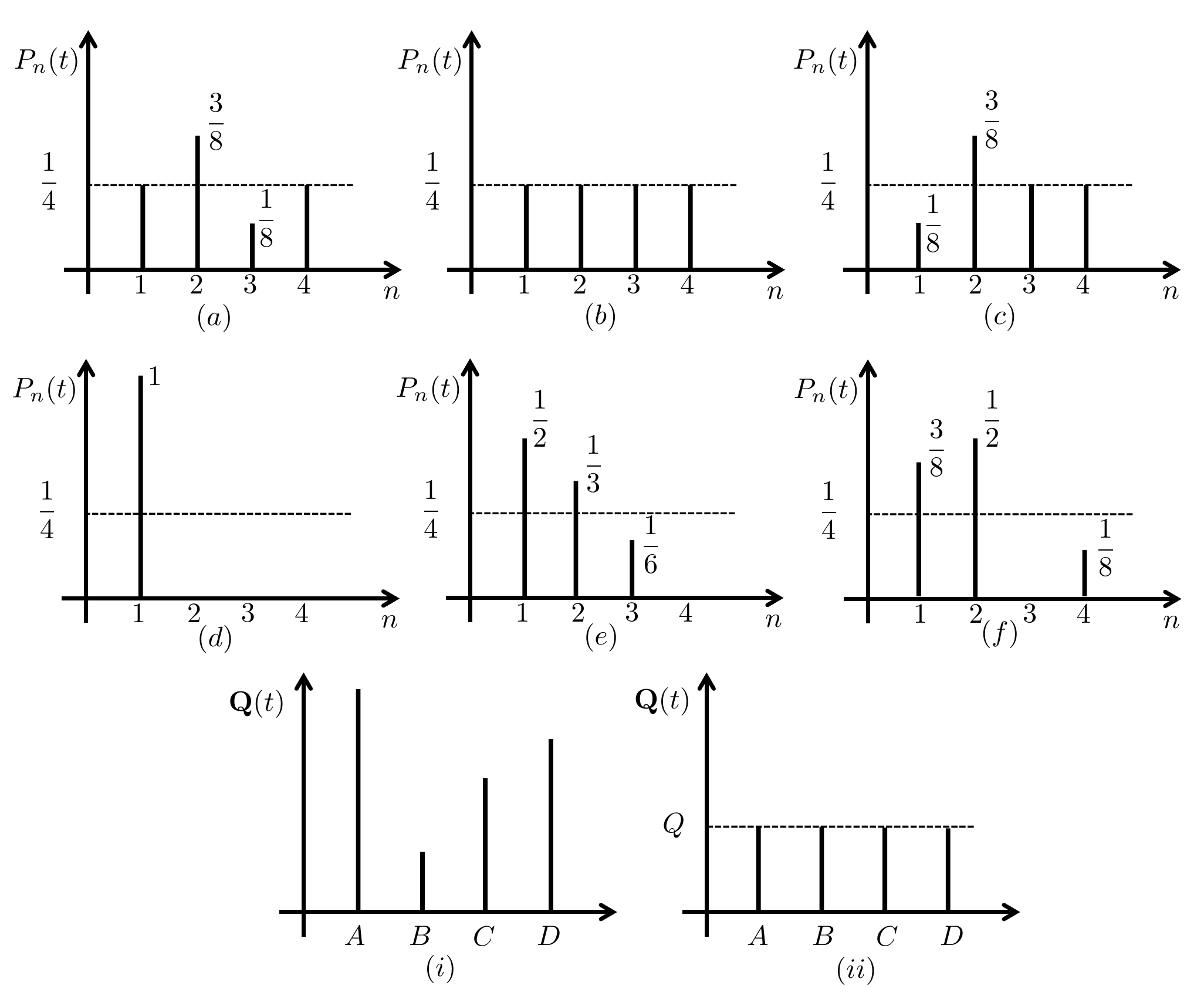}
	\caption{  \small \normalfont{Illustrations of tilted distribution, $\delta$-tilted distribution, and equivalence in inner-product.}}  \label{fig:dispatch} 
	\end{figure}

	We now introduce a class of load balancing algorithms $\Pi$ based 
	on the property of $\mathbf{P}(t)$ or its equivalent distributions in inner product.

	\begin{definition}
		A load balancing policy is said to belong to class $\Pi$ if it satisfies the following two conditions:
		\begin{enumerate}[(i)]
			\item $\mathbf{P}(t)$ or one of its equivalent distributions in inner product is tilted  for all $Z(t)$ and $t\geq0$.
			\item For some finite positive constants $T$ and $\delta$ that both are  independent of $\epsilon$, there exists a time slot $t_k\in\{kT,kT+1,\ldots,(k+1)T-1\}$ for each $k\in \mathbb{N}$ such that $\mathbf{P}(t_k)$ or one of its equivalent distributions in inner product is $\delta$-tilted for all $Z(t_k)$.
			
		\end{enumerate}

	\end{definition}

	In the sequel, we will show that any policy in $\Pi$ satisfies the following two sufficient conditions for throughput and heavy-traffic delay optimality, which are obtained via the Lyapunov-drift based approach developed in~\cite{eryilmaz2012asymptotically}.

	\begin{lemma}
	      \label{LEMMA:THTROUGHPUT}
	    If there exist $T_1>0$, $K_1 \ge 0$, and $\gamma>0$ such that for all $t_0=1,2,\ldots$, all $Z\in \mathcal{Z}$ and  $\lambda_\Sigma < \mu_\Sigma$
	        \begin{equation}
	          \label{eq:condition_throughput}
	            \ex{\sum_{t=t_0}^{t_0+T_1-1} \inner{\Q(t)}{\A(t) -\s(t)} \mid Z(t_0) = Z} \le -\gamma \norm{\Q} + K_1,
	            \end{equation}
	         then the system is throughput-optimal. Moreover, the stationary distribution of the queueing system has bounded moments, i.e., there exist finite $M_r$ such that for all $\epsilon > 0 $ and $r\in \mathbb{N}$
	            \begin{equation*}
	              \ex{\norm{\overline{\Q}^{(\epsilon)} } ^r} \le M_r.
	        \end{equation*}
	    \end{lemma}
	    \begin{proof} See Section \ref{sec:proof_lemma_throughput} of the Appendix.\end{proof}

	     \begin{lemma}
	    \label{LEMMA:DELAY}
	Under the assumptions of Lemma \ref{LEMMA:THTROUGHPUT}, if there further exist $T_2 >0$,  $K_2\ge0$ and $\eta > 0$ that are independent of $\epsilon$,  such that for all $t_0=1,2,\ldots$ and all $Z \in \mathcal{Z}$
	        \begin{equation}
	        \label{eq:cond_delay}
	          \ex{\sum_{t=t_0}^{t_0+T_2-1} \inner{\Qc(t)}{\A(t) -\s(t)} \mid Z(t_0) = Z} \le -\eta \norm{\Qc} +K_2
	        \end{equation}
	        holds for all $\epsilon \in (0,\epsilon_0)$, $\epsilon_0 > 0$, where $\Qc = \Q - \inner{\Q}{\mathbf{c}}\mathbf{c}$ is the perpendicular component of $\Q$ with respect to the line ${\bf c} = \frac{1}{\sqrt{N} }{(1,1,\ldots,1)}$, then the system is heavy-traffic delay optimal, i.e., 

	        \begin{displaymath}
	        \lim_{\epsilon \downarrow 0} \epsilon\ex{\sum_n\overline{Q}_n^{(\epsilon)}} = \lim_{\epsilon \downarrow 0} \epsilon \ex{\overline{q}^{(\epsilon)}}.
	        \end{displaymath}
	    \end{lemma}
	    \begin{proof} See Section \ref{sec:proof_lemma_delay} of the Appendix.\end{proof}

	\begin{remark}
		Note that these two sufficient conditions distilled from the Lyapunov-drift based approach not only provide a unified approach for throughput and heavy-traffic optimality analysis, but also enable us to abstract a class of heavy-traffic delay optimal policies. In particular, using Lemma \ref{LEMMA:THTROUGHPUT} and Lemma \ref{LEMMA:DELAY}, we are able to prove the main result of this paper.
	\end{remark}
	\begin{theorem}
	\label{THM:ACLASS}
		  Any load balancing policy in $\Pi$ is throughput optimal and heavy-traffic delay optimal.
	\end{theorem}
	\begin{proof}[Proof sketch of Theorem \ref{THM:ACLASS}] 
\noindent
	The insight for a policy in $\Pi$ to satisfy the sufficient condition in Eq.~\eqref{eq:condition_throughput} is that under tilted dispatching distribution the performance is no worse than random dispatching. This follows from the following property of tilted distribution
		\begin{equation}
		\label{eq:th_idea}
			\sum_{n=1}^N Q_{\sigma_t(n)}(t) \Delta_n(t) \le 0.
		\end{equation}
	The equality is obtained when all $\Delta_n(t)$ is zero, which is the case of random dispatching as shown in (b) of Fig. \ref{fig:dispatch}. Note that for all other cases of a tilted distribution, Eq. \eqref{eq:th_idea} is strictly less than zero. This is true since  $\sum_{n=1}^N \Delta_n(t)$ is always zero and the permutation is in the non-decreasing order of the queue length.

	The intuition for a policy in $\Pi$ to satisfy the sufficient condition in Eq. \eqref{eq:cond_delay} is that the performance under any $\delta$-tilted dispatching distribution is strictly better than random dispatching, under which the term in Eq. \eqref{eq:cond_delay} is $0$ for homogeneous servers and of order $\epsilon$ for heterogeneous servers. Note that under a $\delta$-tilted distribution, we have
	\begin{equation}
	\label{eq:delay_insight}
		\sum_{n=1}^N Q_{\sigma_t(n)}(t) \Delta_n(t) \le - \delta(Q_{\sigma_t(N)}(t) - Q_{\sigma_t(1)}(t).
	\end{equation}
	This inequality comes from the definition of the $\delta$-tilted distribution and fact that the permutation is in the non-decreasing order of the queue length. In order to have the term of $\norm{\Qc}$, the following inequality would be quite useful
	\begin{equation}
		\norm{\Qc(t)} \le \sqrt{N}(Q_{\sigma_t(N)}(t) - Q_{\sigma_t(1)}(t)).
	\end{equation}
	This is true since $\Qc(t) = \Q(t) - \Qp(t) = \Q(t) - \frac{\sum Q_n(t)}{N}\mathbf{1} = \Q(t) - Q_{\text{avg} }(t)\mathbf{1}$, in which $Q_{\text{avg} }(t)$ is the average queue length among the $N$ servers at time slot $t$. 

	The details of the proof are presented in Section \ref{sec:proof}.
	\end{proof}

	From Eqs. \eqref{eq:th_idea} and \eqref{eq:delay_insight}, it can be seen that the important property  of a given policy is fully characterized by the inner product of $\Q_{\sigma_t}(t)$ and $\Delta(t)$ under the system state $Z(t)$, which is actually the motivation to define equivalent distribution in inner product. That is, even though the dispatching distribution $\mathbf{P}(t)$ is not unique  when there are ties in queue lengths, the inner product is actually the same if two dispatching distributions are  equivalent in inner product, hence preserving the same property in heavy-traffic regime.

	Note that class $\Pi$ is sufficient but not necessary for heavy-traffic delay optimality. Nevertheless, in the next section, we will show that it not only contains many well-known heavy-traffic delay optimal policies but also allows us to design new heavy-traffic delay optimal policies which enjoy nice features of pull-based policies.

\subsection{Important Policies in {\large $\Pi$}}

\subsubsection{Join-shortest-queue (JSQ) policy}

	Under JSQ policy, all the incoming jobs are dispatched to the queue that has the shortest queue length, ties are broken uniformly at random, out of all the servers.

	\begin{proposition}
	 	The JSQ policy belongs to $\Pi$, and hence is throughput optimal and heavy-traffic delay optimal. 
	\end{proposition} 
	The result that JSQ is throughput and heavy traffic delay optimal has been first proven via diffusion limits for two servers in \cite{foschini1978basic} and via Lyapunov-drift argument for $N$ servers in \cite{eryilmaz2012asymptotically}. Here, we present another simple proof based on our main result.

	\begin{proof}
		Note that when there are no ties in queue lengths, the dispatching distribution $\mathbf{P}(t)$ under JSQ satisfies that for all $t$
		\begin{equation}
		\label{eq:jsq}
			P_1(t) = 1 \text{ and } P_n(t) = 0, 2\le n\le N.
		\end{equation}
		In other words, all the arrivals are dispatched to the shortest queue, which is always the queue $\sigma_t(1)$ if there are no ties in queue lengths. If there are ties in queue lengths, this $\mathbf{P}(t)$ is equivalent in inner product to other dispatching distribution under the state $Z(t)$ in which ties exist.  In particular, if there are $m \le N$ queues that all have the shortest queue length, then in this case by random routing the dispatching distribution under JSQ is given by $\hat{P}_i = \frac{1}{m}$ for all $1\le i \le m$, and $\hat{P}_i = 0$ for all $i > m$.  It can be seen that $\mathbf{P}(t)$ in Eq. \eqref{eq:jsq} is equivalent in inner product to $\mathbf{\hat{P}}(t) $ according to the definition because $Q_{\sigma_t(1)} = Q_{\sigma_t(2)} = \ldots = Q_{\sigma_t(m)}$. Thus, for all $Z(t)$, under JSQ the dispatching distribution or its equivalent distribution in inner product is in the form of Eq. \eqref{eq:jsq}. Hence, we have $\Delta_1(t) = 1-\mu_{\sigma_t(1)}/\mu_\Sigma > 0$, and $\Delta_n(t) = -{\mu_{\sigma_t(n)}}/{\mu_\Sigma} < 0$ for all $2 \le n \le N$, which implies that $P_n(t)$ is a $\delta$-tilted probability with $\delta = \mu_{min}/{\mu_\Sigma}$ for all $Z(t), t\ge0$, where $\mu_{min} = \min_{n\in \mathcal{N}} \mu_n$. Therefore, the JSQ policy is contained in the class $\Pi$ under both heterogeneous and homogeneous servers.	
	\end{proof}

\subsubsection{The power-of-{\large$d$} policy}
	Under the power-of-$d$ policy, all the incoming jobs are dispatched to the queue that has the shortest queue length, ties are broken uniformly at random,  out of $d \ge 2$ servers, which are chosen uniformly at random.

	\begin{proposition}
		The power-of-$d$ policy belongs to $\Pi$ under homogeneous servers, and hence is throughput-optimal and heavy-traffic delay optimal.
	\end{proposition}
	The power-of-$d$ policy has been proven to be heavy-traffic delay optimal via Lyapunov drift condition in \cite{maguluri2014heavy}. Here, we will present another proof based on our main result.  
	\begin{proof}
		Note that when there are no ties in queue lengths, the dispatching distribution $\mathbf{P}(t)$ under the power-of-$d$ policy satisfies that for all $t \ge 0$ 
		\begin{equation}
		\label{eq:power-of-2}
			P_n(t) = {\binom{N-n}{d-1}} \bigg/{\binom{N}{d}}, 1\le n \le N-d+1,
		\end{equation}
		 and $P_n(t) = 0$, for all $n > N-d+1$. This comes from the fact that all arrivals are dispatched to the queue with shortest queue length among $d$ uniformly randomly sampled servers. Thus, if the queue $\sigma_t(n)$ is the one with shortest queue length among $d$ samples, the remaining $d-1$ samples must come from queues $\sigma_t(n+1)$, $\sigma_t(n+2)$, $\ldots \sigma_t(N)$ if all the queue lengths are different in $Z(t)$. If there are ties in queue lengths, it can be easily shown that this $\mathbf{P}(t)$ is equivalent in inner product to other dispatching distributions under any given $Z(t)$ in which there are ties in queue lengths. Thus, for all $Z(t)$, the dispatching distribution or its equivalent distribution in inner product under the power-of-$d$ policy can be fully determined by Eq.~\eqref{eq:power-of-2}. Since $P_n(t)$ is decreasing and $\mu_{\sigma_t(n)} = \mu$ under homogeneous servers, $\mathbf{P}(t)$ is a tilted distribution. Note that $\Delta_1(t) = \frac{d-1}{N}$ and $\Delta_N(t) = -\frac{1}{N}$. As a result, $\mathbf{P}(t)$ is a $\delta$-tilted distribution with $\delta = \frac{1}{N}$ for all $Z(t)$, which implies that power-of-$d$ policy is included in the class $\Pi$ for homogeneous servers.
	\end{proof}

\subsubsection{Join-idle-queue policy is not in {$\Pi$}}
\label{sec:importance_JIQ}

Now we will show that the JIQ policy is not contained in the class $\Pi$ because it is in fact not heavy-traffic delay optimal in homogeneous servers. For the heterogeneous case, it is well-known that JIQ is not heavy-traffic delay optimal since it is not even throughput optimal for a fixed number of servers \cite{stolyar2015pull}. However, for the homogeneous case, it is still open whether it is heavy-traffic optimal for a fixed number of servers, although it has been shown to be heavy-traffic optimal when the number of servers goes to infinity in the Halfin-Whitt regime \cite{mukherjee2016universality}. It turns out that when the number of servers is fixed, there exists a class of arrival process, under which the delay performance of JSQ is strictly better than that of JIQ in the heavy-traffic limit. More specifically, as shown in the proof of Theorem \ref{THM:JIQ}, for a class of arrival process, the delay under JIQ cannot achieve the common lower bound (i.e., the resource-pooled system), while JSQ can, which  implies that JIQ is not heavy-traffic delay optimal for homogeneous case.

In particular, we consider the two-server case with constant service process with rate $1$. We are able to find a class of arrival process such that Eq. \eqref{eq:heavy_traffic_def}  under JIQ does not hold. Let us first introduce the class of arrival process $\mathcal{A}$.

\begin{definition}
	An arrival process $A_\Sigma(t)$ is said to belong to $\mathcal{A}$ if 
	\begin{enumerate}[(i)]
		\item $\mathbb{P}(A_\Sigma^{(\epsilon)}(t) = 0) = p_0$, which is a constant independent of $\epsilon$.
		\item $(\sigma_\Sigma^{(\epsilon)})^2$ approaches a constant $\sigma_\Sigma^2$ which satisfies that $\sigma_\Sigma^2 > \frac{8}{p_0}-4$.
	\end{enumerate}
\end{definition}


More concretely, we are able to show the following result.

\begin{theorem}
\label{THM:JIQ}
	Consider a load balancing system with two homogeneous servers, JIQ is not heavy-traffic delay optimal in this case.
\end{theorem}
\begin{proof}
	See Section \ref{sec:proof_JIQ} of the Appendix.
\end{proof}


\subsection{Designing  New Policies in {\large$\Pi$}}
	It has been shown in the last section that the state-of-art push-based policies, e.g., JSQ and power-of-$d$, are all included in $\Pi$. Recall that, both of them need to sample the queue length information upon each new arrival, which directly results in the following two problems.

	\begin{enumerate}[(a)]
		\item The message exchange rate between dispatcher and servers is high, especially for join-shortest-queue.
		\item Each arrival has to wait for completion of the message exchange before being dispatched, which increases the actual response time for each job.
	\end{enumerate}
	To resolve the problem, the pull-based policies, join-idle-queue (JIQ) in \cite{lu2011join} and an equivalent algorithm called PULL in \cite{stolyar2015pull} are proposed, which have been shown to enjoy low message rate (at most one message per job) and have a better performance than the power-of-$d$ policy from light to moderate loads. However, as shown via numerical results in  \cite{lu2011join} and the proof of Theorem \ref{THM:JIQ} in this paper, when the load becomes high, the performance of JIQ is much worse than the power-of-$d$ policy, which motivates us 
	to {design policies that enjoy low message rates, while still guaranteeing throughput and heavy-traffic delay optimality}. 

	\begin{table*}

	\centering

	\caption{Summary of load balancing policies}\label{tab:1}


	\begin{tabular}{|c|c|c|c|c|c|c|}
	\hline
	\multirow{2}{*}{Policy}& \multirow{2}{*}{Message}&\multicolumn{2}{c|}{Throughput-Optimal} & \multicolumn{2}{c|}{Heavy-traffic Delay-Optimal} \\
	\cline{3-6}
	 & &  Homogeneous& Heterogeneous & Homogeneous& Heterogeneous\\
	\hline
	 Random& 0 & $\surd$& $\times$ &$\times$ &$\times$\\
	\hline
	 JSQ \cite{eryilmaz2012asymptotically}& 2$N$& $\surd$&$\surd$ &$\surd$ &$\surd$\\
	\hline
	 Power-of-$d$ \cite{maguluri2014heavy},~\cite{mitzenmacher2001power}& 2$d$ &$\surd$ & $\times$&$\surd$ &$\times$\\
	 \hline
	JIQ\cite{stolyar2015pull},~\cite{lu2011join} & $\le$ 1 & $\surd$& $\times$& $\times$&$\times$\\
	 \hline
	 JBT-$d$ &$\le \frac{N+2d}{T}+1$  & $\surd$&$\times$ &$\surd$&$\times$\\
	 \hline
	 JBTG-$d$ & $\le \frac{N+2d}{T}+1$ & $\surd$ &$\surd$  & $\surd$ &$\surd$ \\
	 \hline
	 \multicolumn{6}{l}{\textsuperscript{*}\footnotesize{\small The message rate for JBT-$d$ and JBTG-$d$ in this table is just a crude upper bound. When the new threshold is larger }}\\
	 \multicolumn{6}{l}{\footnotesize{ \small  than the old one, there is no need for the servers that are already recorded in memory to resend pull-messages.}}\\
	\end{tabular}
	\end{table*}

	\begin{definition}
	 Join-below-threshold-$d$ (JBT-$d$) policy is composed of three components: 

	\begin{enumerate}
		\item A threshold is updated every $T$ units of time by uniformly at random sampling $d$ servers, and taking the shortest queue length among the $d$ servers as the new threshold.
		\item Each server sends its ID to the dispatcher when its queue length is not larger than the threshold for the first time.
		\item Upon a new arrival, the dispatcher checks the available IDs in the memory. If they exist, it removes one uniformly at random, and sends all the new arrivals to the corresponding server. Otherwise, all the new arrivals will be dispatched uniformly at random to one of the servers in the system.
	\end{enumerate}	
	\end{definition}

	To be more specific, we explain the connections of the three components as follows. At the beginning of each time slot, the dispatcher immediately routes the new arrivals to a server only based on its memory state, i.e., no sampling. If there are available IDs  in memory, it removes one uniformly at random and sends the newly arrived jobs to the corresponding server. Otherwise, it sends the new jobs to a server selected uniformly at random among all the servers. At the end of each time slot, if there is no update of threshold, each server will immediately report its ID if its queue length is not larger than the threshold for the first time, i.e., only reporting once for each server before dispatched. Otherwise, the dispatcher updates the threshold by uniformly at random sampling $d$ servers, and the new threshold is set as the shortest queue length among $d$ samples. Then, each server decides to whether or not to report based on its queue length and the new threshold, using the same way as before.

 	\begin{definition}
		The JBT-$d$ policy can be easily generalized for heterogeneous servers, denote by JBTG-$d$, as follows. The only difference is that the dispatching probability distribution for the case of non-empty and empty memory is given by 
		\begin{equation*}
			\boldsymbol{\psi}_i(t) := \frac{\mu_i}{\sum_{j \in m(t)}\mu_j}{\mathbbm{1} }_{\{ i \in m(t)\}} \text{ and } \boldsymbol{\phi}_i(t) := \frac{\mu_i}{\mu_\Sigma} \text{ for all } i
		\end{equation*}
		That is, the probability to be selected for a server that has its ID in memory is weighted by its service rate. This can be easily done by requiring the server to report its service rate $\mu_n$ as well as its ID. 

	\end{definition}

	In the following, we will show that JBT-$d$ and JBTG-$d$ belong to $\Pi$, and hence throughput and heavy-traffic delay optimal. More specifically, we have the following result.

	\begin{proposition}
	\label{prop:jbt}
		For any finite $T$ and $d\ge 1$, the following two assertions are true:
		\begin{enumerate}
			\item JBT-$d$ is in $\Pi$ for homogeneous servers, and hence throughput and heavy-traffic delay optimal.
			\item JBTG-$d$ is in $\Pi$ for both homogeneous and heterogeneous servers, and hence throughput and heavy-traffic delay optimal.
		\end{enumerate}
	\end{proposition}
	\begin{proof}[Proof sketch of Proposition \ref{prop:jbt}] 
		Let us look at JBT-$d$ for some key insights behind this proof. In order to show it is in $\Pi$, we only need to show that it satisfies the two conditions (i) and (ii). For the condition (i), we will show that at any time slot $t$, the dispatching is no worse than the random routing. For the condition (ii), we will show that at time slots $rT+1$, $r \in \{0,1,2,\ldots\}$, the dispatching decision is strictly better than the random routing.

	Note that under the JBT-$d$ policy, if the ID of the server $\sigma_t(n+1)$ is in $m(t)$, we must have that the ID of the server $\sigma_t(n)$ is also in $m(t)$ as the permutation is in the non-decreasing order of the queue length. Denote by $\tilde{p}_k(t)$ the probability that there are $k$ IDs in the memory $m(t)$ for time $t$, i.e., $\tilde{p}_k(t) = \Pr(|m(t)| = k)$. Then, the probability for the server $\sigma_t(n)$ to be selected at time $t$, i.e., $P_n(t)$ is given by
		\begin{equation}
		\label{eq:important1}
			P_n(t) = \sum_{i=n}^N \tilde{p}_i(t)\frac{1}{i}.
		\end{equation}
		This is true since for the server $\sigma_t(n)$ to be selected, there should be at least $n$ IDs in memory, i.e., $|m(t)| \ge n$ and in each case the probability for the server $\sigma_t(n)$ to be chosen is $\frac{1}{|m(t)|}$. Therefore, we can see that the probability of $P_n(t)$ satisfies 
		\begin{equation}
		\label{eq:Pn}
			P_1(t) \ge P_2(t) \ge \ldots \ge P_N(t),
		\end{equation}
		which directly implies that for all $t\ge0$ there exists a $k$ between $2$ and $N$ such that $\Delta_n(t) = P_n(t) - \frac{1}{N} \ge 0$ for all $n<k$ and  $\Delta_n(t) \le 0$ for all $n\ge k$. Therefore, condition (i) of $\Pi$ is satisfied.

		For condition (ii), we will show that there exists a lower bound for $\delta$ such that $\mathbf{P}(rT+1)$ (or an inner product equivalent distribution when there are ties in queue lengths), $r \in \{0,1,2,\ldots\}$ is at least a $\delta$-tilted distribution. In this case, we need only to show that $P_N(rT+1)$ is strictly less than $\frac{1}{N}$ for all the system state $Z(rT+1)$. 

		The full proof is presented in Section \ref{sec:52}.
	\end{proof}
	\subsection{Features of JBT-d}
This section summarizes the  main features of JBT-d policy and compares it with existing policies in Table \ref{tab:1}. In particular, we compare the number of messages for each new arrival under different policies. For push-based polices, e.g., JSQ and power-of-$d$, there are $d$ query and $d$ response messages for each new arrival ($d=N$ for JSQ policy). For JIQ policy, for each new arrival, it requires at most one pull-message since when there are no pull-messages in memory, the arrival is dispatched randomly without costing any pull-message. Similarly, our JBT-$d$ policy requires $2d$ push-messages every $T$ time slots to update the threshold. Due to the threshold update, the old pull-messages may be discarded, which is upper bounded by $N$. Hence, the pull-message for each new arrival under JBT-$d$ is at most $1+\frac{2d+ N}{T}$.

In sum, the JBT-$d$ policy has the following nice features: a) It is throughput and heavy-traffic delay optimal. b) It is able to guarantee heavy-traffic delay optimal with very low complexity when $T$ is relatively large. c) The computation overhead is small. It only needs to keep a list of the available IDs and choose randomly. d) The arrival is immediately dispatched, i.e., there is no dispatching delay as compared to push-based policies such as JSQ and Power-of-$d$.



It is worth pointing that by just changing the way of updating the threshold in JBT-d, we can design other new policies which also enjoy the nice features above. For example, it can be easily shown via similar arguments that if the threshold is updated by sampling all the servers and taking the average value of the queue length as the new threshold, this corresponding new policy is still in the class $\Pi$.
\section{Numerical Results}
\label{sec:simulation}
	In this section, we use simulations to compare our proposed policies JBT-$d$ and JBTG-$d$ with join-shortest-queue (JSQ), join-idle-queue (JIQ), power-of-$d$ (SQ($d$)) and power-of-$d$ with memory (SQ($d$,$m$)). The power-of-$d$ with memory policy (SQ($d$,$m$)) improves power-of-$d$ by using extra memory to store the $m$ shortest queues sampled at the previous time slot \cite{mitzenmacher2002load}. 
	
	We compare the throughput performance, delay performance, heavy-traffic delay performance and message overhead performance under various arrival and service processes as well as different system sizes. 
	Moreover, the $95\%$ confidence intervals for all the simulation results can be found in Section \ref{sec:CI} of the Appendix, which justify the accuracy of the simulation results. Some of the confidence intervals are also included in figures and similar accuracy goes for other points as well.
	The exogenous arrival $A_{\Sigma}(t)$ and potential service $S_n(t)$ are drawn from a Poisson distribution with rate $\lambda_{\Sigma}$ and $\mu_n$ for each time slot unless otherwise specified.
	The traffic load is equal to $\rho = \lambda_{\Sigma}/\mu_{\Sigma}$. The parameter $T$ is the threshold update interval for JBT-$d$ and JBTG-$d$. 


	Before we dive into each case, let us first summarize some key observations in these simulation results.
	\begin{enumerate}[(i)]
		\item \textbf{Throughput performance:}
		\begin{enumerate}[(a)]
			\item Our proposed policy JBT-$d$ continues to stabilize all the considered loads in heterogeneous systems under all the different settings.
			\item JIQ and SQ($d$) cannot stabilize the system when the load is high in all the cases.
			\item JIQ appears to have a larger capacity region as the number of servers increases. This agrees with the theoretical result in \cite{stolyar2015pull}.
		\end{enumerate}
		\item \textbf{Delay performance:}
		\begin{enumerate}[(a)]
			\item Our proposed policy JBT-$d$ continues to have good performance from light to heavy traffic among all the cases.
			\item As the system size increases, JBT-$d$ achieves the same performance as JSQ for a larger range of loads. Meanwhile, the gains of JBT-$d$ over SQ($d$) and SQ($d$,$m$) become larger as the system size increases.
			\item The gain of JBT-$d$ over JIQ decreases as the system size increase. This is also intuitive since as $N$ goes to infinity, it is more likely to find an idle server, which results in the fact that JIQ is heavy-traffic delay optimal in the Halfin-Whitt regime \cite{mukherjee2016universality}.
			\item The gain of JBT-$d$ over JIQ increases when burstness is introduced in arrival or service process. This agrees with the insight in the proof of Theorem \ref{THM:JIQ} that larger variance of arrival or service process will degrade the performance of JIQ.
		\end{enumerate}

		\item \textbf{Message overhead performance:}
		\begin{enumerate}[(a)]
			\item Our proposed policy JBT-$d$ continues to have a low message overhead among all the cases.
			\item Push-based policies such as SQ($d$) and SQ($d$,$m$) have to increase their message overhead linearly with respect to $d$ to achieve good delay performance as the system size increases. In contrast, our proposed JBT-$d$ is able to achieve good performance with a message rate that is less than $1$ for all the cases when $T$ is large.
		\end{enumerate}

		\item \textbf{Confidence interval:}
		\begin{enumerate}[(a)]
			\item Our proposed policy JBT-$d$ continues to have good $95\%$ confidence intervals in all the various settings.
		\end{enumerate}
	\end{enumerate}


\subsection{Throughput Performance}

	We investigate the throughput region of different load balancing policies in the case of heterogeneous servers. In particular, we consider the case that the system consisting of two server pools each with five servers and the rates are $1$ and $10$, respectively. A turning point in the curve indicates that the load approaches the throughput region boundary of the corresponding policy. 

	\begin{figure}[t]
	\graphicspath{{./Figures/}}
		\centering
		\includegraphics[width=3.2in]{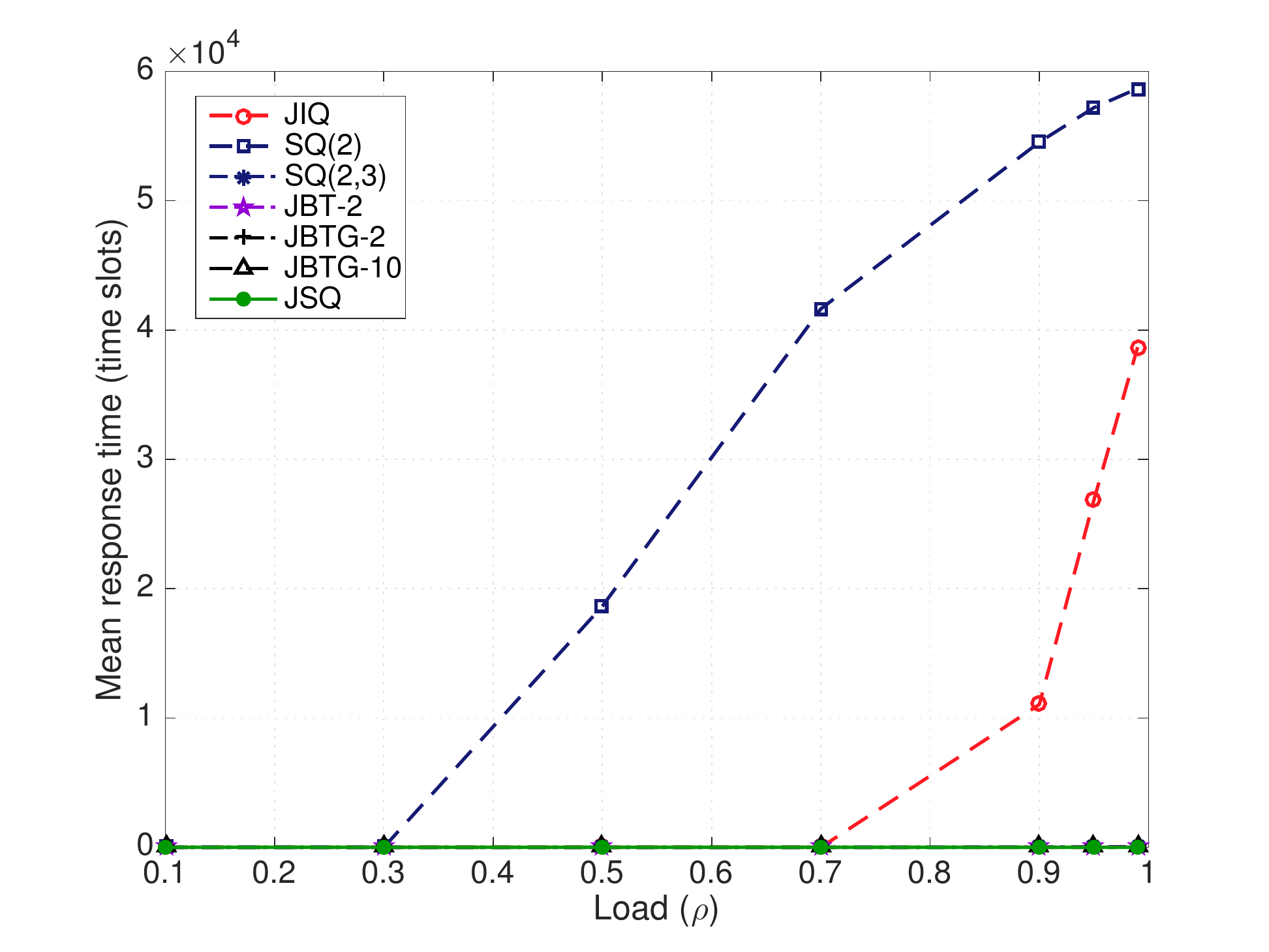}
		\caption{Throughput performance in $10$ heterogeneous servers.}\label{fig:throughput_N10}
	\end{figure}

	Figure \ref{fig:throughput_N10} shows that the system becomes unstable when $\rho > 0.5$ under the policy power-of-2 (SQ($2$)), and it becomes unstable under JIQ when $\rho > 0.9$. In contrast, our proposed JBTG-$d$ policy remains stable for all the considered loads which agrees with the theoretical results. It can be seen that JBT-$2$ is also able to stabilize the system for all the considered loads in this case. Note that the system remains stable under the power-of-$2$ with memory  policy SQ($2$,$3$), which demonstrates the benefit of using memory to obtain maximum throughput as first discussed in \cite{shah2002use}.

	We further provide additional simulation results on throughput performance under different arrival and service process as well as different system sizes in  Section \ref{sec:addition} of the Appendix. 

\subsection{Delay Performance}
	We investigate the mean response time under different load balancing policies in homogeneous servers with different system sizes and various arrival and service processes.  The time interval for threshold update of  JBT-$d$ is set $T = 1000$.

	Let us first look at the regime when $\rho$ is from $0.3$ to $0.99$, which ranges from light traffic to heavy traffic. Figure \ref{fig:delay_N10} shows that our proposed policy JBT-$d$ outperforms both power-of-$2$ and power-of-$2$ with memory (SQ($2$,$3$), which uses the same amount of memory as in JBT-$d$) for nearly the whole regime. Moreover, JBT-$d$ policy achieves nearly the same response time of JIQ when the load is not too high. However, as the load becomes heavier, the performance of JIQ gets worse and worse, and its mean response time is as large as two times of the response time under JBT-$d$ policy when the load is $0.99$.

\begin{figure}[t]
\graphicspath{{./Figures/}}
\begin{minipage}{3.2in}
\begin{center}
\includegraphics[width=3.2in]{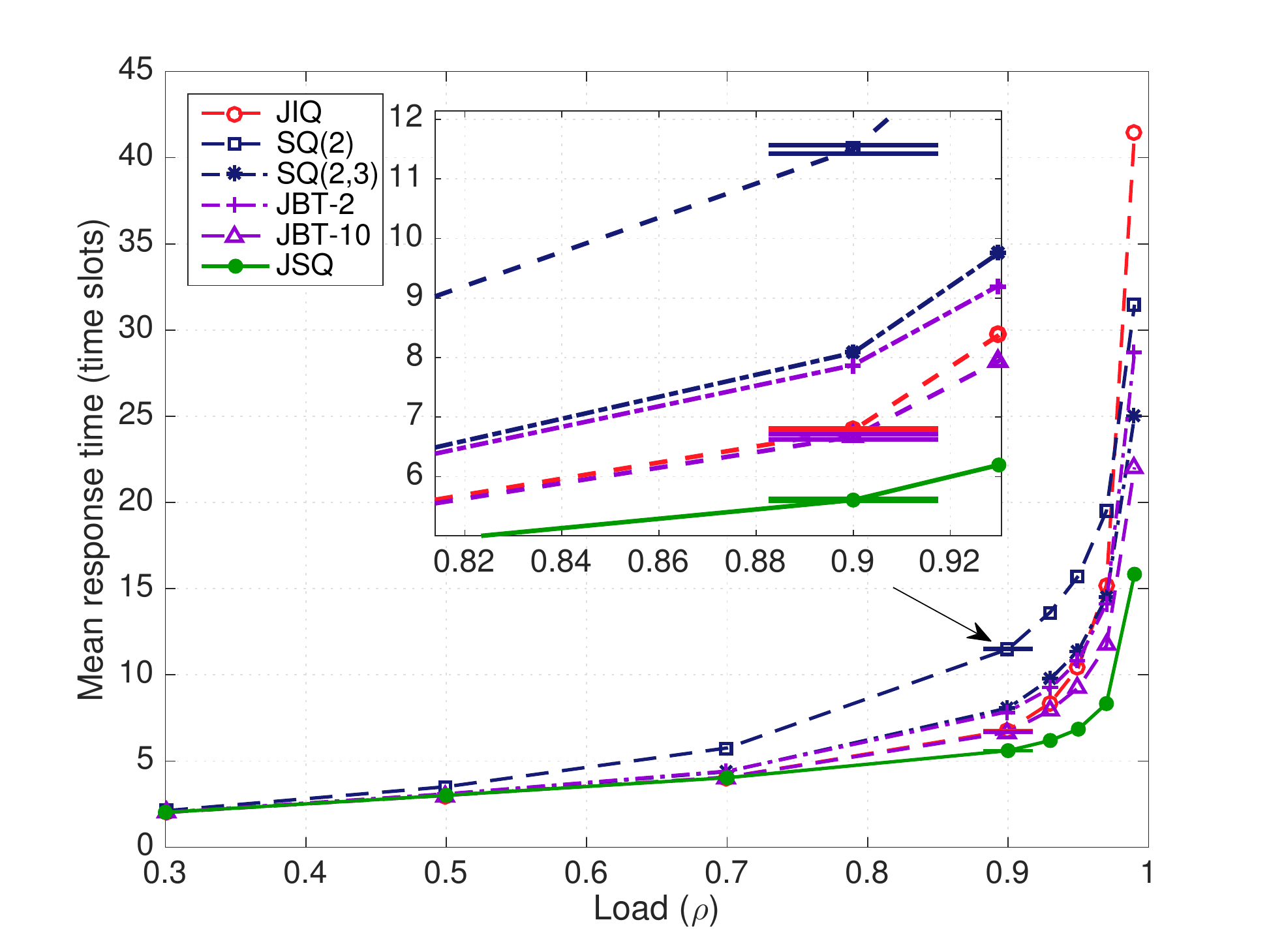}
\caption{Delay performance under $10$ homogeneous servers.}
\label{fig:delay_N10}
 \end{center}
\end{minipage}
\hfill
\begin{minipage}{3.2in}
\begin{center}
\includegraphics[width=3.2in]{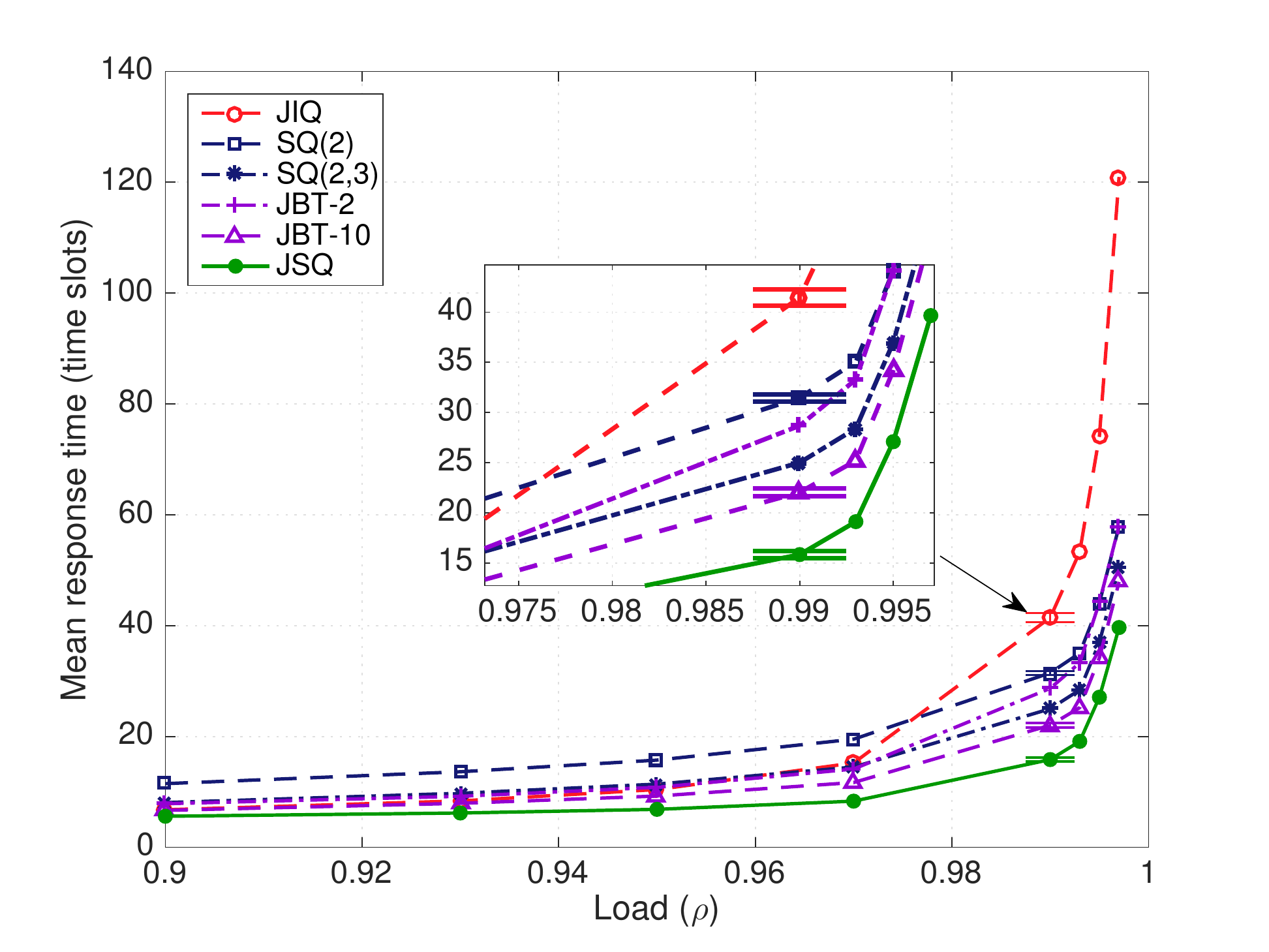}
\caption{Heavy-traffic delay performance under $10$ homogeneous servers.}
\label{fig:heavydelay_N10}
 \end{center}
\end{minipage}
\end{figure}

\begin{figure}[t]
\graphicspath{{./Figures/}}
\begin{minipage}{3.2in}
\begin{center}
\includegraphics[width=3.2in]{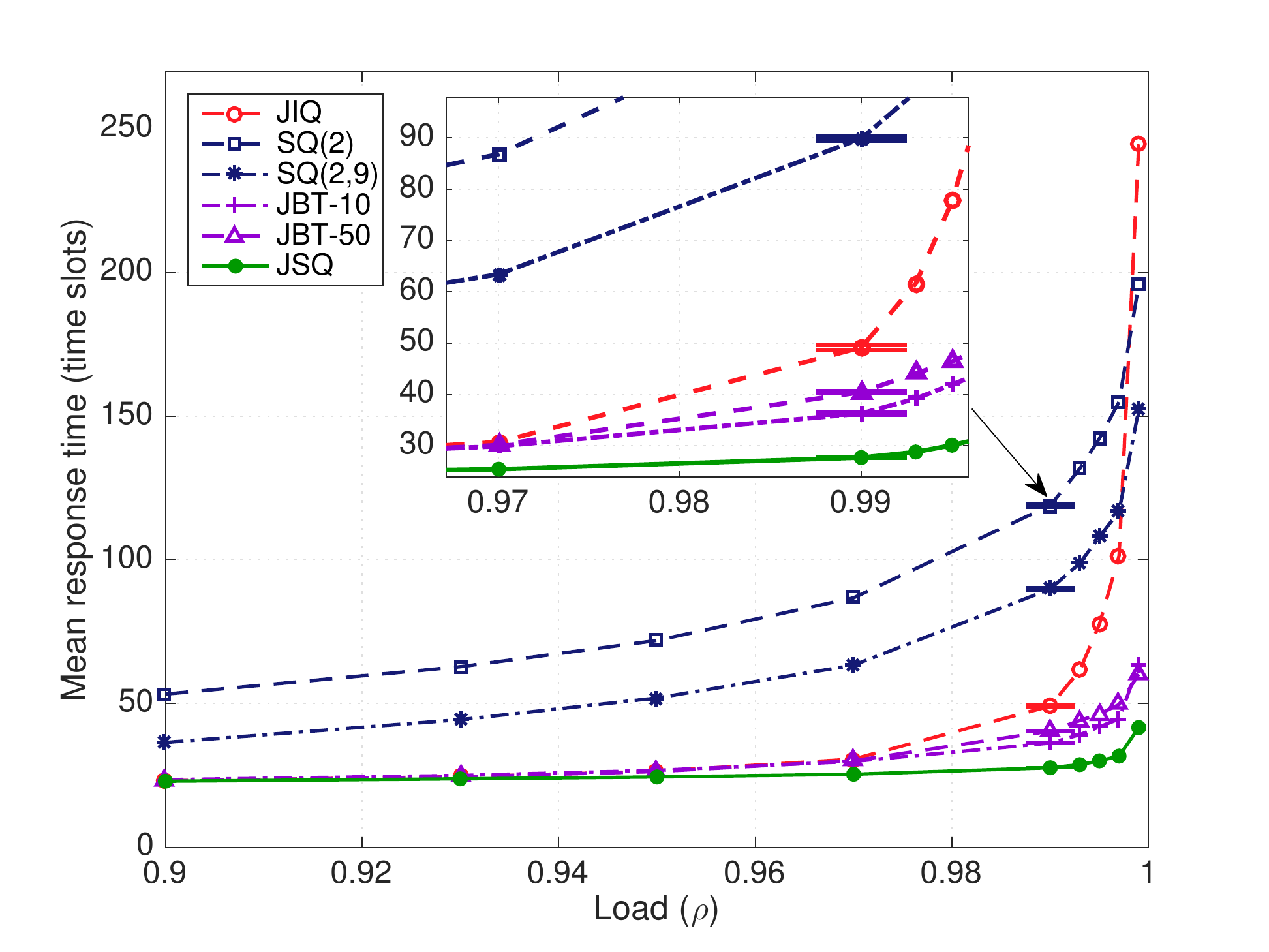}
\caption{Heavy-traffic delay performance under $50$ homogeneous servers.}
\label{fig:heavydelay_N50_sim}
 \end{center}
\end{minipage}
\hfill
\begin{minipage}{3.2in}
\begin{center}
\includegraphics[width=3.2in]{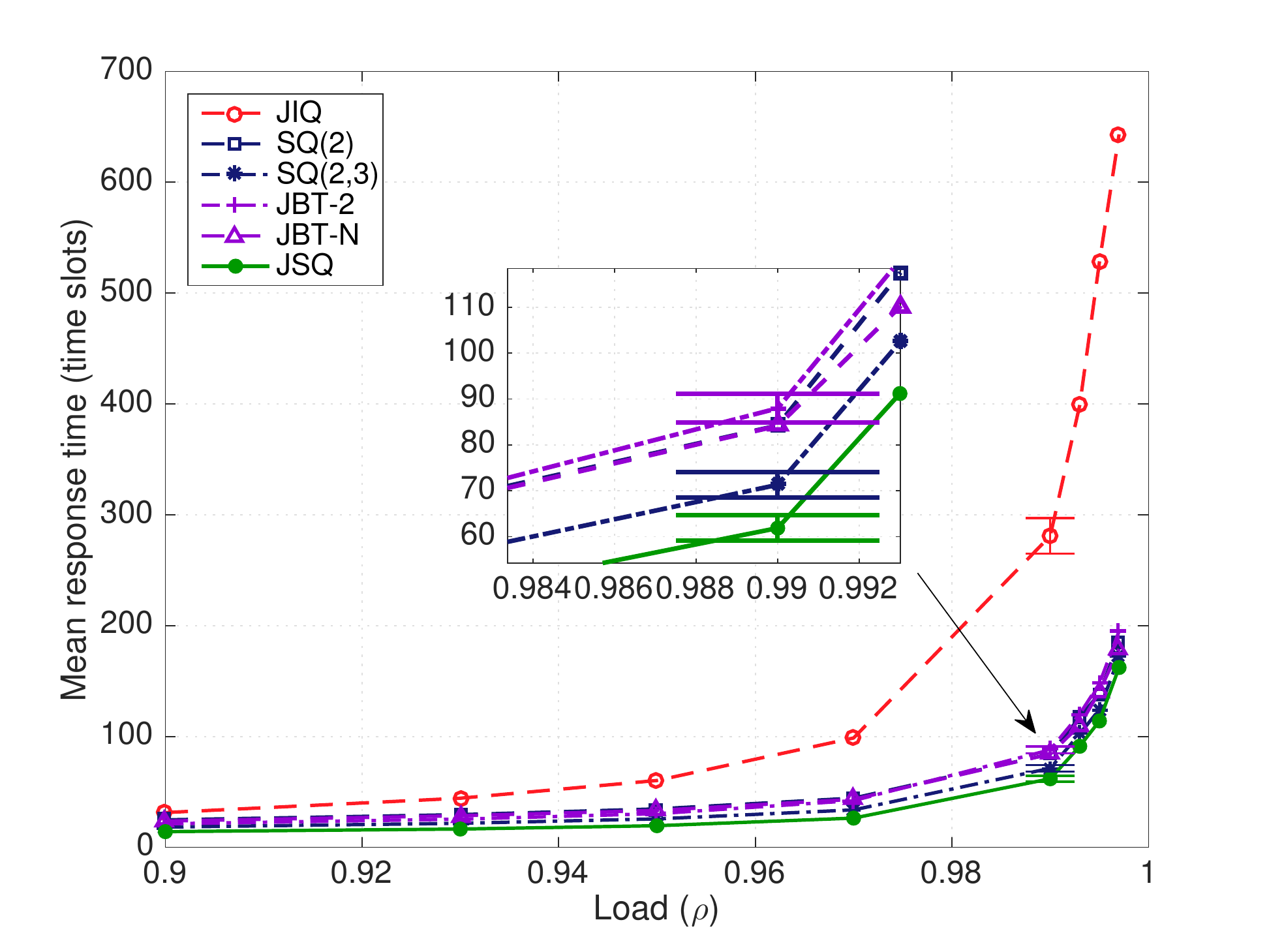}
\caption{Heavy-traffic delay performance under $10$ homogeneous servers with Poisson arrival and bursty service.}
\label{fig:heavydelay_poiss_burst_sim}
 \end{center}
\end{minipage}
\end{figure}





	Now, let us get a closer look at the  delay performance in heavy-traffic regime, i.e., $\rho > 0.9$, as shown in Figure \ref{fig:heavydelay_N10}. It can be seen that JBT-$10$ outperforms JIQ when $\rho > 0.9$ and JBT-$2$ outperforms JIQ when $\rho > 0.95$ in this case. More importantly, the gap between them keeps increasing as the load gets higher. Note that power-of-$2$ with memory (SQ($2$,$3$)) also has good performance in this case, which, however, uses a much higher message rate compared to our JBT-$d$ policy, as discussed in the next section.

	Last, we further provide some results on heavy-traffic delay performance for a larger system size and a bursty service process, respectively. Due to space limitation, the comprehensive results can be found in Section \ref{sec:addition} of Appendix. 
	Figure \ref{fig:heavydelay_N50_sim} illustrates the heavy-traffic performance under Poisson arrival and Poisson service when $N=50$. In this case, first thing to note is that even though the power-of-$d$ with memory policy (SQ($2$,$9$)) uses the same amount of memory as in JBT-$d$, it has a much poorer performance with a much higher message overhead since the message overhead of JBT-$d$ is strictly less than $1$ when $T=1000$ in this case. This means that to improve delay performance in large system size, power-of-$d$ with memory has to increase its message overhead linearly with respect $d$, while our JBT-$d$ policy is able to achieve good performance with message rate less than $1$ even for $d=N$. Moreover, as $\rho$ approaches to $1$, the performance of JIQ degrades substantially while our proposed JBT-$d$ remains quite close to JSQ. 
	In Figure \ref{fig:heavydelay_poiss_burst_sim}, the potential number of jobs served in each time slot is either $0$ or $10$. In this bursty service case, JIQ degrades much faster than that in the Poisson service process. Moreover, in this setting we can easily observe the difference between non-heavy-traffic policy (JIQ) and heavy-traffic optimal policies (all the others). 
	Note that the message overhead of SQ($2$,$3$) is nearly $8$ times as large as that of JBT-$d$, as shown in next section, though its delay is slightly better than JBT-$d$.

\begin{figure}[t]
\graphicspath{{./Figures/}}
\begin{minipage}{3.2in}
\begin{center}
\includegraphics[width=3.2in]{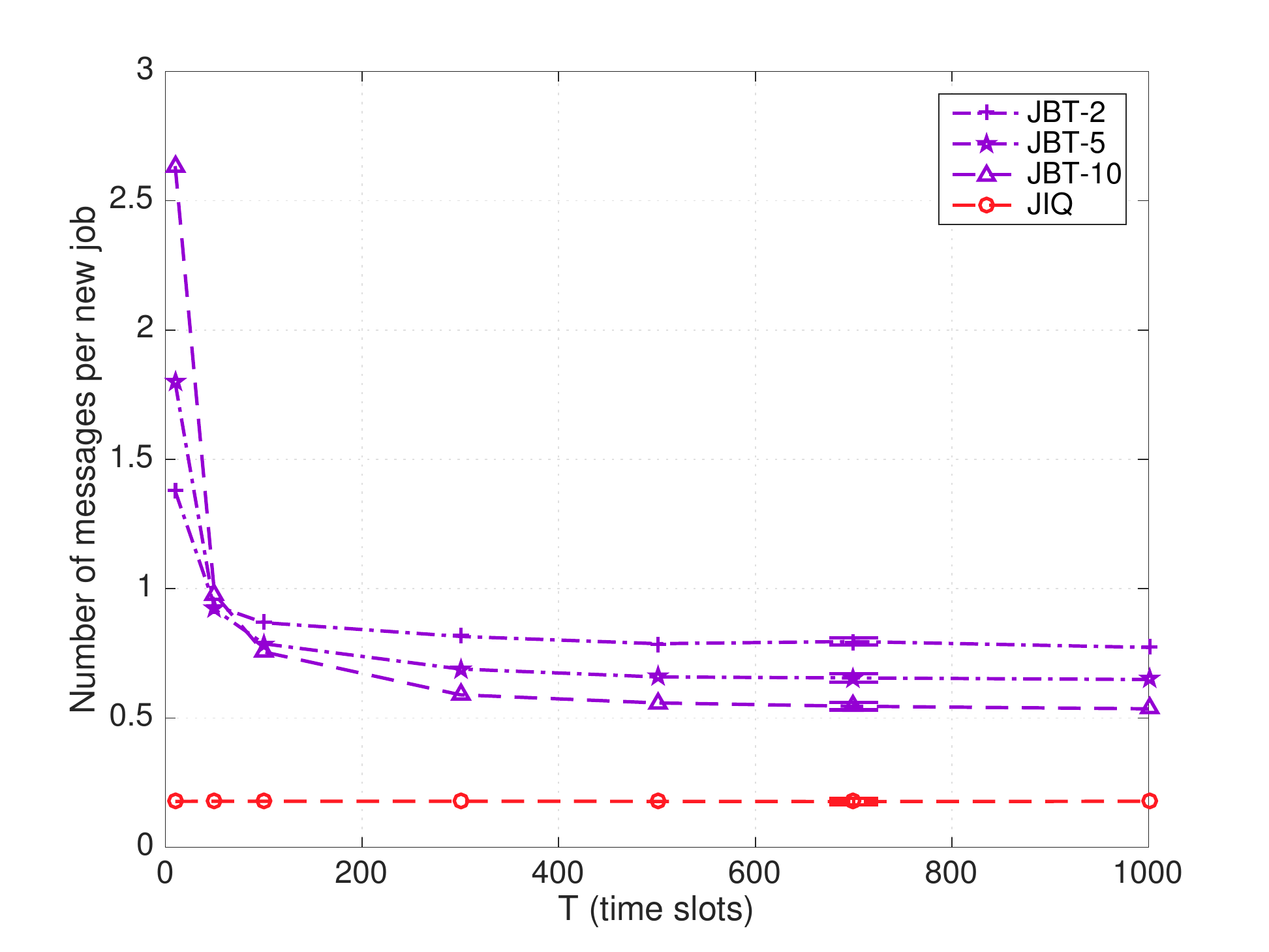}
\caption{Message per new job arrival under $10$ homogeneous servers with respect to $T$.}
\label{fig:mess_N10_Td}
 \end{center}
\end{minipage}
\hfill
\begin{minipage}{3.2in}
\begin{center}
\includegraphics[width=3.2in]{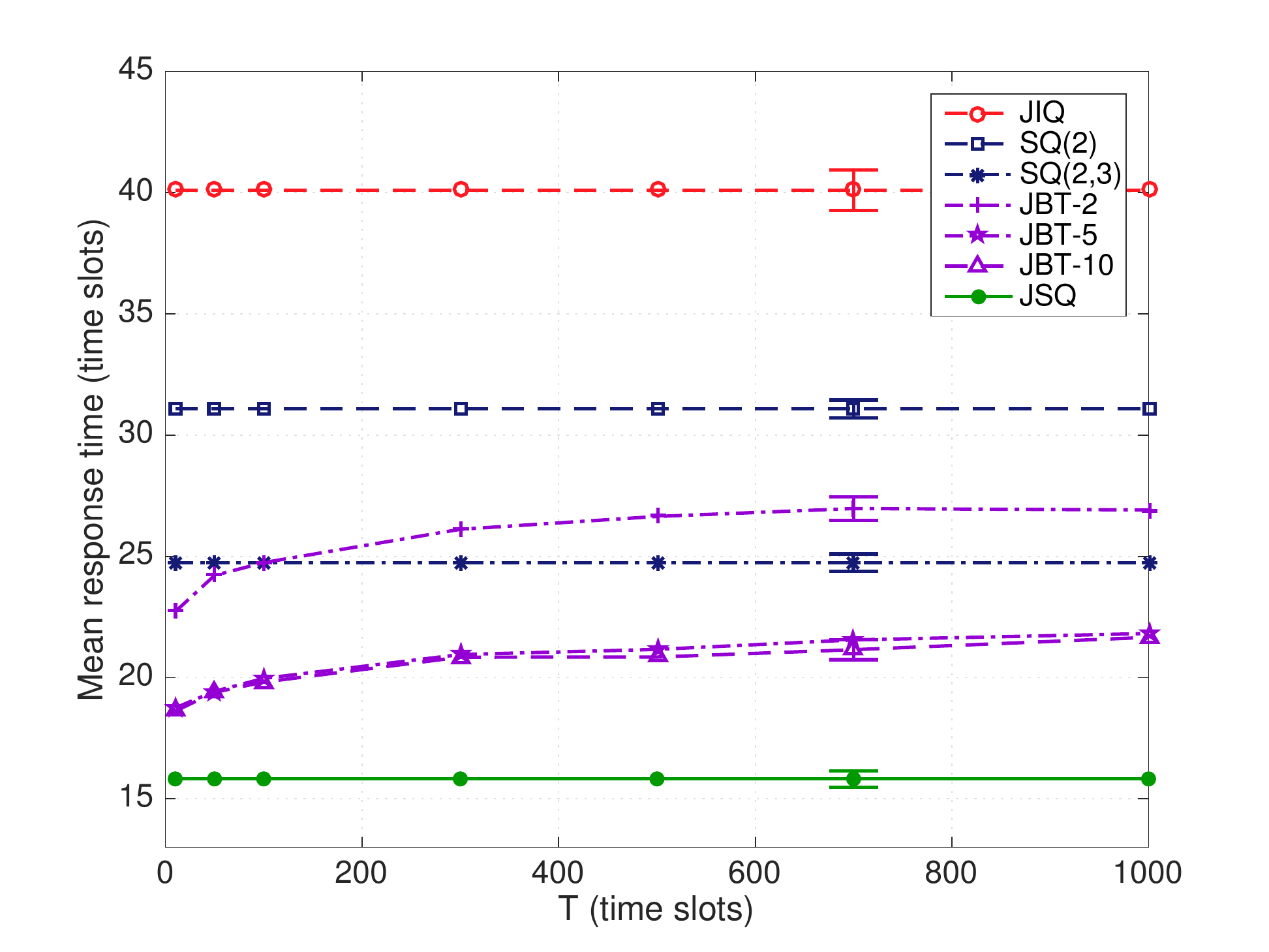}
\caption{Delay performance under $10$ homogeneous servers with respect to $T$.}
\label{fig:delay_N10_Td}
 \end{center}
\end{minipage}
\end{figure}



\subsection{Message Overhead}
	We use simulations to further show the low message rate of our proposed JBT-$d$ policy, though a crude upper bound has been established. Here, we consider the $10$ homogeneous servers  with Poisson arrival and Poisson service, and more results for different settings can be found in Section \ref{sec:addition} of Appendix. More specifically, we investigate the impact of different values of $T$, i.e., the time interval for updating the threshold, on the message rate and its corresponding delay performance at a fixed load $\rho = 0.99$. In particular, we calculate the average number of messages per new job arrival under each policy. For push-based policies, e.g., JSQ, power-of-$2$ (SQ(2)) and power-of-$2$ with memory (SQ($2$,$3$)), the message only includes the  push-message and is easily calculated as $20$, $4$, and $4$, which is independent of $T$. For JIQ, we know that the rate is at most one for each new job arrival, which is also independent with $T$ and serves as the benchmark. 

	Figure \ref{fig:mess_N10_Td} shows the message rate of JBT-$d$ with respect to $T$ for different values of $d$, and the corresponding delay performance is shown in Figure \ref{fig:delay_N10_Td}. The first thing to note is that the message rate of JIQ is much smaller than one since the traffic is heavy and hence there are few idle servers in this case, which directly results in the poor performance in the heavy-traffic regime. Second, the message rate of JBT-$d$ is smaller than all push-based policies and becomes less than one when $T>100$ in this case, which means that it is able to achieve throughput and heavy-traffic delay optimality by requiring a slightly more message than JIQ. Moreover, it can be seen that as $T$ increases, there is no significant change of the delay performance, which indicates that we are allowed to adopt a sufficiently large $T$ while not incurring the loss of performance very much in this case. Last, it is worth noting that a larger $d$ does not necessarily mean a larger message overhead when $T$ is large. This is because when $T$ is large, the push-message in JBT-$d$ will be dominated by the pull-message. For a small $d$, the number of pull-message may be larger since the threshold may be higher than that under a larger $d$. As shown in the additional results in Section \ref{sec:addition} of Appendix, the observations above hold almost for all the considered cases.
	The exact impact and relationship of $T$ and $d$ would be one of our future research focus.

\section{Proof of Main Results}
\label{sec:proofMain}
The high-level insight for class $\Pi$ to be heavy-traffic delay optimal is that it always has a preference to shorter queues in the way that is specified by the $\delta$-tilted distribution. The key step behind the proof that JBT-$d$ is heavy-traffic delay optimal is to show that the dispatching distribution for the time slot that is immediately after the threshold update is always a $\delta$-titled distribution. 
\subsection{Proof of Theorem \ref{THM:ACLASS}}
\label{sec:proof}
	Before we adopt the sufficient conditions in Lemma \ref{LEMMA:THTROUGHPUT} and Lemma \ref{LEMMA:DELAY} to prove Theorem \ref{THM:ACLASS}, we first present the following lemmas on the tilted distribution and $\delta$-tilted distribution, respectively. 
	\begin{lemma}
	\label{lem:tilted}
		For a system with mean arrival rate $\lambda_\Sigma = \mu_\Sigma -\epsilon$ and a tilted distribution $\mathbf{P}(t)$ under $Z(t)$, we have 
		\begin{equation}
		\label{eq:tilted_Q}
			\ex{\inner{\Q(t)}{\A(t)-\s(t)}\mid Z(t)} \le -\epsilon\frac{\mu_{min}}{\mu_\Sigma} \norm{\Q(t)}
		\end{equation}
		and 
		\begin{equation}
		\label{eq:tilted_Qc}
			\ex{\inner{\Qc(t)}{\A(t)-\s(t)}\mid Z(t)} \le \epsilon \sqrt{N}\norm{\Qc(t)}
		\end{equation}
	\end{lemma}
	\begin{proof}
		Consider the left-hand-side (LHS) of Eq. \eqref{eq:tilted_Q} 
		\begin{equation*}
		\begin{split}
			&\ex{\inner{\Q(t)}{\A(t)-\s(t)}\mid Z(t)}\\
			= &\sum_{n=1}^N  Q_{\sigma_t(n)}(t)\left[(\Delta_n(t)+\frac{\mu_{\sigma_t(n)}}{\mu_\Sigma})\lambda_\Sigma - \mu_{\sigma_t(n)} \right]\\
			 \ep{a} &\sum_{n=1}^N Q_{\sigma_t(n)}(t) \Delta_n(t)\lambda_\Sigma + \sum_n^N Q_{\sigma_t(n)}(t)(-\epsilon\frac{\mu_{\sigma_t(n)}}{\mu_\Sigma})\\
			\lep{b} &\sum_{n=1}^N Q_{\sigma_t(n)}(t)(-\epsilon\frac{\mu_{\sigma_t(n)}}{\mu_\Sigma})\\
			\lep{c}& -\epsilon\frac{\mu_{min}}{\mu_\Sigma} \norm{\Q(t)},
		\end{split}	
		\end{equation*}
		in which equation (a) holds since $\lambda_\Sigma = \mu_\Sigma - \epsilon$; inequality (b) comes from the fact that $\sum_{n=1}^N Q_{\sigma_t(n)}(t) \Delta_n(t) \le 0$ under a tilted distribution. This fact is true since $Q_{\sigma_t(1)}(t)\le Q_{\sigma_t(2)}(t)\le \ldots \le Q_{\sigma_t(N)}(t)$ and $\sum_{n=1}^N \Delta_n(t) = 0$; inequality (c) follows from the fact that $\norm{\mathbf{x}}_1 \ge \norm{\mathbf{x}}$ for any $\mathbf{x} \in \mathbb{R}^N$.

		Note that $\Qc(t) = \Q(t) - \Qp(t) = \Q(t) - \frac{\sum Q_n(t)}{N}\mathbf{1} = \Q(t) - Q_{\text{avg} }(t)\mathbf{1}$, in which $Q_{\text{avg} }(t)$ is the average queue length among the $N$ servers at time slot $t$. Then, consider the left-hand-side (LHS) of Eq. \eqref{eq:tilted_Qc}
		\begin{equation}
		\label{eq:tilted_proof}
		\begin{split}
			&\ex{\inner{\Qc(t)}{\A(t)-\s(t)}\mid Z(t)}\\
			 = &\sum_{n=1}^N  (Q_{\sigma_t(n)}(t) - Q_{avg}(t))\left[(\Delta_n(t)+\frac{\mu_{\sigma_t(n)}}{\mu_\Sigma})\lambda_\Sigma - \mu_{\sigma_t(n)} \right]\\
			\ep{a} & \sum_{n=1}^N Q_{\sigma_t(n)}(t) \Delta_n(t)\lambda_\Sigma +\sum_n^N (Q_{\sigma_t(n)}(t) - Q_{avg}(t))(-\epsilon\frac{\mu_{\sigma_t(n)}}{\mu_\Sigma})\\
			\lep{b} &\sum_{n=1}^N (Q_{\sigma_t(n)}(t) - Q_{avg}(t))(-\epsilon\frac{\mu_{\sigma_t(n)}}{\mu_\Sigma})\\
			\lep{c} &\epsilon \sum_{n=1}^N |(Q_{\sigma_t(n)}(t) - Q_{avg}(t))|\\
			\lep{d} &\epsilon \sqrt{N} \norm{\Qc(t)},
		\end{split}
		\end{equation}
		where equation (a) comes from the facts that $\sum_{n=1}^N \Delta_n(t) = 0$ and $\lambda_\Sigma = \mu_\Sigma -\epsilon$; inequality (b) holds since $\sum_{n=1}^N Q_{\sigma_t(n)}(t) \Delta_n(t) \le 0$ under a tilted distribution; inequality (c) is true since $x \le |x|$ for any $x \in \mathbb{R}$ and $|\epsilon\frac{\mu_{\sigma_t(n)}}{\mu_\Sigma}| \le \epsilon$ for all $n \in \mathcal{N}$; inequality (d) is true since $\norm{\mathbf{x} }_1 \le \sqrt{N} \norm{\mathbf{x} }$ for any $\mathbf{x} \in \mathbb{R}^N$.
	\end{proof}

	\begin{lemma}
	\label{lem:propertilted}
		For a system with mean arrival rate $\lambda_\Sigma = \mu_\Sigma -\epsilon$ and a $\delta$-tilted distribution $\mathbf{P}(t)$ under $Z(t)$, we have 
		\begin{equation}
		\label{eq:proper_tiltied}
			\ex{\inner{\Qc(t)}{\A(t)-\s(t)}\mid Z(t)} \le \sqrt{N} \norm{\Qc(t)}(\epsilon - \frac{\delta \lambda_\Sigma}{N})
		\end{equation}
	\end{lemma}
	\begin{proof}
		Consider the left-hand-side (LHS) of Eq. \eqref{eq:proper_tiltied}, we have 
		\begin{equation*}
		\begin{split}
			&\ex{\inner{\Qc(t)}{\A(t)-\s(t)}\mid Z(t)}\\
			= &\sum_{n=1}^N  (Q_{\sigma_t(n)}(t) - Q_{avg}(t))\left[(\Delta_n(t)+\frac{\mu_{\sigma_t(n)}}{\mu_\Sigma})\lambda_\Sigma - \mu_{\sigma_t(n)} \right]\\
			\ep{a} &\sum_{n=1}^N Q_{\sigma_t(n)}(t) \Delta_n(t)\lambda_\Sigma +\sum_{n=1}^N (Q_{\sigma_t(n)}(t) - Q_{avg}(t))(-\epsilon\frac{\mu_{\sigma_t(n)}}{\mu_\Sigma})\\
			\lep{b} &\sum_{n=1}^N{} Q_{\sigma_t(n)}(t) \Delta_n(t)\lambda_\Sigma + \epsilon \sqrt{N} \norm{\Qc(t)}\\
			\lep{c}& -\lambda_\Sigma \delta(Q_{\sigma_t(N)}(t) - Q_{\sigma_t(1)}(t)) + \epsilon \sqrt{N} \norm{\Qc(t)}\\
			\lep{d}& -\lambda_\Sigma \frac{\delta}{\sqrt{N}} \norm{\Qc(t)} + \epsilon \sqrt{N} \norm{\Qc(t)}\\
			= &\sqrt{N} \norm{\Qc(t)}(\epsilon - \frac{\delta \lambda_\Sigma}{N})
		\end{split}
		\end{equation*}
		where equation (a) holds since $\sum_{n=1}^N \Delta_n(t) = 0$ and $\lambda_\Sigma = \mu_\Sigma -\epsilon$; inequality (b) follows from steps (c) and (d) in Eq. \eqref{eq:tilted_proof}; inequality (c) follows from the definition of $\delta$-tilted probability and the fact that $Q_{\sigma_t(1)}(t)\le Q_{\sigma_t(2)}(t)\le \ldots \le Q_{\sigma_t(N)}(t)$; inequality (d) is follows from the fact that $\norm{\Qc(t)} \le \sqrt{N}(Q_{\sigma_t(N)}(t) - Q_{\sigma_t(1)}(t))$.
	\end{proof}

	Now we are ready to present the proof of Theorem \ref{THM:ACLASS}
	
	\textbf{Proof of Theorem \ref{THM:ACLASS}:}
	The proof is a direct application of the sufficient conditions for throughput and heavy-traffic delay optimality, i.e., we need only to show Eq. \eqref{eq:condition_throughput} and Eq. \eqref{eq:cond_delay} hold. 

	Fix a load balancing policy $p$ in $\Pi$.
	Let us first consider the left-hand-side (LHS) of Eq. \eqref{eq:condition_throughput} with $T_1 = T$,
	\begin{equation*}
	\begin{split}
		LHS &\ep{a} {\sum_{t=t_0}^{t_0+T-1} \ex{\inner{\Q(t)}{\A(t) -\s(t)} \mid Z(t_0)  = Z}}\\
		& \ep{b} {\sum_{t=t_0}^{t_0+T-1} \ex{\ex{\inner{\Q(t)}{\A(t) -\s(t)} \mid Z(t) }| Z(t_0) = Z}}\\
		& \lep{c} {\sum_{t=t_0}^{t_0+T-1} \ex{-\epsilon\frac{\mu_{min}}{\mu_\Sigma} \norm{\Q(t)}\mid Z(t_0) = Z}}\\
		& \le -\epsilon\frac{\mu_{min}}{{}\mu_\Sigma} \norm{\Q(t_0)}
	\end{split}
	\end{equation*}
	where equation (a) comes from the linearity of condition expectation; equation (b) follows from the tower property of conditional expectation and the fact that $\Q(t)$, $\A(t)$ and $\s(t)$ are conditionally independent of $Z(t_0)$ when given $Z(t)$. inequality (c) follows from Lemma \ref{lem:tilted} since the policy $p$ adopts a tilted distribution within every time slot for all $Z(t)$. Hence, the condition of Lemma \ref{LEMMA:THTROUGHPUT} is satisfied and hence policy $p$ is throughput optimal.

	Let us now turn to consider the left-hand-side (LHS) of Eq. \eqref{eq:cond_delay} with $T_2 = T$ and $\epsilon < \epsilon_0 = \frac{\delta\mu_\Sigma}{2TN+2\delta}$.
	\begin{equation*}
	\begin{split}
		LHS \ep{a} &{\sum_{t=t_0}^{t_0+T-1} \ex{\inner{\Qc(t)}{\A(t) -\s(t)} \mid Z(t_0)  = Z}}\\
		\ep{b} & {\sum_{t=t_0}^{t_0+T-1} \ex{\ex{\inner{\Qc(t)}{\A(t) -\s(t)} \mid Z(t) }\mid Z(t_0) = Z}}\\
		\lep{c} &\sum_{t \neq t^*} \ex{ \epsilon\sqrt{N}\norm{\Qc(t)} \mid Z(t_0) = Z} \\
		& +\ex{\sqrt{N} \norm{\Qc(t^*)}(\epsilon - \frac{\delta \lambda_\Sigma}{N}) \mid Z(t_0) = Z}\\
		\lep{d} & (T-1)\epsilon\sqrt{N}(\norm{\Qc(t_0)} +M) \\
		& +  (\epsilon - \frac{\delta \lambda_\Sigma}{N})\sqrt{N}(\norm{\Qc(t_0)} -M)\\
		= & (T\epsilon - \frac{\delta\lambda_\Sigma}{N})\sqrt{N}\norm{\Qc(t_0)} + \sqrt{N}M(\frac{\delta\lambda_\Sigma}{N}+(T-2)\epsilon)\\
		\lep{e} &(T\epsilon - \frac{\delta\lambda_\Sigma}{N})\sqrt{N}\norm{\Qc(t_0)} + K_2\\
		\lep{f} & -\frac{\delta\mu_\Sigma}{2N}\sqrt{N} \norm{\Qc(t_0)} + K_2
	\end{split}
	\end{equation*}
	where equation (a) comes from the linearity of condition expectation; equation (b) follows from the tower property of conditional expectation and the fact that $\Qc(t)$, $\A(t)$ and $\s(t)$ are conditionally independent of $Z(t_0)$ when given $Z(t)$; inequality (c) follows from Lemmas \ref{lem:tilted} and \ref{lem:propertilted} since under policy $p$ there exists at least one time slot $t^*$ within which at least a $\delta$-tilted distribution (or one of its equivalent distribution in inner-product is $\delta$-tilted) is adopted and all the other time slots a tilted distribution is used; inequality (d) follow from the fact $|\norm{\Qc(t_0+T)} - \norm{\Qc(t_0)}| \le M = 2T\sqrt{N}\max\{A_{max},S_{max}\}$ which is shown in Eq. \eqref{eq:boundedQc} and the fact $\epsilon - \frac{\delta \lambda_\Sigma}{N} < 0$ for all $\epsilon < \epsilon_0$; inequality (e) comes from the fact that $\sqrt{N}M(\frac{\delta\lambda_\Sigma}{N^2}+(T-2)\epsilon) \le K_2 = \sqrt{N}M(\frac{\delta\mu_\Sigma}{N}+T\mu_\Sigma)$, which is independent of $\epsilon$; inequality (f) holds since $\epsilon < \epsilon_0$ and $\lambda_\Sigma = \mu_\Sigma - \epsilon$. Therefore, since both $ -\frac{\delta\mu_\Sigma}{2N}\sqrt{N}$ and $K_2$ are independent of $\epsilon$, the condition of Lemma \ref{LEMMA:DELAY} is satisfied, hence policy $p$ is heavy-traffic delay optimal.\qed

\subsection{Proof of Proposition \ref{prop:jbt}}
\label{sec:52}
Let us first look at assertion 1, i.e., JBT-$d$ is in $\Pi$ under homogeneous servers. Based on Eq. \eqref{eq:Pn}, 
 we can conclude that for any $t \ge 0$, the dispatching distribution is a tilted distribution for all $Z(t)$. We are left to show that at time slot $rT+1$, $r \in \{0,1,2,\ldots\}$, the dispatching distribution is at least a $\delta$-distribution for some positive $\delta$. This is equivalent to finding the maximum value for $P_N(rT+1)$ and the minimum value of $P_1(rT+1)$ for all queue length state. In fact, they are achieved at the same time when $\tilde{p}_N$ is in its largest value based on Eq. \eqref{eq:important1}, which is repeated as follows. 
\begin{equation*}
P_n(t) = \sum_{i=n}^N \tilde{p}_i(t)\frac{1}{i}
\end{equation*}
Then, there are two cases to consider.

(a) Note that  $\tilde{p}_N(rT+1)=1$ if and only if all the servers have the same queue length at the end of time slots $rT$ (sampling slots), which are also the queue length state at the beginning of $rT+1$, i.e., $\Q(rT+1)$. In this case, it can be easily seen that $P_n(rT+1) = \frac{1}{N}$ for all $n$, which is not a $\delta$-tilted distribution. However, it is the an equivalent distribution in inner-product to $\hat{P}_1(rT+1) = 1$ and $\hat{P}_n = 0$ for $2\le n\le N$ as all the queue lengths are equal, which is indeed a $\delta$-distribution.

(b) If the queue lengths are not all equal at the end of time slots $rT$, then the maximum value for $\tilde{p}_N(rT+1)$ is strictly less than $1$ and it is obtained when the queue length in the state that there are $N-1$ servers that have the same queue length, which is strictly larger than the remaining one. In this case, by sampling $d$ servers uniformly at random at the end of times slots $rT$, the probability for the event that there are $N$ IDs in memory, i.e., $\tilde{p}_N(rT+1)$ is given by 
\begin{equation*}
	\tilde{p}_N(rT+1) = 1 - \tilde{p}_1(rT+1) = 1- \frac{d}{N}
\end{equation*}
Therefore, we have $P_1(rT+1) = \frac{d}{N} + \frac{N-d}{N^2}$ and $P_n(rT+1) = \frac{N-d}{N^2}$, which is inner product equivalent to $\hat{P}_1(rT+1) = \frac{d}{N} + \frac{N-d}{N^2}$, $\hat{P}_2(rT+1) = (N-1)\frac{N-d}{N^2}$ and $\hat{P}_n(rT+1) = 0$ for all $3 \le n\le N$ as the $N-1$ queues have the same queue length. As a result, we have for this state $Z(rT+1)$
\begin{equation*}
	\hat{\Delta}_1(rT+1) = \frac{d}{N}(1-\frac{1}{N}) \text{ and } \hat{\Delta}_N(rT+1) = -\frac{1}{N}
\end{equation*}
Thus, it is a $\delta$-distribution with $\delta = \min\{\frac{d}{N}(1-\frac{1}{N}), \frac{1}{N}\}$, which is the lower bound for $\delta$. That is, for any state $Z(rT+1)$, the dispatching distribution is at least a $\delta$-distribution. Therefore, every $T+1$ time slots, there exists one time slot in which the dispatching distribution (or an inner product equivalent distribution) is at least a $\delta$-tilted distribution with $\delta = \min\{\frac{d}{N}(1-\frac{1}{N}), \frac{1}{N}\}$.

The proof for heterogeneous servers follows exact the same idea with additional care on the service rate.  The probability for the server $\sigma_t(n)$ to be selected at time $t$, i.e., $P_n(t)$ is given by
		\begin{equation*}
		\label{eq:important}
			P_n(t) = \mu_{\sigma_t(n)} \sum_{i=n}^N \frac{\tilde{p}_i(t)}{\sum_{ \{j \in m(t), |m(t)| = i \} } \mu_j}
		\end{equation*}
From it we can easily see that if $\Delta_n(t) = P_n(t) - \frac{\mu_{\sigma_t(n)}} {\mu_\Sigma} $ is positive, then we must have that $\Delta_{n-1}(t)$ is also positive as it has one more term in the equation above. Therefore, we can find a $k$ between $2$ and $N$ such that $\Delta_n(t) = P_n(t) - \frac{\mu_{\sigma_t(n)}} {\mu_\Sigma} \ge 0$ for all $n<k$ and  $\Delta_n(t) \le 0$ for all $n\ge k$. Therefore, condition (i) of $\Pi$ is satisfied.

For the condition (ii), we need to find the maximum value of $\tilde{p}_N(rT+1)$ to bound $\delta$. There are also two cases as before.

(a) If $\tilde{p}_N(rT+1) = 1$, the we must have that the queue lengths are all equal at the end of time slots $rT$, which is the same as the beginning of time slot $rT+1$. In this case, $P_n(rT+1) = \frac{\mu_{\sigma_t(n)}}{\mu_\Sigma}$ for all $n$. Note that this dispatching distribution an equivalent distribution in inner-product to $\hat{P}_1(rT+1) = 1$ and $\hat{P}_n = 0$ for $2\le n\le N$ as all the queue lengths are equal, which is a $\delta$-distribution.

(b) If $\tilde{p}_N(rT+1) \neq 1$, the maximum value of $\tilde{p}_N(rT+1)$ is obtained when there are $N-1$ servers that have the same queue length, which is strictly larger than the remaining one. In this case, we have $\tilde{p}_N(rT+1) = 1 - \tilde{p}_1(rT+1) = 1- \frac{d}{N}$ as before. Thus, we can obtain 
\begin{equation*}
	P_1(rT+1) = \frac{d}{N} + (1-\frac{d}{N})\frac{\mu_{\sigma_t(1)} }{\mu_\Sigma} 
\end{equation*}
and $P_n(rT+1) = (1-\frac{d}{N})\frac{\mu_{\sigma_t(n)} }{\mu_\Sigma}$, which is inner product equivalent to $\hat{P}_1(rT+1) = P_1(rT+1)$, $\hat{P}_2(rT+1) = \sum_{n=2}^N P_n(rT+1)$ and $\hat{P}_n(rT+1) = 0$ for all $3 \le n \le N$ since the last $N-1$ servers have the same queue lengths. As a result, we have for this $Z(rT+1)$
\begin{equation*}
	\hat{\Delta}_1(rT+1) = \frac{d}{N}(1-\frac{\mu_{\sigma_t(1)}}{\mu_\Sigma}) \text{ and } \hat{\Delta}_N(rT+1) = -\frac{\mu_{\sigma_t(N)}}{\mu_\Sigma}
\end{equation*}
Thus, it is  a $\delta$-distribution with $\delta = \min\{\frac{d}{N}(1-\frac{\mu_{max}}{\mu_\Sigma}), \frac{\mu_{min}}{\mu_\Sigma}\}$, in which  $\mu_{max} = \max_{n \in \mathcal{N}} \mu_n$ and $\mu_{min} = \min_{n \in \mathcal{N}} \mu_n$, which is the lower bound of $\delta$. Hence, for any $Z(rT+1)$, the dispatching probability distribution (or its inner product equivalent one) is at least a $\delta$-distribution. \qed

\section{Conclusion}
  We introduce a class $\Pi$ of flexible load balancing policies, which are shown to be throughput and heavy-traffic delay optimal.
  This class includes as special cases
  JSQ, power-of-d, and also allows flexibility in designing other new policies.  
  The JIQ policy, albeit exhibiting a good performance when the traffic load is not heavy,
  is not in $\Pi$ since it is not heavy-traffic delay optimal even for homogeneous servers.  A new policy called JBT-$d$ is proposed in the class $\Pi$,  which enjoys the simplicity of JIQ while guaranteeing heavy-traffic delay 
  optimal. A unified analytic framework is established to characterize this class of policies by exploring their common characteristics and provide sufficient conditions that guarantee the heavy-traffic delay optimality.  Extensive simulations  are used to demonstrate the good performance and low complexity of the proposed policy compared to other existing ones.

\bibliographystyle{plain}
\bibliography{ref}
\appendix
%



\section{Proof of Lemma \ref{LEMMA:THTROUGHPUT}}
\label{sec:proof_lemma_throughput}

	Before we present the proof of Lemma \ref{LEMMA:THTROUGHPUT}, we first introduce two lemmas which will be the key ingredients in the proof. The first lemma enables us to bound the moments of a stationary distribution based on drift condition, which can be simplified by the second lemma. 

	The following lemma is introduced in \cite{wang2016maptask}, which is an extension of Lemma 1 in \cite{eryilmaz2012asymptotically} and can be proved from the results in \cite{hajek1982hitting}.
	    \begin{lemma}
	      \label{lem:basis}
	        For an irreducible aperiodic and positive recurrent Markov chain $\{X(t), t \ge 0\}$ over a countable state space $\mathcal{X}$, which converges in distribution to $\overline{X}$,  and suppose $V: \mathcal{X} \rightarrow \mathbb{R}_{+}$ is a Lyapunov function. We define the $T$ time slot drift of $V$ at $X$ as 
	        \[\Delta V(X):= [V(X(t_0+T)) - V(X(t_0))] \mathcal{I}(X(t_0) = X),\]
	        where $\mathcal{I}(.)$ is the indicator function. Suppose for some positive finite integer $T$, the $T$ time slot drift of $V$ satisfies the following conditions:

	        \begin{itemize}
	          \item (C1) There exists an $\gamma> 0$ and a $\kappa <  \infty$ such that for any $t_0 = 1,2,\ldots$ and for all $X \in \mathcal{X}$ with $V(X)\ge \kappa$, 
	          \[\mathbb{E}\left[\Delta V(X) \mid X(t_0) = X\right]\le -\gamma.\]
	          \item (C2) There exists a constant $D < \infty$ such that for all $X\in \mathcal{X}$,
	          \[\mathbb{P}(|\Delta V(X)| \le D) = 1.\]
	        \end{itemize}

	        Then $\{V(X(t)), t\ge0\}$ converges in distribution to a random variable $\overline{V}$, and there exists constants $\theta^* > 0$ and $C^* < \infty$ such that $\mathbb{E}\left[e^{\theta^* \overline{V}}\right] \le C^* $, which directly implies that all moments of random variable $\overline{V}$  exist and are finite. More specifically, there exist finite constants $\{ M_r, r\in \mathbb{N} \}$ such that for each positive $r$, $\ex{V(\overline{X})^r} \le M_r$, where $M_r$ are fully determined by $\kappa$, $\gamma$ and $D$.
	    \end{lemma}

	\begin{lemma}
	    \label{lem:drift_of_Qsquare}
	      For any $t\ge0$, we have 
	      \begin{equation}
	      \label{eq:lemma2}
	        \norm{\Q(t+1)}^2 - \norm{\Q(t)}^2 \le 2\inner{\Q(t)}{\A(t)-\s(t)} + K
	      \end{equation}
	      where $K$ is a finite constant.
	\end{lemma}

	    \begin{proof}
	      Consider the left-hand-side (LHS) of Eq. \eqref{eq:lemma2}.
	      \begin{equation*}
	        \begin{split}
	        LHS &= \norm{\Q(t) + \A(t)-\s(t) +\UU(t)}^2- \norm{\Q(t)}^2\\
	        & \lep{a} \norm{\Q(t) + \A(t)-\s(t)}^2- \norm{\Q(t)}^2\\
	        & = 2\inner{\Q(t)}{\A(t)-\s(t)} + \norm{\A(t)-\s(t)}^2\\
	        & \lep{b} 2\inner{\Q(t)}{\A(t)-\s(t)}+ K
	        \end{split}
	      \end{equation*}
	      where inequality (a) holds as $[\max(a,0)]^2 \le a^2 $ for any $a\in \mathbb{R}$; in inequality (b), $K \triangleq N\max(A_{\max},S_{\max})^2$ holds due to the assumptions that  $A_\Sigma(t) \le A_{\max}$ and $S_n(t) \le S_{\max}$ for all $t \ge 0$ and all $n \in \mathcal{N}$, and independent of the queue length.
	    \end{proof}  
	    We are now ready to prove Lemma \ref{LEMMA:THTROUGHPUT}.

    \textbf{Proof of Lemma \ref{LEMMA:THTROUGHPUT}:} 
    The proof follows from the application of Lemma \ref{lem:basis} to the Markov chain $\{Z^{(\epsilon)}(t),t\ge0\}$ with Lyapunov function $V(Z^{(\epsilon)}) := \norm{\Q^{(\epsilon)}}$ and $T = T_1$ since $m^{(\epsilon)}(t)$ is always finite. In particular, this proof is completed in two steps, where the superscript $^{(\epsilon)}$ will be omitted for ease of notations. 

    (i) First, in order to apply Lemma \ref{lem:basis}, we need to show that the Markov chain $\{Z(t), t\ge0\}$ is irreducible, aperiodic and positive recurrent under the hypothesis of Lemma \ref{LEMMA:THTROUGHPUT}. It can be easily seen that $\{Z(t),t\ge0\}$ is irreducible and aperiodic. Thus, we are left with the task to prove that the Markov chain is positive recurrent. By the extension of Foster-Lyapunov theorem, it suffices to find a Lyapunov function and a positive constant $T$ such that the expected $T$ time slot Lyapunov drift is bounded within a finite subset of the state space and negative outside this subset. 

    Consider the Lyapunov function $W(Z):= \norm{\Q}^2$, and the corresponding expected $T_1$ time slot mean conditional Lyapunov drift under the hypothesis of Lemma \ref{LEMMA:THTROUGHPUT}.
    \begin{equation}
    \label{eq:upperonD}
        \begin{split}
          &\ex{W(Z(t_0+T_1)) - W(Z(t_0)) \mid Z(t_0)} \\
          = & \ex{\norm{\Q(t_0+T_1)}^2 - \norm{\Q(t_0)}^2 \mid Z(t_0)}\\
          = & \ex{\sum_{t = t_0}^{t_0+T_1-1}(\norm{\Q(t+1)}^2 - \norm{\Q(t)}^2) \mid Z(t_0)}\\
          \lep{a} &2 \ex{\sum_{t = t_0}^{t_0+T_1-1}\inner{\Q(t)}{\A(t)-\s(t)} + K \mid Z(t_0)}\\
          \lep{b} & -2\gamma \norm{\Q(t_0)} + 2K_1+ 2KT_1
        \end{split}
    \end{equation}
    where inequality (a) follows from Lemma \ref{lem:drift_of_Qsquare}, and inequality (b) results directly from the hypothesis in Eq. \eqref{eq:condition_throughput}. Pick any $\beta > 0$ and let $\mathcal{B} = \{Z\in \mathcal{S}: \norm{\Q} \le \frac{K_1+KT_1+\beta}{\gamma}\}$. Then $\mathcal{B}$ is a finite subset of $\mathcal{S}$ as $m(t)$ is finite. Moreover, for any $Z\in \mathcal{B}$, the conditional mean drift is less or equal to $2K_1+ 2KT_1$, and for any $Z \in \mathcal{B}^c$, it is less than or equal to  $-\beta$. This finishes the proof of positive recurrence for any $\epsilon > 0$, and hence throughput optimal.

    (ii) Second, in order to show that the hypothesis in Lemma \ref{LEMMA:THTROUGHPUT} also ensures the bounded moments for the stationary distribution, we will resort to Lemma \ref{lem:basis}. Thus, we need to check Conditions (C1) and (C2), respectively.

    For Condition (C1), we have 
    \begin{equation*}
        \begin{split}
          &\ex{\Delta V({Z})\mid Z(t_0)=Z}\\
          = &\ex{\norm{\Q(t_0+T_1)} - \norm{\Q(t_0)} \mid Z(t_0) = Z}\\
          = &\ex{\sqrt{\norm{\Q(t_0+T_1)}^2} - \sqrt{\norm{\Q(t_0)}^2}  \mid Z(t_0) = Z}\\
          \lep{a} &\frac{1}{2\norm{\Q(t_0)}}\ex{\norm{\Q(t_0+T_1)}^2 - \norm{\Q(t_0)}^2 \mid Z(t_0) = Z}\\
          \lep{b} & -\gamma + \frac{K_1 + KT_1}{\norm{\Q(t_0)}}
        \end{split}
    \end{equation*}
    where inequality (a) follows from the fact that $f(x) = \sqrt{x}$ is concave; (b) comes from the upper bound in Eq. \eqref{eq:upperonD}. Hence, (C1) in Lemma \ref{lem:basis} is verified.

    For Condition (C2), we have 
    \begin{equation*}
        \begin{split}
          |\Delta V(Z)| &= | \norm{\Q(t_0+T_1)} - \norm{\Q(t_0)} | \mathcal{I}(Z(t_0) = Z)\\
          & \lep{a} \norm{\Q(t_0+T_1) - \Q(t_0)}\mathcal{I}(Z(t_0) = Z)\\
          & \lep{b} T_1 \sqrt{N}\max(A_{\max},S_{\max})
        \end{split} 
    \end{equation*}
    where inequality (a) follows from the fact that  $|\norm{{\bf x}} - \norm{{\bf y}}| \le \norm{{\bf x} - {\bf y}}$ holds for any ${\bf x}$, ${\bf y} \in \mathbb{R}^N$; inequality (b) holds due to the assumptions that the $A_\Sigma(t) \le A_{\max}$ and $S_n(t) \le S_{\max}$ for all $t \ge 0$ and all $n \in \mathcal{N}$, and independent of the queue length. This verifies Condition (C2) and hence complete the proof of Lemma \ref{LEMMA:THTROUGHPUT}.\qed

\section{Proof of Lemma \ref{LEMMA:DELAY}}
\label{sec:proof_lemma_delay}

    We now proceed to prove Lemma \ref{LEMMA:DELAY}. Before we present the proof, the following lemmas which serve as useful preliminary steps are first introduced. 
Denote by $\Qp$ and $\Qc$ the parallel and perpendicular components of the queue length vector $\Q$ with respect to the line $\mathbf{c} = \frac{1}{\sqrt{N}}\mathbf{1}$, i.e.,
    \begin{equation}
    \label{eq:def_Qc}
    	\Qp:=\inner{\mathbf{c}}{\Q}\mathbf{c}  \quad \quad \Qc:=\Q - \Qp
    \end{equation}
The following lemma is a natural extension of Lemma 7 in \cite{eryilmaz2012asymptotically} to $T$ time slots.
    \begin{lemma}
    \label{lem:bound}
      Define the following Lyapunov functions
      \begin{equation*}
        V_\perp(Z) := \norm{\Qc}, W(Z) :=\norm{\Q}^2 \text{ and } W_\parallel(Z):= \norm{\Qp}^2
      \end{equation*}
      with the corresponding $T$ time-slot drift given by
      \begin{equation*}
        \begin{split}
            &\Delta V_\perp(Z):= [V_\perp(Z(t_0+T)) - V_\perp(Z(t_0))] \mathcal{I}(Z(t_0) = Z)\\
            &\Delta W(Z):= [W(Z(t_0+T)) - W(Z(t_0))] \mathcal{I}(Z(t_0) = Z)\\
            &\Delta W_\parallel(Z):= [W_\parallel(Z(t_0+T)) - W_\parallel(Z(t_0))] \mathcal{I}(Z(t_0) = Z)\\
        \end{split}   
      \end{equation*}
      Then, the drift of $V_\perp(.)$ can be bounded in terms of $W(.)$ and $W_\parallel(.)$ as follows.
      \begin{equation*}
        \Delta V_\perp(Z) \le \frac{1}{2\norm{\Qc}}(\Delta W(Z)-\Delta W_\parallel(Z)) 
      \end{equation*}
      for all $Z \in \mathcal{S}$.
    \end{lemma}

    \begin{lemma}
    \label{lem:drift_of_Qparallel}
        For any $t\ge 0$, we have 
      \begin{equation*}
        \norm{\Qp(t+1)}^2 - \norm{\Qp(t)}^2 \ge 2\inner{\Qp(t)}{\A(t)-\s(t)}.
       \end{equation*}
    \end{lemma}
    \begin{proof}
      \begin{equation*}
        \begin{split}
        &\norm{\Qp(t+1)}^2 - \norm{\Qp(t)}^2 \\ 
        = &2\inner{\Qp(t)}{\Qp(t+1) - \Qp(t)} + \norm{\Qp(t+1) - \Qp(t)}^2\\
        \ge &2\inner{\Qp(t)}{\Qp(t+1) - \Qp(t)} \\
        = &2 \inner{\Qp(t)}{\Q(t+1) -\Q(t)} - 2\inner{\Qp(t)}{\Qc(t+1) - \Qc(t)}\\
        \gep{a} &2 \inner{\Qp(t)}{\Q(t+1) -\Q(t)}\\
        \gep{b} &2 \inner{\Qp(t)}{\A(t)-\s(t)}
      \end{split}
    \end{equation*}
    where the inequality (a) is true because $\inner{\Qp(t)}{\Qc(t)} = 0$ and $\inner{\Qc(t+1)}{\Qp(t)} = 0$; (b) follows from the fact that all the components of $\Qp(t)$ and $\UU(t)$ are nonnegative.
    \end{proof}
    We are now ready to prove the following result, which is  often called \textit{state space collapse} and is the key ingredient for establishing heavy traffic delay optimality. It shows that under the hypothesis of Lemma \ref{LEMMA:DELAY}, the multi-dimension space for the queue length vector will reduce to one dimension in the sense that the deviation from the line $\mathbf{c}$ is bounded by a constant, which is independent with the heavy-traffic parameter $\epsilon$.
    \begin{lemma}
    \label{lem:collapse}
    	If the hypothesis of Lemma \ref{LEMMA:DELAY} holds, then we have that $\Qc$ is bounded in the sense that in steady state there exists finite constants $\{L_r, r \in \mathbb{N}\}$ such that 
    	\begin{equation*}
    		\ex{\norm{\overline{\Q}_\perp^{(\epsilon)} } ^r} \le L_r
    	\end{equation*}
    	for all $\epsilon \in (0,\epsilon_0)$ and $r \in \mathbb{N}$.
    \end{lemma}
    \begin{proof}
    	It suffices to show that $V_\perp(Z)$ satisfies the Conditions (C1) and (C2) in Lemma \ref{lem:basis}. Fix $\epsilon \in (0,\epsilon_0)$, and the superscript will be omitted for simplicity in the following arguments.

      (i) For the Condition (C1), let $\Lambda(t) := \norm{\Q(t+1)}^2 - \norm{\Q(t)}^2$ and $\Lambda_\parallel(t) := \norm{\Qp(t+1)}^2 - \norm{\Qp(t)}^2$. Then, we have 
    \begin{equation*}
          \begin{split}
            &\ex{\Delta V_\perp(Z) \mid Z(t_0) = Z} \\
            \lep{a}&\frac{1}{2\norm{\Qc}} \ex{\Delta W(Z) - \Delta W_\parallel(Z)\mid Z(t_0) = Z}\\
            = &\frac{1}{2\norm{\Qc}}\ex{\sum_{t=t_0}^{t_0+T_2-1}\Lambda(t)  - \Lambda_\parallel(t)\mid Z(t_0) = Z }\\
            \lep{b}& \frac{1}{2\norm{\Qc(t_0)}}\ex{\sum_{t=t_0}^{t_0+T_2-1} 2\inner{\Qc(t)}{\A(t)-\s(t)} + K \mid Z(t_0) = Z}\\
            \lep{c} &-\eta + \frac{2K_2 + KT_2}{2\norm{\Qc(t_0)} }
          \end{split}
    \end{equation*}
    where the inequality (a) follows from Lemma \ref{lem:bound}; the inequality (b) holds as a result of Lemmas \ref{lem:drift_of_Qsquare} and \ref{lem:drift_of_Qparallel}; the inequality (c) follows directly from the hypothesis  Eq. \eqref{eq:cond_delay}. Hence, the Condition (C1) is verified.

    (ii) For the Condition (C2), we have 
    \begin{align}
    \label{eq:boundedQc}
            & |\Delta V_\perp(Z)| \nonumber\\
            = & | \norm{\Qc(t_0+T_2)} - \norm{\Qc(t_0)} | \mathcal{I}(Z(t_0) = Z)\nonumber \\
            \lep{a} & \norm{\Qc(t_0+T_2) - \Qc(t_0)}\mathcal{I}(Z(t_0) = Z)\nonumber\\
            = &\norm{\Q(t_0+T_2) - \Qp(t_0+T_2) - \Q(t_0) + \Qp(t_0)} \mathcal{I}(Z(t_0) = Z)\nonumber\\
            \lep{b} &\norm{\Q(t_0+T_2) - \Q(t_0)}  + \norm{\Qp(t_0+T_2) - \Qp(t_0)} \mathcal{I}(Z(t_0) = Z)\nonumber\\
            \lep{c} & 2\norm{\Q(t_0+T_2) - \Q(t_0)} \mathcal{I}(Z(t_0) = Z)\nonumber\\
            \lep{d} & 2T_2\sqrt{N} \max(A_{\max},S_{\max})  
    \end{align}
    where the inequality (a) follows from the fact that  $|\norm{{\bf x}} - \norm{{\bf y}}| \le \norm{{\bf x} - {\bf y}}$ holds for any ${\bf x}$, ${\bf y} \in \mathbb{R}^N$; inequality (b) follows from triangle inequality; (c) holds due to the non-expansive property of projection to a convex set. (d) holds due to the assumptions that the $A_\Sigma(t) \le A_{\max}$ and $S_n(t) \le S_{\max}$ for all $t \ge 0$ and all $n \in \mathcal{N}$, and independent of the queue length. This verifies Condition (C2) and hence complete the proof of Lemma \ref{lem:collapse}.
    \end{proof}

    The following result on the unused service is another key ingredient for establishing heavy-traffic delay optimal.
    \begin{lemma}
    \label{lem:unused_service}
    For any $\epsilon >0$ and $t\ge 0$, we have 
    \begin{equation*}
    \label{eq:QtimesU}
    Q_n^{(\epsilon)}(t+1)U_n^{(\epsilon)}(t) = 0 \text{ and } q^{(\epsilon)}(t+1)u^{(\epsilon)}(t) = 0.
    \end{equation*}
    If the system has a finite first moment, then we have for some constants $c_1$ and $c_2$
 	\begin{equation*}
 	\label{eq:ule}
 		\ex{\norm{\overline{\UU}^{(\epsilon)}}^2} \le c_1 \epsilon \text{ and } \ex{({\overline{u}^{(\epsilon)}})^2} \le c_2\epsilon
 	\end{equation*}
    \end{lemma}
    \begin{proof}
    	According to the queue dynamic in Eq. \eqref{eq:q_dynamic}, we can see when $U_n(t)$ is positive, $Q_n(t+1)$ must be zero, which gives the results $Q_n^{(\epsilon)}(t+1)U_n^{(\epsilon)}(t) = 0$ for all $n\in \mathcal{N}$ and all $t\ge0$, and the corresponding result for the resource-pooled system $q^{(\epsilon)}(t+1)u^{(\epsilon)}(t) = 0$. 

	Then, let us consider the Lyapunov function $W_1(Z(t)) = \norm{\Q(t)}_1$. In the steady state with a finite first moment, the mean drift of $W_1(Z(t))$ is zero. Then, we have
	\begin{equation*}
	\begin{split}
	  0 = \ex{\norm{\A^{(\epsilon)}}_1 - \norm{\s}_1 + \norm{\overline{\UU}^{(\epsilon)} }_1}
	\end{split}  
	\end{equation*}
	which directly implies 
  \begin{equation}
  \label{eq:unused_one_norm}
    \ex{\norm{\overline{\UU}^{(\epsilon)}}_1} = \epsilon
  \end{equation}
   Moreover, due to the fact that $U_n(t) \le S_{\max}$ for all $n\in\mathcal{N}$ and $t\ge0$, we have $\norm{\overline{\UU}^{(\epsilon)}}^2 \le S_{\max}\norm{\overline{\UU}^{(\epsilon)}}_1$. Therefore, we can conclude that $\ex{\norm{\overline{\UU}^{(\epsilon)}}^2} \le S_{\max} \epsilon$ and $\ex{({\overline{u}^{(\epsilon)})}^2} \le NS_{\max} \epsilon$.
    \end{proof}
    Now, we are well prepared to prove Lemma \ref{LEMMA:DELAY}

    \textbf{Proof of Lemma \ref{LEMMA:DELAY}:}
    First, let us consider the Lyapunov function $V_1(Z) := \norm{\Q}_1^2$ and the corresponding conditional mean drift, defined as $D_1(Z(t)) := \ex{V_1(Z(t+1)) - V_1(Z(t)) \mid Z(t) = Z}.$

  Then, we have the following equation, in which the time reference $(t)$ will be omitted after the second step for brevity and $\Q^+ := \Q(t+1)$.
    \begin{equation}
    \label{eq:heavy}
      \begin{split}
        &D_1(Z(t))\\
        = & \ex{ \norm{\Q(t+1)}_1^2 - \norm{\Q(t)}_1^2 \mid Z(t) = Z    }\\
        = & \ex{\left(\norm{\Q(t)}_1+\norm{\A(t)}_1-\norm{\s(t)}_1 + \norm{\UU(t)}_1 \right)^2  \mid Z(t) = Z  }\\
        & - \ex{\norm{\Q(t)}_1^2 \mid Z(t) = Z  }\\
        = & \mathbb{E}\left[ 2\norm{\Q{}}_1\left( \norm{\A}_1 - \norm{\s}_1\right) + \left( \norm{\A}_1-\norm{\s}_1\right)^2\right.\\
         & +\left.{} 2\left(\norm{\Q}_1 +\norm{\A}_1 - \norm{\s}_1 \right) \norm{\UU}_1 + \norm{\UU}_1^2 \vphantom{\left( \norm{\A}_1-\norm{\s}_1\right)^2}      \mid Z \right]\\
        = & \mathbb{E}\left[ 2\norm{\Q}_1\left( \norm{\A}_1 - \norm{\s}_1\right) + \left( \norm{\A}_1-\norm{\s}_1\right)^2\right.\\
         & +\left.{} 2\norm{\Q^+}_1 \norm{\UU}_1 - \norm{\UU}_1^2 \vphantom{\left( \norm{\A}_1-\norm{\s}_1\right)^2}      \mid Z \right]\\
        \le &  \mathbb{E}\left[ 2\norm{\Q}_1\left( \norm{\A}_1 - \norm{\s}_1\right) + \left( \norm{\A}_1-\norm{\s}_1\right)^2\right.\\
         & +\left.{} 2\norm{\Q^+}_1 \norm{\UU}_1 \vphantom{\left( \norm{\A}_1-\norm{\s}_1\right)^2}      \mid Z \right]
      \end{split}
    \end{equation}


    Under the hypothesis of Lemma \ref{LEMMA:DELAY}, there exists a steady-state distribution with finite moments for any $\epsilon > 0$. Therefore, the mean drift in steady-state is zero, i.e., $\ex{{D_1}(\overline{Z}^{(\epsilon)})} = 0$. Therefore, taking the expectation of both sides of Eq. \eqref{eq:heavy} with respect to the steady-state distribution $\overline{Z}^{(\epsilon)}$, yields 
    \begin{equation*}
        \begin{split}
          \epsilon \ex{\sum_{n=1}^{N}\overline{Q}_n^{(\epsilon)}} \le \frac{\zeta^{(\epsilon)}}{2} + \ex{\norm{\overline{\Q}^{(\epsilon)}(t+1)}_1 \norm{\overline{\UU}^{(\epsilon)} (t)}_1}
        \end{split}
    \end{equation*}
    where $\zeta^{(\epsilon)} = (\sigma_\Sigma^{(\epsilon)})^2 + \nu_\Sigma^2 + \epsilon^2$. 
	For the resource-pooled system, by letting $N=1$ in Eq. \eqref{eq:heavy} and taking the expectation  with respect to $\overline{q}^{(\epsilon)}$, we have
    \begin{equation*}
        \begin{split}
        \epsilon \ex{\overline{q}^{(\epsilon)}} & = \frac{\zeta^{(\epsilon)}}{2} + \ex{\overline{q}^{(\epsilon)}(t+1)u^{(\epsilon)}(t)} - \frac{1}{2}\ex{(u^{(\epsilon)})^2}.\\
        \end{split}
    \end{equation*}

    Then, based on the property on the unused service in Lemma \ref{lem:unused_service}, we have 
    \begin{equation}
    \label{eq:cross}
          \frac{\zeta^{(\epsilon)}}{2}  - \frac{c_2}{2}\epsilon \le  \epsilon\ex{\overline{q}^{(\epsilon)}} \le \epsilon\ex{\sum_{n=1}^{N}\overline{Q}_n^{(\epsilon)}} \le \frac{\zeta^{(\epsilon)}}{2} + \overline{B}^{(\epsilon)}
    \end{equation}
    where $\overline{B}^{(\epsilon)} :=\ex{\norm{\overline{\Q}^{(\epsilon)}(t+1)}_1 \norm{\overline{\UU}^{(\epsilon)} (t)}_1}$. 

    Therefore, in order to show heavy-traffic delay optimality, all we need to show is that $\lim_{\epsilon \downarrow 0} \overline{B}^{(\epsilon)} = 0$. Note that $ \overline{B}^{(\epsilon)}$ can be bounded as follows.
    \begin{equation*}
    \begin{split}
      \overline{B}^{(\epsilon)} & \ep{a} N\ex{\inner{\overline{\UU}^{(\epsilon)}(t)} {-\overline{\Q}^{(\epsilon)}_{\perp}(t+1)}}\\
      & \lep{b} N \sqrt{\ex{\norm{\overline{\UU}^{(\epsilon)}_{\perp}(t)}^2} \ex{\norm{\overline{\Q}^{(\epsilon)}_{\perp}(t+1)}^2}}\\
      & \ep{c} N \sqrt{\ex{\norm{\overline{\UU}^{(\epsilon)}_{\perp}(t)}^2} \ex{\norm{\overline{\Q}^{(\epsilon)}_{\perp}(t)}^2}},\\
    \end{split} 
    \end{equation*}
    where the equality (a) comes from the property $Q_n^{(\epsilon)}(t+1)U_n^{(\epsilon)}(t) = 0$ for all $n\in \mathcal{N}$ and all $t\ge0$ in Lemma \ref{lem:unused_service} and the definition of $\Qc$; the inequality (b) holds due to Cauchy-Schwartz inequality; the last equality (c) is true since the distributions of $\overline{\Q}^{(\epsilon)}_{\perp}(t+1)$ and $\overline{\Q}^{(\epsilon)}_{\perp}(t)$ are the same in steady state.

    As shown in Lemma \ref{lem:collapse}, $\ex{\norm{ \overline{\Q}_{\perp}^{(\epsilon)} }^2} \le L_2$ holds for all $\epsilon \in (0,\epsilon_0)$ and some constant $L_2$ which is independent of $\epsilon$. Meanwhile, note that $\ex{\norm{\overline{\UU}^{(\epsilon)}}^2} \le c_1 \epsilon$ for some $c_1$ independent of $\epsilon$ based on Lemma \ref{lem:unused_service}. Then, we have for all $\epsilon \in (0,\epsilon_0)$
    \begin{equation}
    \label{eq:final}
            \overline{B}^{(\epsilon)} \le  N\sqrt{c_1 \epsilon L_2}
    \end{equation}
    Therefore, it can be  seen from Eq. \eqref{eq:final} that $\lim_{\epsilon \downarrow 0} \overline{B}^{(\epsilon)} = 0$, which directly implies $\lim_{\epsilon \downarrow 0} \epsilon\ex{\sum_n\overline{Q}_n^{(\epsilon)}} = \lim_{\epsilon \downarrow 0} \epsilon \ex{\overline{q}^{(\epsilon)}}$, and thus the proof of Lemma \ref{LEMMA:DELAY} is completed.\qed

\externaldocument{proofappendix.tex}
\section{Proof of theorem \ref{THM:JIQ}}
\label{sec:proof_JIQ}
Before we present the proof of Theorem \ref{THM:JIQ}, we first show that both the following Lyapunov functions have finite expectations in steady state under JIQ, i.e., throughput optimal in a strong sense.
\begin{equation*}
	W(Z):=\norm{\Q}^2, V_1(Z):= \norm{\Q}_1^2
\end{equation*}

\begin{lemma}
\label{lem:pre_JIQ}
	Consider a load balancing system with homogeneous servers under JIQ policy, the steady state means $\ex{W(\overline{Z}^{(\epsilon)})}$ and $\ex{V_1(\overline{Z}^{(\epsilon)})}$ are both finite for any $\epsilon > 0$.
\end{lemma}
\begin{proof}
	This proof is a direct application of Lemma \ref{LEMMA:THTROUGHPUT}. Let us consider $T_1 = 1$ in Lemma \ref{LEMMA:THTROUGHPUT}, we have 
	\begin{equation*}
	\begin{split}
		&\ex{\inner{\Q(t)}{\A(t)-\s(t)}\mid Z(t) = Z}\\
		=& \ex{\inner{\Q(t)}{\A(t)}\mid Z(t) = Z} - \ex{\inner{\Q(t)}{\s(t)}\mid Z(t) = Z}\\
		\ep{a}& \ex{\inner{\Q(t)}{\A(t)}\mid Z(t) = Z} - \mu \norm{\Q}_1\\
		\lep{b}& \left(\mu - \frac{\epsilon}{N}\right) \norm{\Q}_1 - \mu \norm{\Q}_1\\
		\lep{c} & -\frac{\epsilon}{N}\norm{\Q}
	\end{split}	
	\end{equation*}
	where step (a) follows from the fact that the servers are homogeneous and each service process is independent with the system state; inequality (b) holds due to the property of JIQ; inequality (c) holds since $l_1$ norm $\norm{\mathbf{x}}_1$ of any vector $\mathbf{x} \in \mathbb{R}^N$ is no smaller than its $l_2$ norm $\norm{\mathbf{x}}$. Thus, according to Lemma 1, we have $\ex{W(\overline{Z}^{(\epsilon)})}$ is finite for any $\epsilon > 0$.

	Note that $\norm{\Q}_1^2 \le N \norm{\Q}^2$. Therefore, it follows that $\ex{V_1(\overline{Z}^{(\epsilon)})}$ is also finite.
\end{proof}
 Now, we are well prepared to prove Theorem \ref{THM:JIQ}

 \textbf{Proof of Theorem \ref{THM:JIQ}:}
Recall Remark \ref{rem:not_optimal}, to prove JIQ is not heavy-traffic delay optimal, it suffices to find a certain combination of arrival and service processes under which the heavy-traffic limit under JIQ cannot achieve the heavy-traffic limit of resource-pooled system, while JSQ can. Mathematically, it is sufficient to show that there exists a subsequence of $\epsilon$ such that Eq.~\eqref{eq:heavy_traffic_def} (definition of heavy-traffic delay optimality) does not hold. In particular, in this proof we consider the case of two homogeneous servers with arrival process in $\mathcal{A}$ and constant service with rate one.

We will prove the result by contradiction. Suppose the result in Theorem \ref{THM:JIQ} does not hold, i.e., JIQ is heavy-traffic delay optimal. Then it means that  $\lim_{\epsilon \downarrow 0} \overline{B}^{(\epsilon)} = 0$, i.e., $\limsup_{\epsilon \downarrow 0} \overline{B}^{(\epsilon)} = 0$ must hold in our considered case. This is true since 
\begin{equation*}
	\frac{\zeta^{(\epsilon)}}{2} + \overline{B}^{(\epsilon)} - c_1 \epsilon \le \epsilon \ex{\sum_{n=1}^{2}\overline{Q}_n^{(\epsilon)}} \le \frac{\zeta^{(\epsilon)}}{2} + \overline{B}^{(\epsilon)},
\end{equation*}
which comes from the fact that $\ex{V_1(\overline{Z}^{(\epsilon)})}$ is finite under JIQ by Lemma \ref{lem:pre_JIQ} and the property of unused service in Lemma \ref{lem:unused_service}.

Note that $\Q^+ := \Q(t+1)$, we have 
\begin{equation*}
\label{eq:B}
\begin{split}
	\overline{B}^{(\epsilon)} &\ep{a} \ex{(\overline{Q}_1^{+})^{(\epsilon)}\overline{U}_2^{(\epsilon)} + (\overline{Q}_2^{+})^{(\epsilon)}\overline{U}_1^{(\epsilon)}}\\
	&\ep{b} 2\ex{(\overline{Q}_1^{+})^{(\epsilon)}\overline{U}_2^{(\epsilon)}} \\
	&\ep{c} 2 \sum_{k=1}^\infty k \mathbb{P}\left( (\overline{Q}_1^{+})^{(\epsilon)} = k,  \overline{U}_2^{(\epsilon)}>0   \right), \\
\end{split}	
\end{equation*}
where inequality (a) follows from the property shown in Lemma \ref{lem:unused_service} that $Q_n(t+1)U_n(t) = 0$ for all $t\ge0$; equality (b) holds due to the symmetry property of JIQ policy; inequality (c) comes from the fact that the unused service is one when it is positive since the service rate is constant $1$. Since $\limsup_{\epsilon \downarrow 0} \overline{B}^{(\epsilon)} = 0$, we have 
\begin{equation}
	\label{eq:case2}
		\lim_{\epsilon \downarrow 0} \sum_{k=1}^\infty k \mathbb{P}\left( (\overline{Q}_1^{+})^{(\epsilon)} = k,  \overline{U}_2^{(\epsilon)}>0   \right) = 0.
\end{equation}
Now, let us consider another Lyapunov function $V_2(Z(t)) = (Q_1(t)-Q_2(t))^2$. By the fact that the service process is constant, we have 
	\begin{equation*}
	\begin{split}
		&\ex{V_2(Z(t+1)) - V_2(Z(t)) \mid Z(t) = Z}\\
		=&2\ex{(Q_1-Q_2)(A_1(t)-A_2(t))\mid Z} + \ex{(A_1(t) - A_2(t))^2 \mid Z}\\
		&-2\ex{Q_1^+U_2(t) + Q_2^+U_1(t) \mid Z} - \ex{(U_1(t) - U_2(t))^2 \mid Z}\\
		= & 2\ex{(Q_1-Q_2)(A_1(t)-A_2(t))\mid Z} + \ex{(A_1(t) - A_2(t))^2 \mid Z}\\
		& - 4\ex{Q_1^+U_2(t)\mid Z} - \ex{(U_1(t) - U_2(t))^2 \mid Z}.
	\end{split}	
	\end{equation*}
Since $\ex{V_2(Z(t))} \le \ex{W(Z(t))} < \infty$ by Lemma \ref{lem:pre_JIQ}, the mean drift of $V_2(.)$ in steady state is zero, which implies the following equation in steady state for any $\epsilon > 0$.
\begin{equation}
	\label{eq:JIQ_important}
	\begin{split}
		&4 \ex{(\overline{Q}_1^{+})^{(\epsilon)}\overline{U}_2^{(\epsilon)}}\\
		=& 2\ex{\left(\overline{Q}_1^{(\epsilon)} - \overline{Q}_2^{(\epsilon)}\right) \left( \overline{A}_1^{(\epsilon)} - \overline{A}_2^{(\epsilon)}\right)} + \ex{\left( \overline{A}_1^{(\epsilon)} - \overline{A}_2^{(\epsilon)}\right)^2}\\
		& -\ex{\left( \overline{U}_1^{(\epsilon)}  -  \overline{U}_2^{(\epsilon)}  \right)^2}.
	\end{split}
\end{equation}

In the following, we will analyze each term on the right-hand side of Eq. \eqref{eq:JIQ_important}.

To start with, let us look at the first term on the right-hand side of Eq. \eqref{eq:JIQ_important}, denoted as $\mathcal{T}_1$. For ease of exposition, we shall omit the superscript $^{(\epsilon)}$ after the first step, and reintroduce it when necessary. In the following term $\mathcal{T}_1$, the equality (a) follows from the symmetry property of JIQ policy; (b) holds since when both queues are non-idle random routing with equal probability is adopted in JIQ policy; (c) comes from the fact that when one queue is idle, the arrival is always routed to the idle queue under JIQ; equality (d) holds since $\overline{\Q}(t+1)$ has the same distribution as $\overline{\Q}(t)$ in steady state.
\begin{equation*}
 \begin{split}
 	\mathcal{T}_1 &=2\ex{\left(\overline{Q}_1^{(\epsilon)} - \overline{Q}_2^{(\epsilon)}\right) \left( \overline{A}_1^{(\epsilon)} - \overline{A}_2^{(\epsilon)}\right)}\\
 	&\ep{a}2\ex{\left(\overline{Q}_1 - \overline{Q}_2 \right)\left( \overline{A}_1- \overline{A}_2\right) \mathbbm{1}\left(\overline{Q}_1 >0, \overline{Q}_2 > 0 \right)}\\
 	& \quad + 4\ex{\left(\overline{Q}_1 - \overline{Q}_2 \right)\left( \overline{A}_1- \overline{A}_2\right) \mathbbm{1}\left(\overline{Q}_1 \ge 0, \overline{Q}_2 = 0 \right)}\\
 	&\ep{b}  4\ex{\left(\overline{Q}_1 - \overline{Q}_2 \right)\left( \overline{A}_1- \overline{A}_2\right) \mathbbm{1}\left(\overline{Q}_1 \ge 0, \overline{Q}_2 = 0 \right)}\\
 	&\ep{c} -4\lambda_{\Sigma} \sum_{k=1}^\infty k\mathbb{P}\left(\overline{Q}_1 = k, \overline{Q}_2 = 0\right)\\
 	& \ep{d} -4\lambda_{\Sigma} \sum_{k=1}^\infty k\mathbb{P}\left(\overline{Q}_1^{+} = k, \overline{Q}_2^{+} = 0, \overline{U}_2 = 0\right)\\
 	& \quad -4\lambda_{\Sigma} \sum_{k=1}^\infty k\mathbb{P}\left(\overline{Q}_1^{+} = k, \overline{Q}_2^{+} = 0, \overline{U}_2 > 0\right).\\ 
 \end{split}
\end{equation*}

Now, let us define three events as follows where $k\ge 1$.
\begin{equation*}
\begin{split}
	&E_{1,k} := \{ \overline{Q}_1^{+} = k, \overline{Q}_2^{+} = 0, \overline{U}_2 = 0\},\\
	&E_{2,k} := \{ \overline{Q}_1(t+2) = k, \overline{Q}_2(t+2) = 0, \overline{U}_2^+ > 0\},\\
	&E_{3,k} := \{ \overline{Q}_1^+ = k, \overline{Q}_2^+ = 0, \overline{U}_2 > 0\}.\\
\end{split}	
\end{equation*}
Then, it can be easily seen that for each occurrence of event $E_{1,k}$, due to the fact that arrival process is in $\mathcal{A}$, there exists a positive probability $p_0$, i.e., the probability of no arrival, such that $E_{2,k-1}$ will happen. This is because when there is no arrival, the unused service for $Q_2$ is $1$ and the queue length of $Q_1$ will decrease by $1$ due to the constant service rate $1$. Since
\begin{equation*}
	\mathbb{P}\left(E_{2,k-1}\right) \ge \mathbb{P}\left(E_{1,k}\right) \mathbb{P}\left(E_{2,k-1} \mid E_{1,k}\right),
\end{equation*}
which directly implies that $\mathbb{P}\left(E_{2,k-1}\right) \ge p_0 \mathbb{P}\left(E_{1,k}\right)$. Thus, we have 
\begin{equation}
\label{eq:E3_E1}
\mathbb{P}\left(E_{3,k-1}\right) = \mathbb{P}\left(E_{2,k-1}\right) \ge p_0 \mathbb{P}\left(E_{1,k}\right). 
\end{equation}
which holds since these events are defined on the steady-state distribution.

Now, let us first rewrite $\mathcal{T}_1$ with respect to events $E_{1,k}$ and $E_{3,k}$ as follows.
\begin{equation}
\label{eq:T1}
\begin{split}
	\mathcal{T}_1 &= -4\lambda_{\Sigma}\left( \sum_{k=1}^\infty k \mathbb{P}\left(E_{1,k}\right) + \sum_{k=1}^\infty k \mathbb{P}\left(E_{3,k}\right) \right)\\
	& \gep{a} -4\lambda_{\Sigma} \left(  \sum_{k=1}^\infty k \frac{1}{p_0}\mathbb{P}\left(E_{3,k-1}\right)  + \sum_{k=1}^\infty k \mathbb{P}\left(E_{3,k}\right)\right)\\
	& = -4\lambda_{\Sigma} \left(  \sum_{k=1}^\infty (k-1+1) \frac{1}{p_0}\mathbb{P}\left(E_{3,k-1}\right)  + \sum_{k=1}^\infty k \mathbb{P}\left(E_{3,k}\right)\right)\\
	& = -4\lambda_{\Sigma} \left( \left( \frac{1}{p_0} + 1\right) \sum_{k=1}^\infty k \mathbb{P}\left(E_{3,k}\right) + \frac{1}{p_0}\sum_{k=1}^\infty\mathbb{P}\left(E_{3,k-1}\right) \right)\\
	& \gep{b} -8 \left( \left( \frac{1}{p_0} + 1\right) \sum_{k=1}^\infty k \mathbb{P}\left(E_{3,k}\right) + \frac{1}{p_0}  \right),
\end{split}	
\end{equation}
where inequality (a) follows from Eq. \eqref{eq:E3_E1}; inequality (b) holds since $\sum_{k=1}^\infty\mathbb{P}\left(E_{3,k-1}\right) \le 1$ and $\lambda_\Sigma \le 2$.

Now, let us reintroduce the superscript $^{(\epsilon)}$. Note that 
\begin{equation*}
	 \lim_{\epsilon \downarrow 0}\sum_{k=1}^\infty k \mathbb{P}\left(E_{3,k}^{(\epsilon)}\right) \ep{a} \lim_{\epsilon \downarrow 0}\sum_{k=1}^\infty k \mathbb{P}\left( (\overline{Q}_1^{+})^{(\epsilon)} = k,  \overline{U}_2^{(\epsilon)}>0   \right) \ep{b} 0,
\end{equation*}
where equality (a) follows from the fact that $U_2(t) > 0$ implies that $Q_2(t+1) = 0$ for all $t\ge0$ and $\epsilon > 0$; equality (b) follows directly from Eq. \eqref{eq:case2}. Taking liminf on both sides of Eq. \eqref{eq:T1}, yields 
\begin{equation}
\label{eq:T1_inf}
	\liminf_{\epsilon \downarrow 0} \mathcal{T}_1 \ge -\frac{8}{p_0}.
\end{equation}
Next, let us turn to consider the second term on the right-hand side of Eq. \eqref{eq:JIQ_important}
\begin{equation*}
\begin{split}
	\mathcal{T}_2 &= \ex{\left( \overline{A}_1^{(\epsilon)} - \overline{A}_2^{(\epsilon)}\right)^2}\\
	& \ep{a} \ex{\left( \overline{A}_1^{(\epsilon)} + \overline{A}_2^{(\epsilon)}\right)^2}\\
	& = \ex{\left(A_{\Sigma}^{(\epsilon)}\right)^2}\\
	& = \left(\sigma_{\Sigma}^{(\epsilon)}\right)^2 + \left(\lambda_{\Sigma}^{(\epsilon)}\right)^2,
\end{split}	
\end{equation*}
where (a) holds since $A_1(t)A_2(t) = 0$ for all $t\ge0$ and $\epsilon > 0$. Taking liminf of both sides, we obtain
\begin{equation}
\label{eq:T2_inf}
	\liminf_{\epsilon \downarrow 0} \mathcal{T}_2 = \sigma_{\Sigma}^2 + 4,
\end{equation}
since $\left(\sigma_{\Sigma}^{(\epsilon)}\right)^2 $ approaches to $\sigma_{\Sigma}^2$.

We are left with the third term on the right-hand side of Eq. \eqref{eq:JIQ_important}.
\begin{equation*}
\begin{split}
	\mathcal{T}_3 &= -\ex{\left( \overline{U}_1^{(\epsilon)}  -  \overline{U}_2^{(\epsilon)}  \right)^2}\\
	& \ge -\ex{ \left(\overline{U}_1^{(\epsilon)}\right)^2 + \left(\overline{U}_2^{(\epsilon)}\right)^2}\\
	& \ep{a} -\ex{\left(\overline{U}_1^{(\epsilon)}\right) + \left(\overline{U}_2^{(\epsilon)}\right)}\\
	& \ep{b} -\epsilon,
\end{split}
\end{equation*}
where (a) follows from the fact the unused service is at most one; (b) follows directly from the property of unused service shown in Eq. \eqref{eq:unused_one_norm} of Lemma \ref{lem:unused_service}. Taking liminf of both sides, yields,
\begin{equation}
\label{eq:T3_inf}
	\liminf_{\epsilon \downarrow 0} \mathcal{T}_3 \ge 0.
\end{equation}
Now, taking liminf on both sides of Eq. \eqref{eq:B} and using Eq. \eqref{eq:JIQ_important}, yields,
\begin{equation*}
\begin{split}
	\liminf_{\epsilon \downarrow 0} B^{(\epsilon)} &\ge \frac{1}{2}\liminf_{\epsilon \downarrow 0}\left(\mathcal{T}_1 + \mathcal{T}_2 + \mathcal{T}_3\right)\\
	&\gep{a} \frac{1}{2}\left(\sigma_{\Sigma}^2 + 4 - \frac{8}{p_0}\right).\\
	&\gp{b} 0.
\end{split}	
\end{equation*}
where (a) follows from the super-additivity of liminf and Eqs. \eqref{eq:T1_inf}, \eqref{eq:T2_inf} and \eqref{eq:T3_inf}; inequality (b) comes from the fact that the arrival process is in $\mathcal{A}$. Therefore, we arrive that $\limsup_{\epsilon \downarrow 0} \overline{B}^{(\epsilon)} > 0 $, which contradicts with our assumption that $\limsup_{\epsilon \downarrow 0} \overline{B}^{(\epsilon)} = 0$. Therefore, the result in Theorem \ref{THM:JIQ} must hold, i.e., JIQ is not heavy-traffic delay optimal for a two-server homogeneous settings.\qed

\section{Additional Simulation Results}
\setcounter{figure}{9}
\label{sec:addition}
	In this section, we provide additional simulation results with various parameters and system sizes. More specifically, we provide the same set of simulation results for larger system sizes, e.g., $50$ and $100$ servers, as well as different combinations of arrival and service distributions for each time slot, including Poisson arrival-constant service, Poisson arrival-bursty service and bursty arrival-Poisson service.

\subsection{Throughput Performance}
	In this subsection, we further explore the throughput performance for different system sizes and different combinations of arrival and service process. In particular, we present the results for $50$ and $100$ heterogeneous servers under Poisson arrival and Poisson service. Moreover, simulation results for $10$ heterogeneous servers under Poisson arrival-constant service, Poisson arrival-bursty service and bursty arrival-Poisson arrival are all presented. Similar trend can be observed in all these results.

\begin{figure}[t]
\graphicspath{{./Figures/}}
\begin{minipage}{3.2in}
\begin{center}
\includegraphics[width=3.2in]{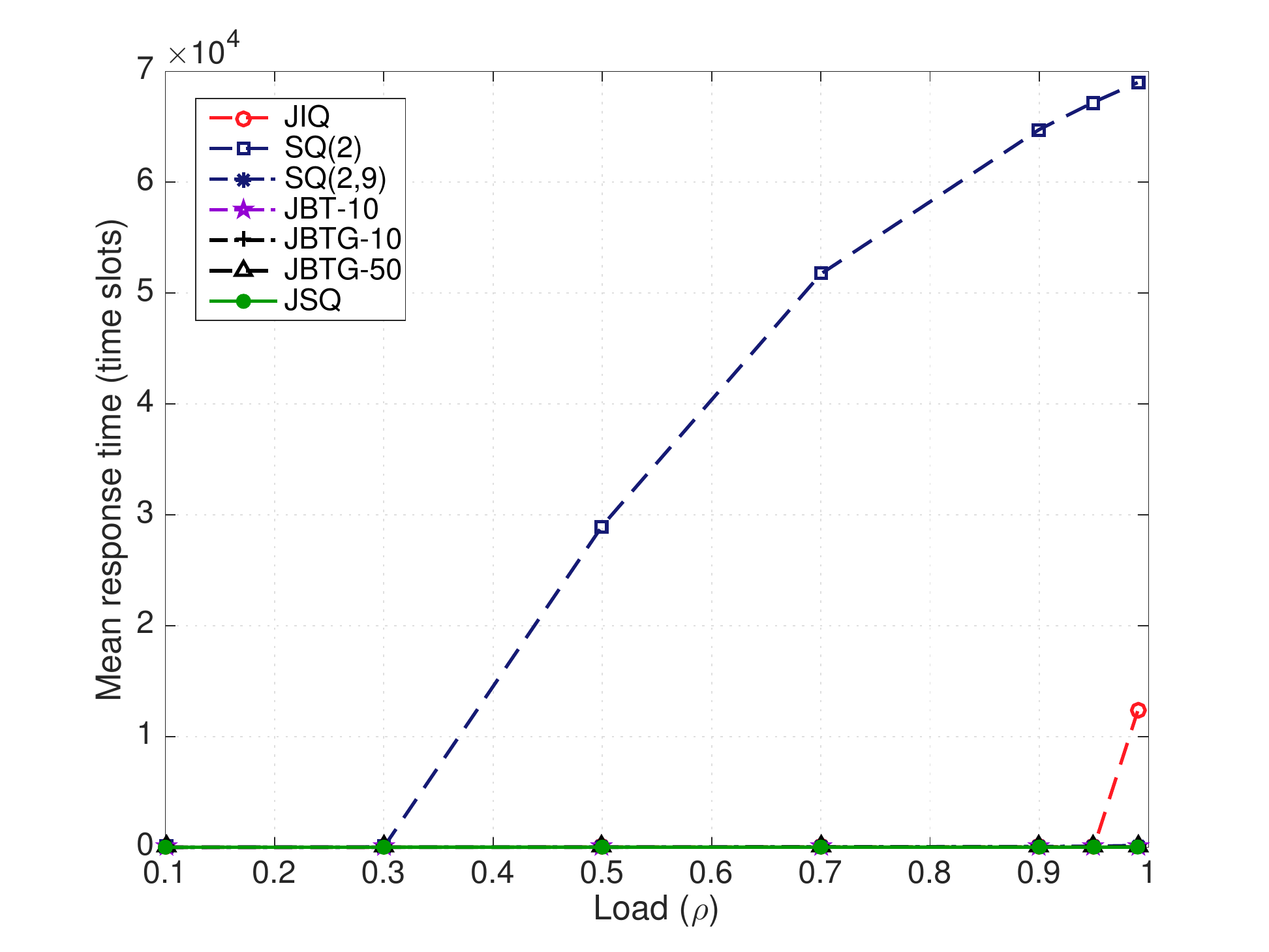}
\caption{Throughput performance in $50$ heterogeneous servers.}\label{fig:throughput_N50}
 \end{center}
\end{minipage}
\hfill
\begin{minipage}{3.2in}
\begin{center}
\includegraphics[width=3.2in]{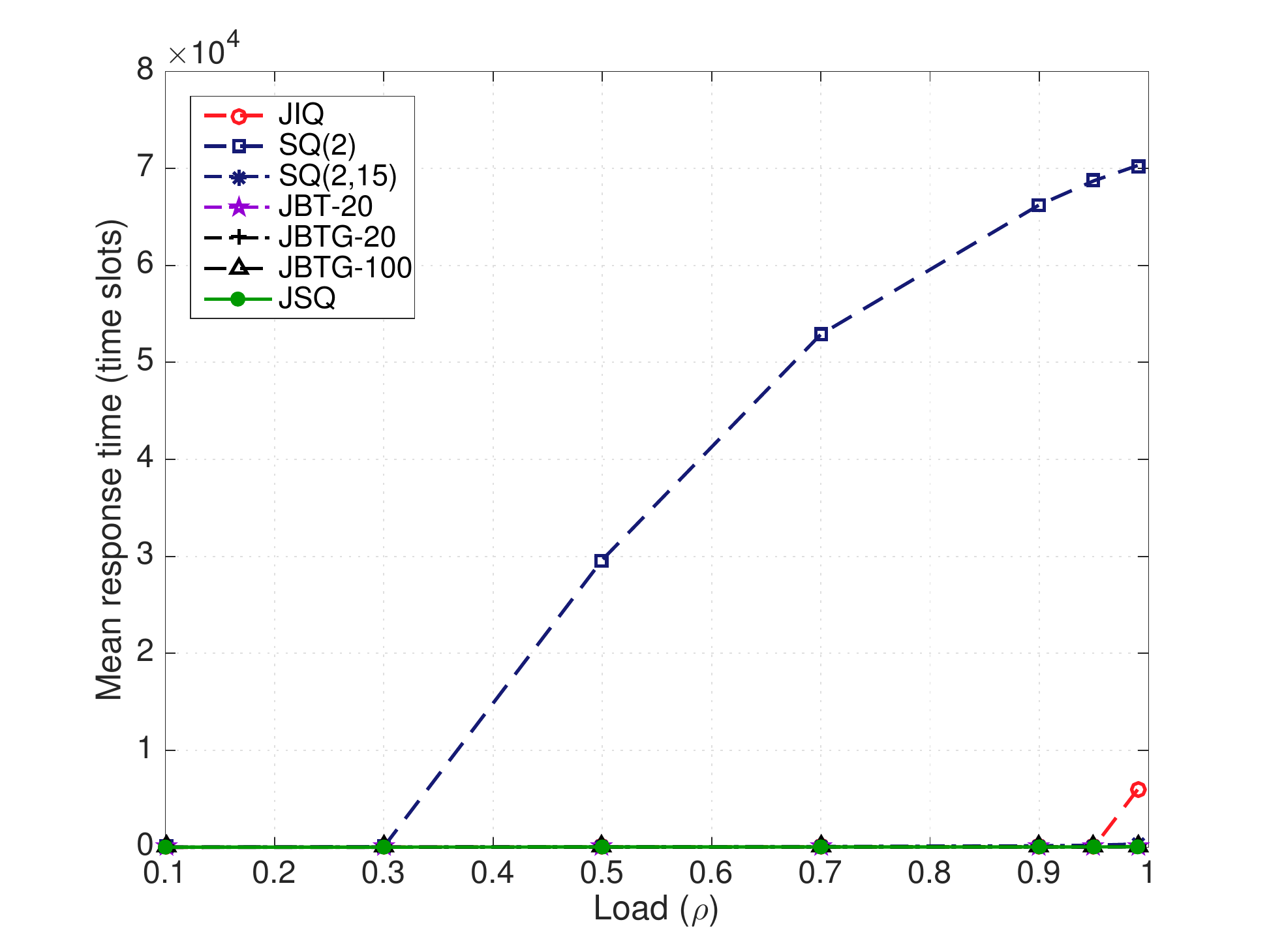}
\caption{Throughput performance in $100$ heterogeneous servers.}\label{fig:throughput_N100}
 \end{center}
\end{minipage}
\end{figure}



	Let us first look at the throughput performance under larger system sizes. Figures \ref{fig:throughput_N50} and \ref{fig:throughput_N100} demonstrate the throughput performance under $50$ and $100$ heterogeneous servers, respectively. As before, the servers are equally divided into two server pools with rate $1$ and $10$. A turning point in each curve indicates that the load approaches the throughput region boundary of the corresponding policy. It can be seen that power-of-$d$ policy has nearly the same throughput region as in the case of $10$ servers. For JIQ policy, it remains unstable for heavy loads, yet it tends to have a larger throughput region than that in the case of $10$ servers, which is intuitive since it is more likely for an arrival to find an idle queue in a  larger system. As expected, our proposed JBTG-$d$ policy is able to stabilize the system for all the considered loads which again agrees with our theoretical results. The JBT-$d$ policy can also ensure stability in these two cases as in the case of $10$ servers. The use of memory for improving stability of power-of-$d$ is verified by the power-of-$d$ with memory policies for both cases (SQ($2$,$9$),SQ($2$,$15$)).

\begin{figure}[t]
\graphicspath{{./Figures/}}
\begin{minipage}{3.2in}
\begin{center}
\includegraphics[width=3.2in]{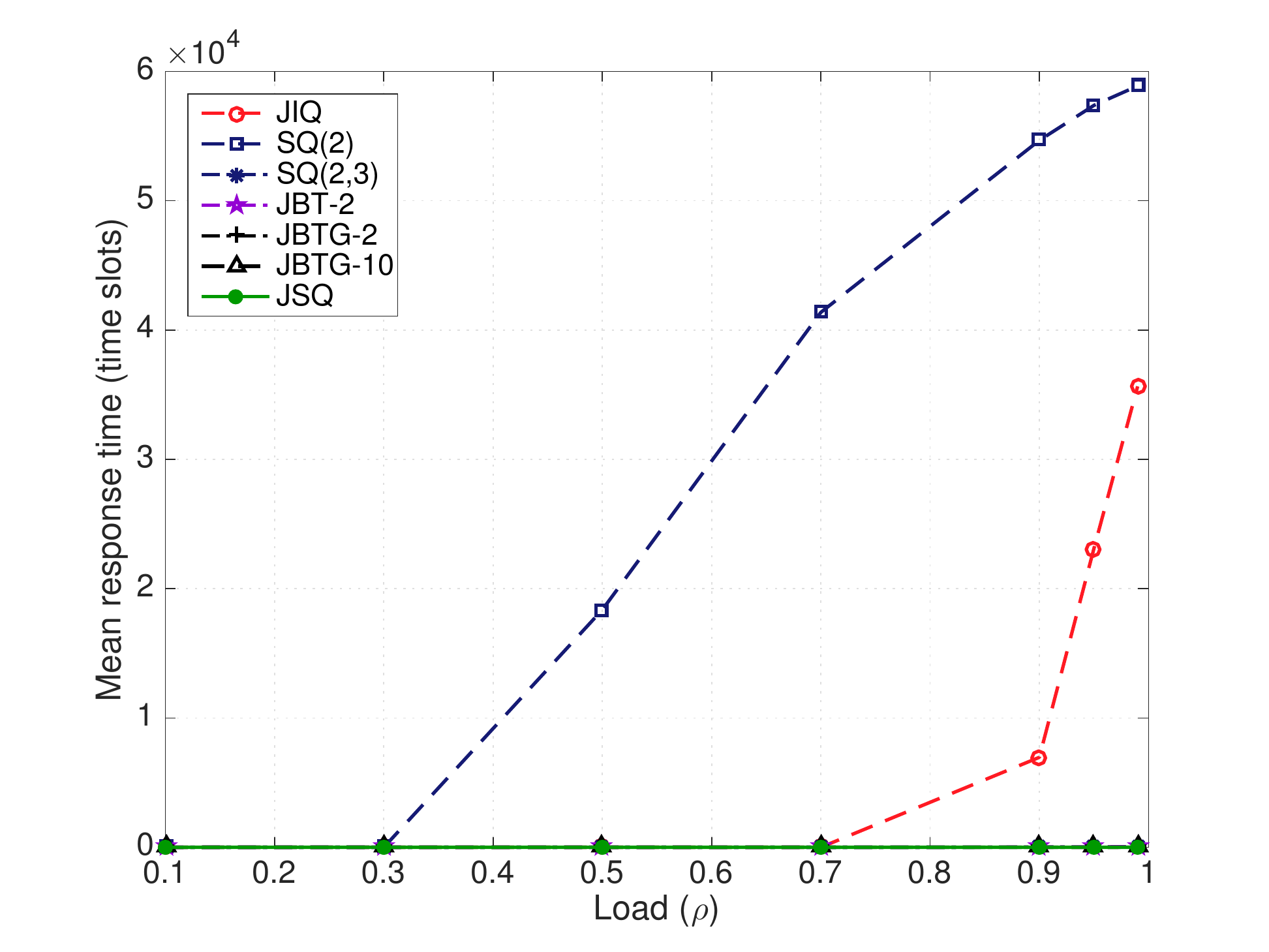}
\caption{Throughput performance in $10$ heterogeneous servers with Poisson arrival and constant service.}\label{fig:throughput_poiss_const}
 \end{center}
\end{minipage}
\hfill
\begin{minipage}{3.2in}
\begin{center}
\includegraphics[width=3.2in]{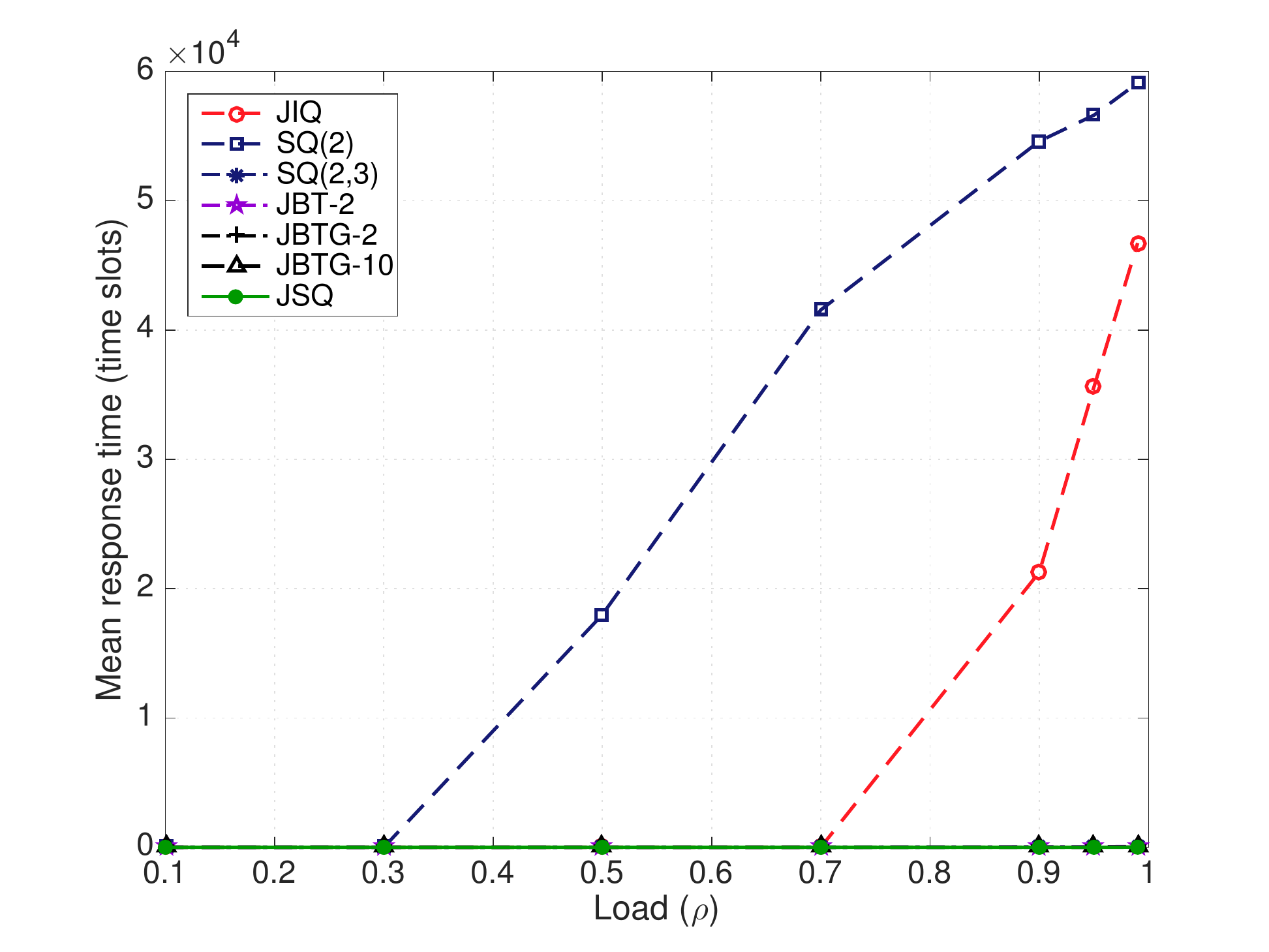}
\caption{Throughput performance in $10$ heterogeneous servers with Poisson arrival and bursty service.}\label{fig:throughput_poiss_burst}
 \end{center}
\end{minipage}
\end{figure}



\begin{figure}[t]
\graphicspath{{./Figures/}}
\begin{minipage}{3.2in}
\begin{center}
\includegraphics[width=3.2in]{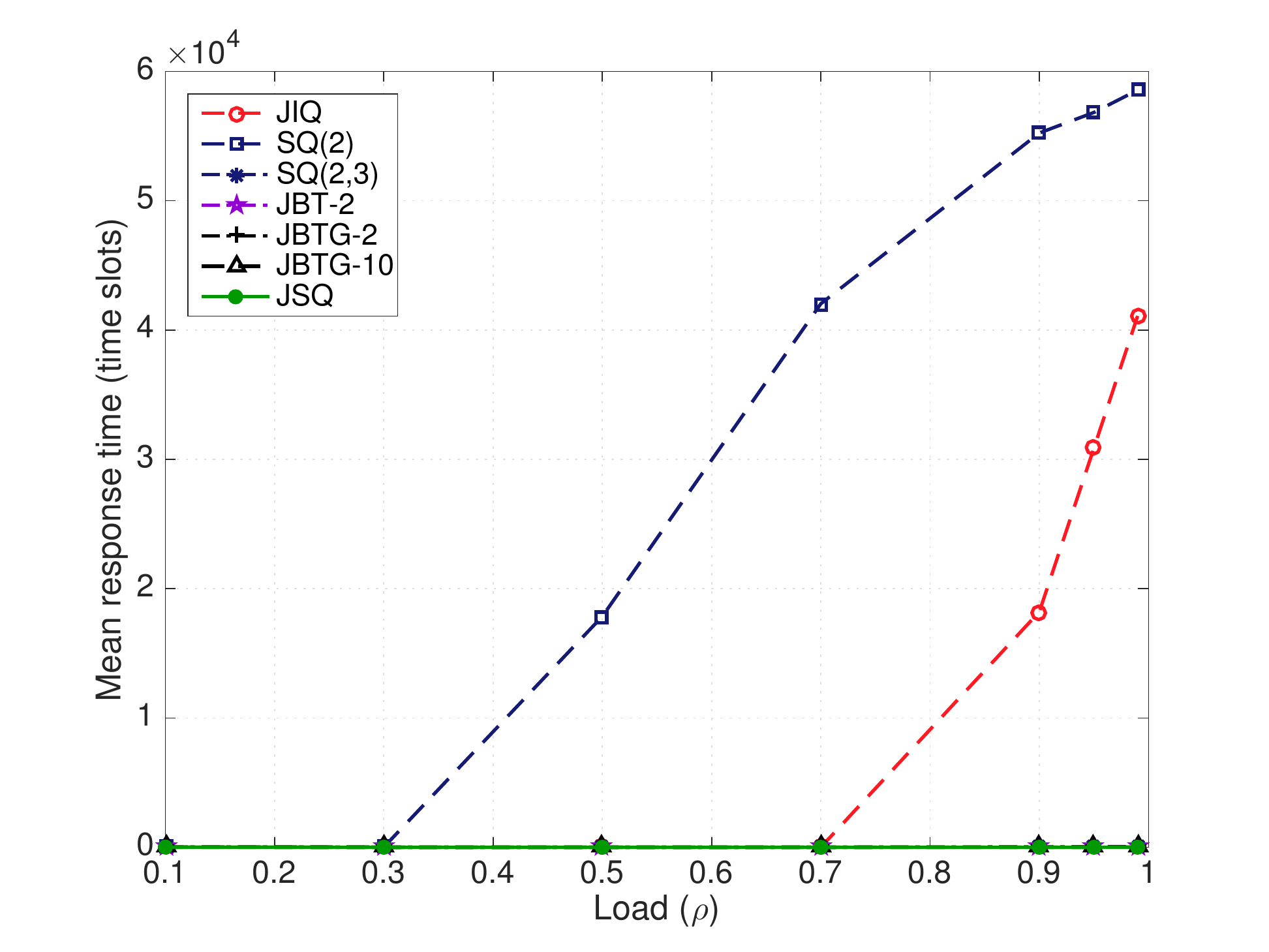}
\caption{Throughput performance in $10$ heterogeneous servers with bursty arrival and Poisson service.}\label{fig:throughput_burst_poiss}
 \end{center}
\end{minipage}
\hfill
\begin{minipage}{3.2in}
\begin{center}
\includegraphics[width=3.2in]{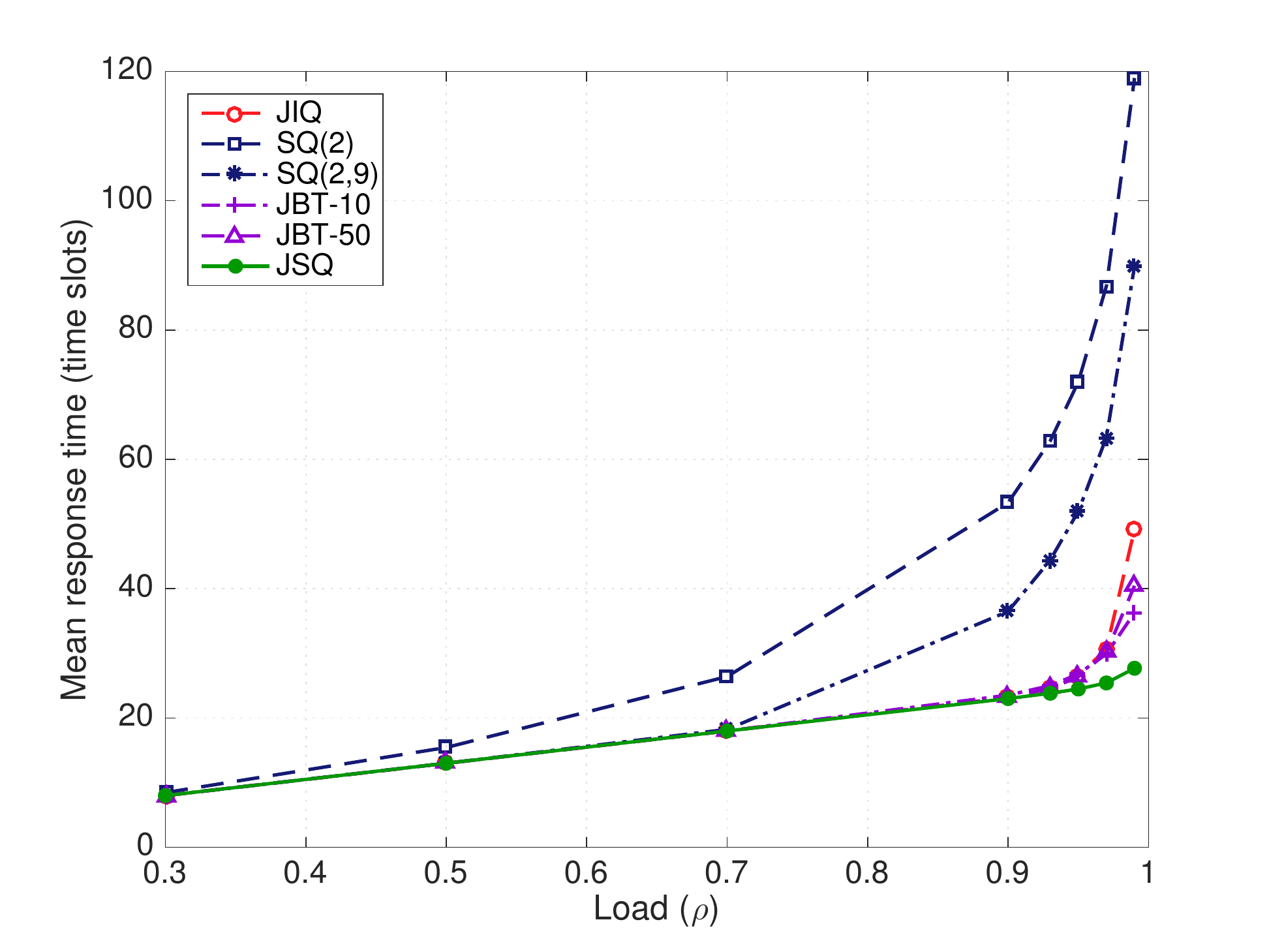}
\caption{Delay performance under $50$ homogeneous servers.}\label{fig:delay_N50}
 \end{center}
\end{minipage}
\end{figure}



\begin{figure}[t]
\graphicspath{{./Figures/}}
\begin{minipage}{3.2in}
\begin{center}
\includegraphics[width=3.2in]{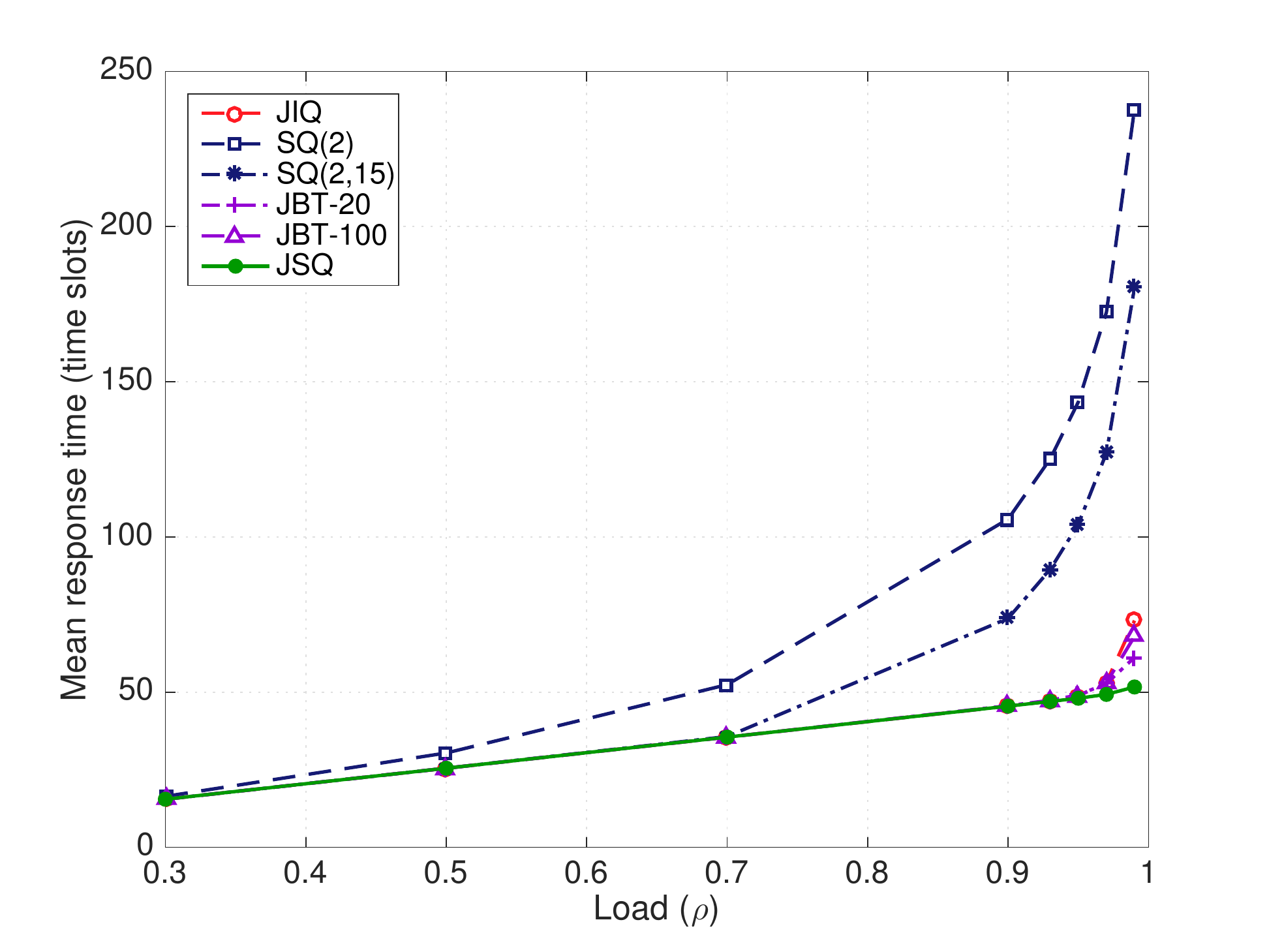}
\caption{Delay performance under $100$ homogeneous servers.}\label{fig:delay_N100}
 \end{center}
\end{minipage}
\hfill
\begin{minipage}{3.2in}
\begin{center}
\includegraphics[width=3.2in]{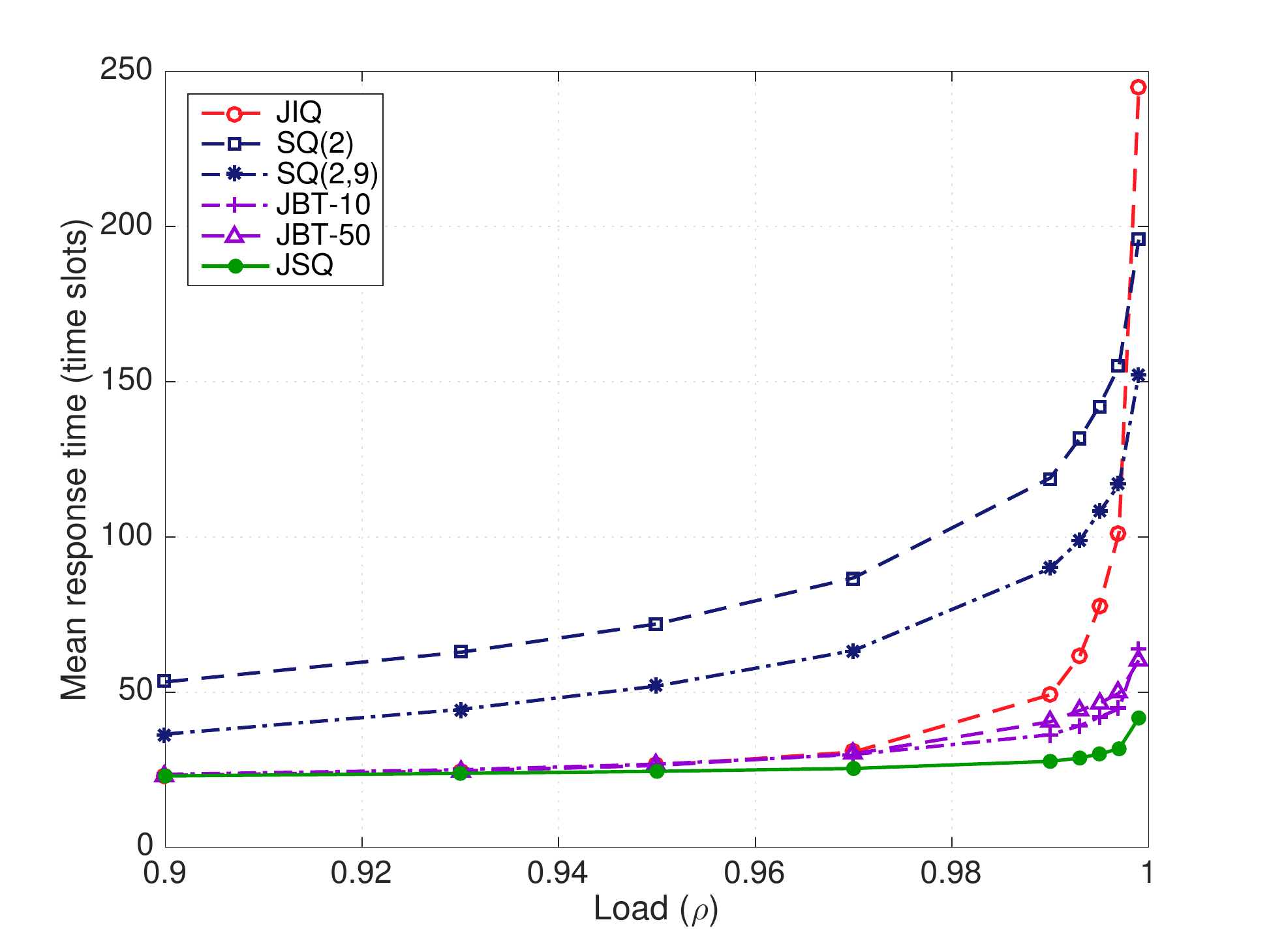}
\caption{Heavy-traffic Delay performance under $50$ homogeneous servers.}\label{fig:heavydelay_N50}
 \end{center}
\end{minipage}
\end{figure}



\begin{figure}[t]
\graphicspath{{./Figures/}}
\begin{minipage}{3.2in}
\begin{center}
\includegraphics[width=3.2in]{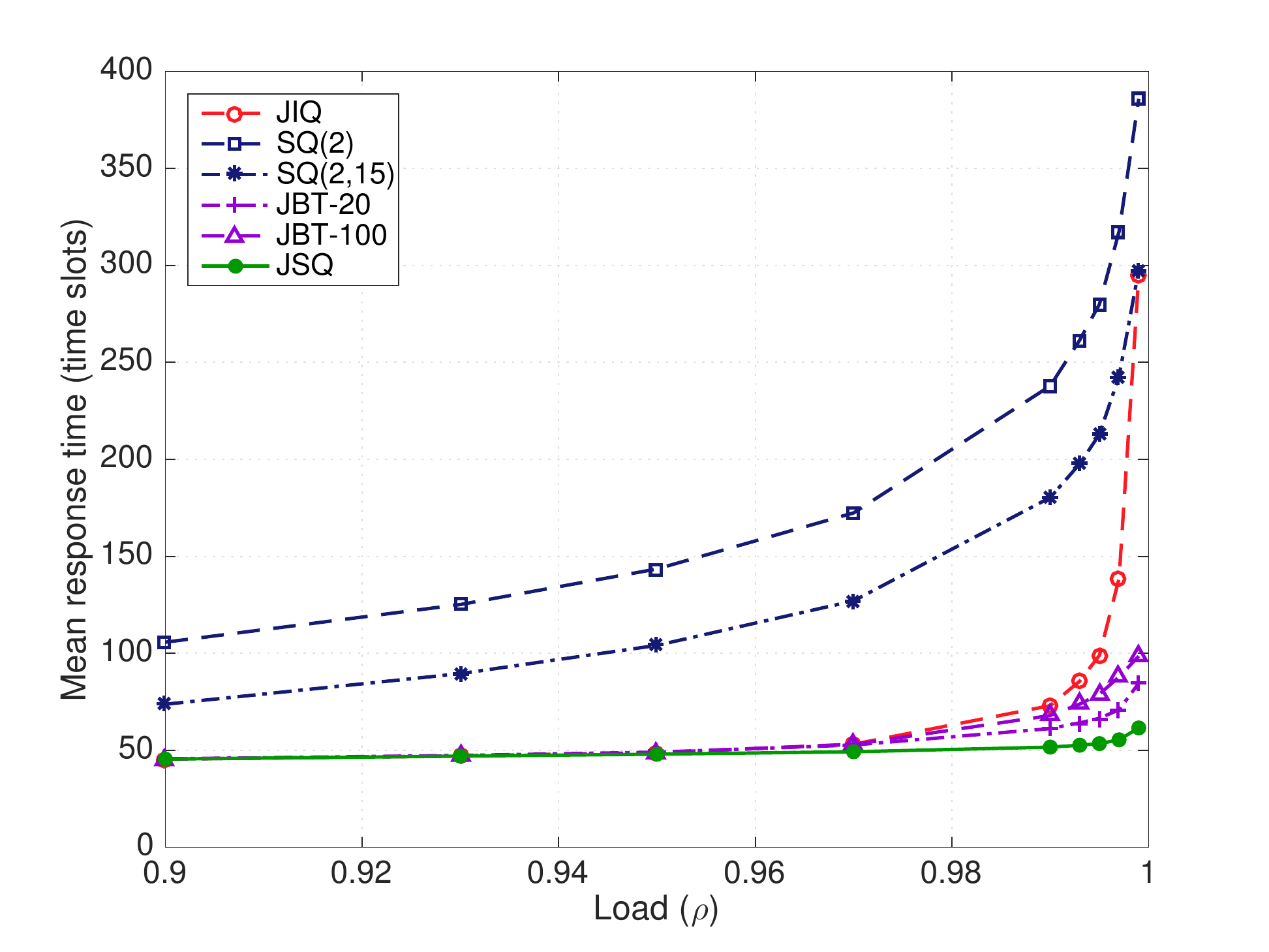}
\caption{Heavy-traffic Delay performance under $100$ homogeneous servers.}\label{fig:heavydelay_N100}
 \end{center}
\end{minipage}
\hfill
\begin{minipage}{3.2in}
\begin{center}
\includegraphics[width=3.2in]{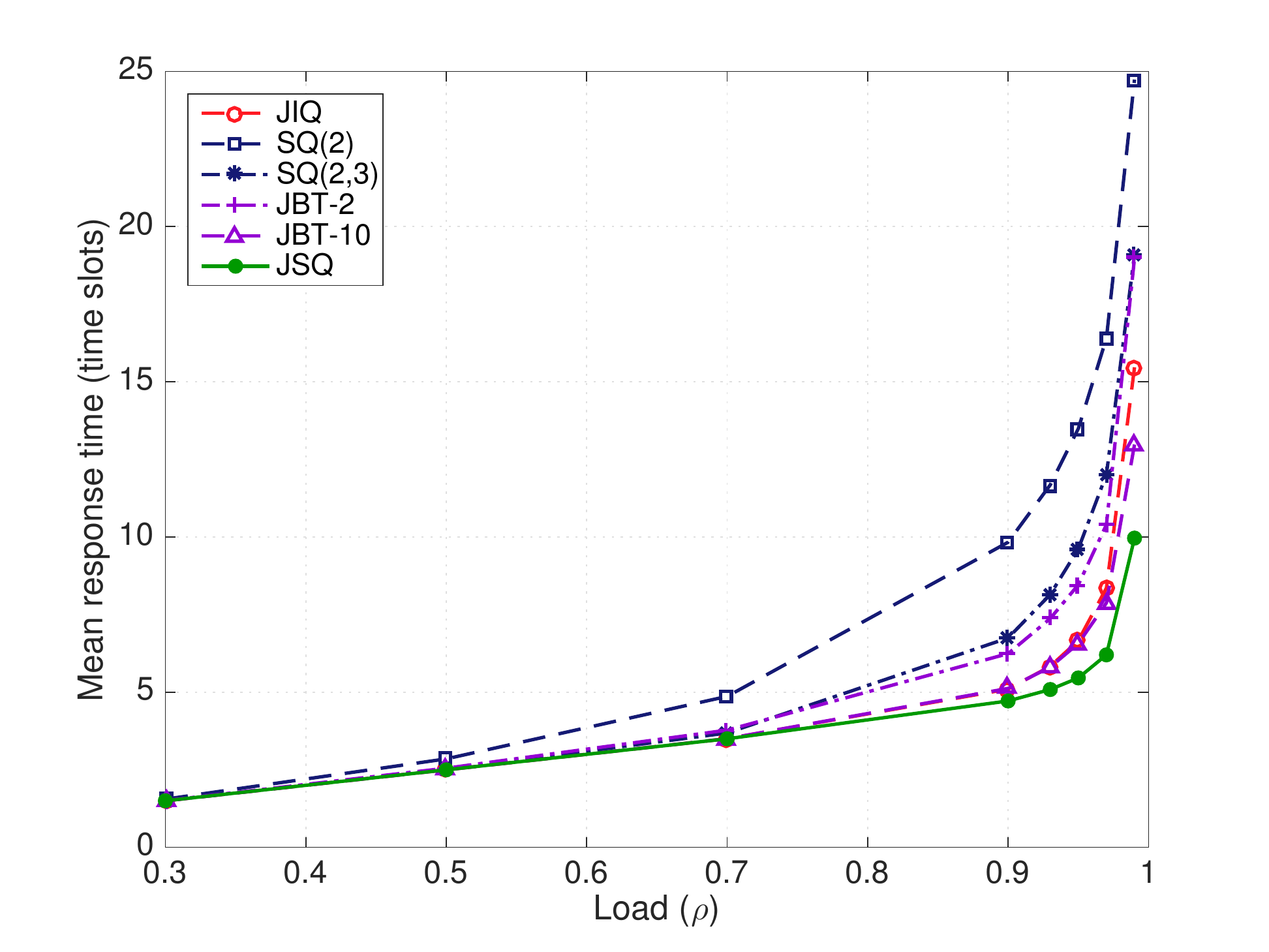}
\caption{Delay performance under $10$ homogeneous servers with Poisson arrival and constant service.}\label{fig:delay_poiss_const}
 \end{center}
\end{minipage}
\end{figure}



\begin{figure}[t]
\graphicspath{{./Figures/}}
\begin{minipage}{3.2in}
\begin{center}
\includegraphics[width=3.2in]{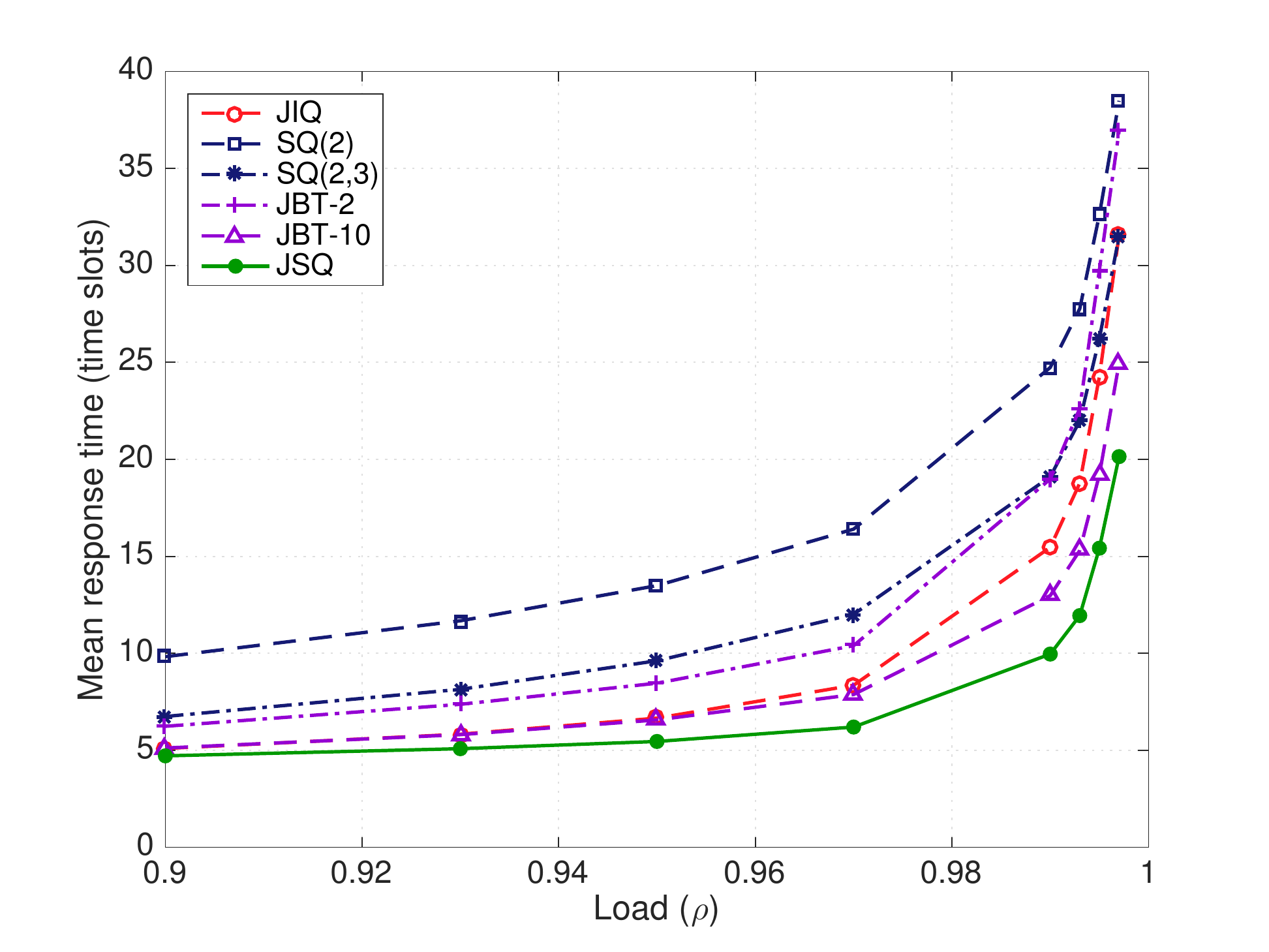}
\caption{Heavy-traffic delay performance under $10$ homogeneous servers with Poisson arrival and constant service.}\label{fig:heavydelay_poiss_const}
 \end{center}
\end{minipage}
\hfill
\begin{minipage}{3.2in}
\begin{center}
\includegraphics[width=3.2in]{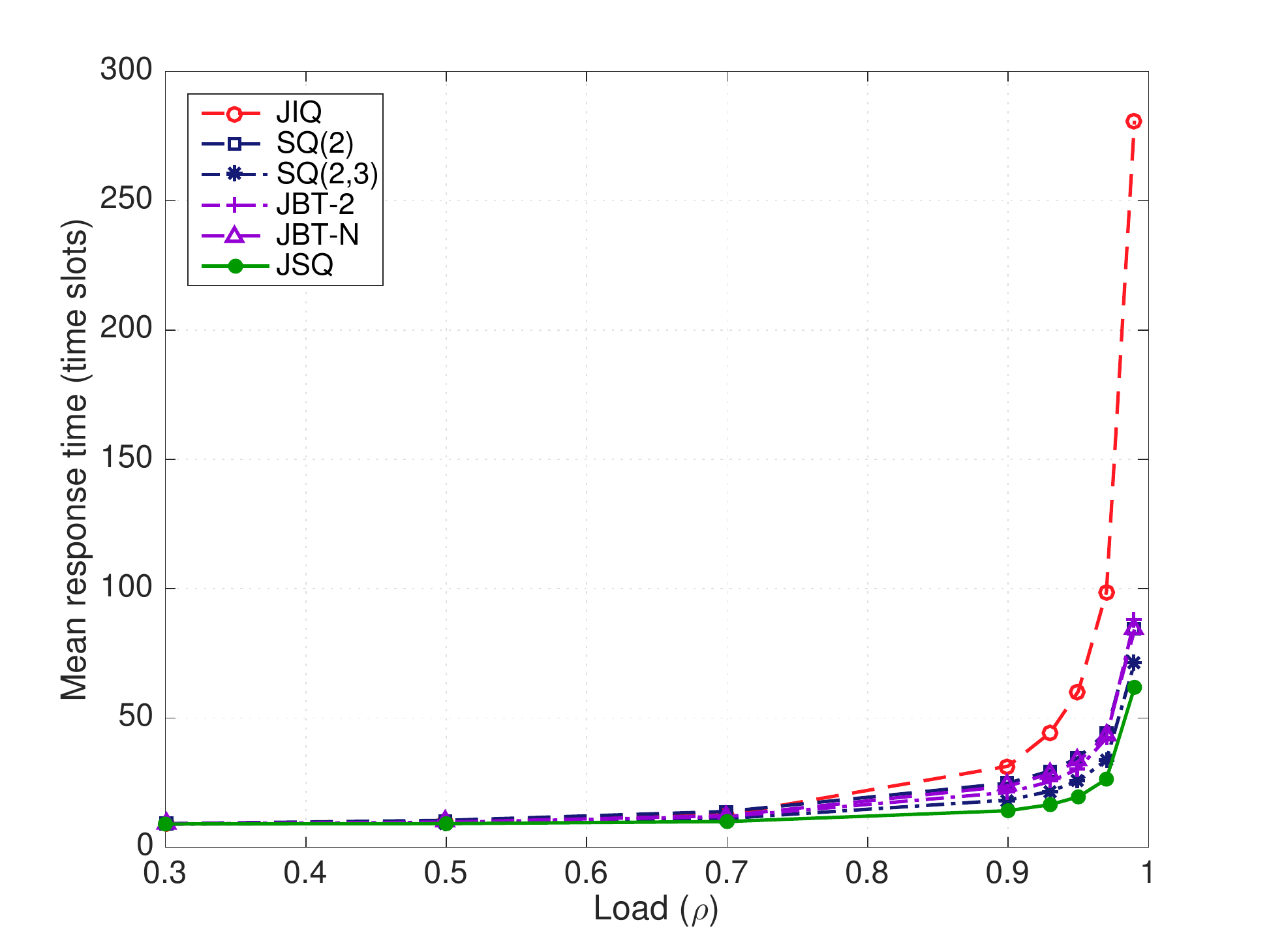}
\caption{Delay performance under $10$ homogeneous servers with Poisson arrival and bursty service.}\label{fig:delay_poiss_burst}
 \end{center}
\end{minipage}
\end{figure}



\begin{figure}[t]
\graphicspath{{./Figures/}}
\begin{minipage}{3.2in}
\begin{center}
\includegraphics[width=3.2in]{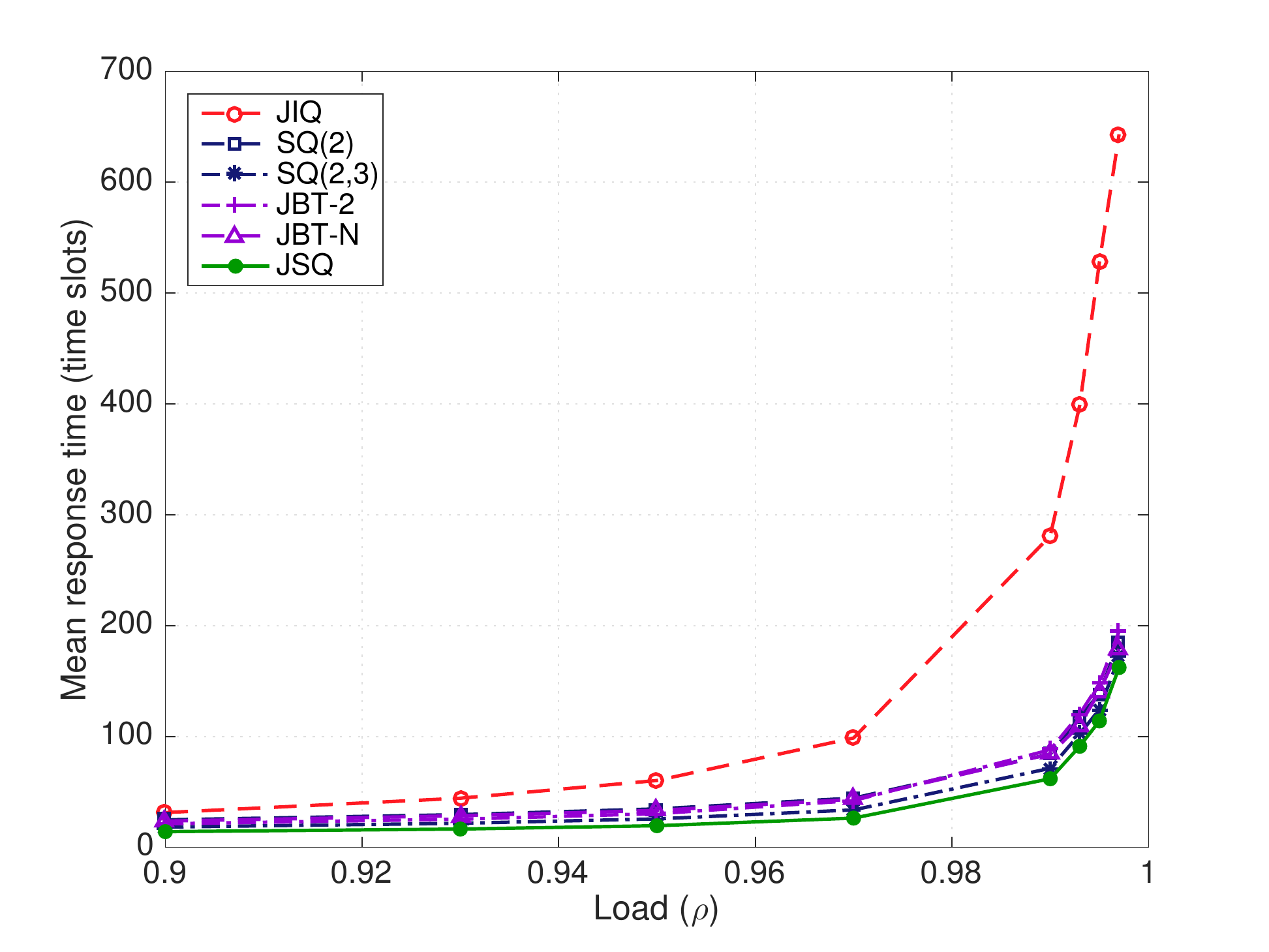}
\caption{Heavy-traffic delay performance under $10$ homogeneous servers with Poisson arrival and bursty service.}\label{fig:heavydelay_poiss_burst}
 \end{center}
\end{minipage}
\hfill
\begin{minipage}{3.2in}
\begin{center}
\includegraphics[width=3.2in]{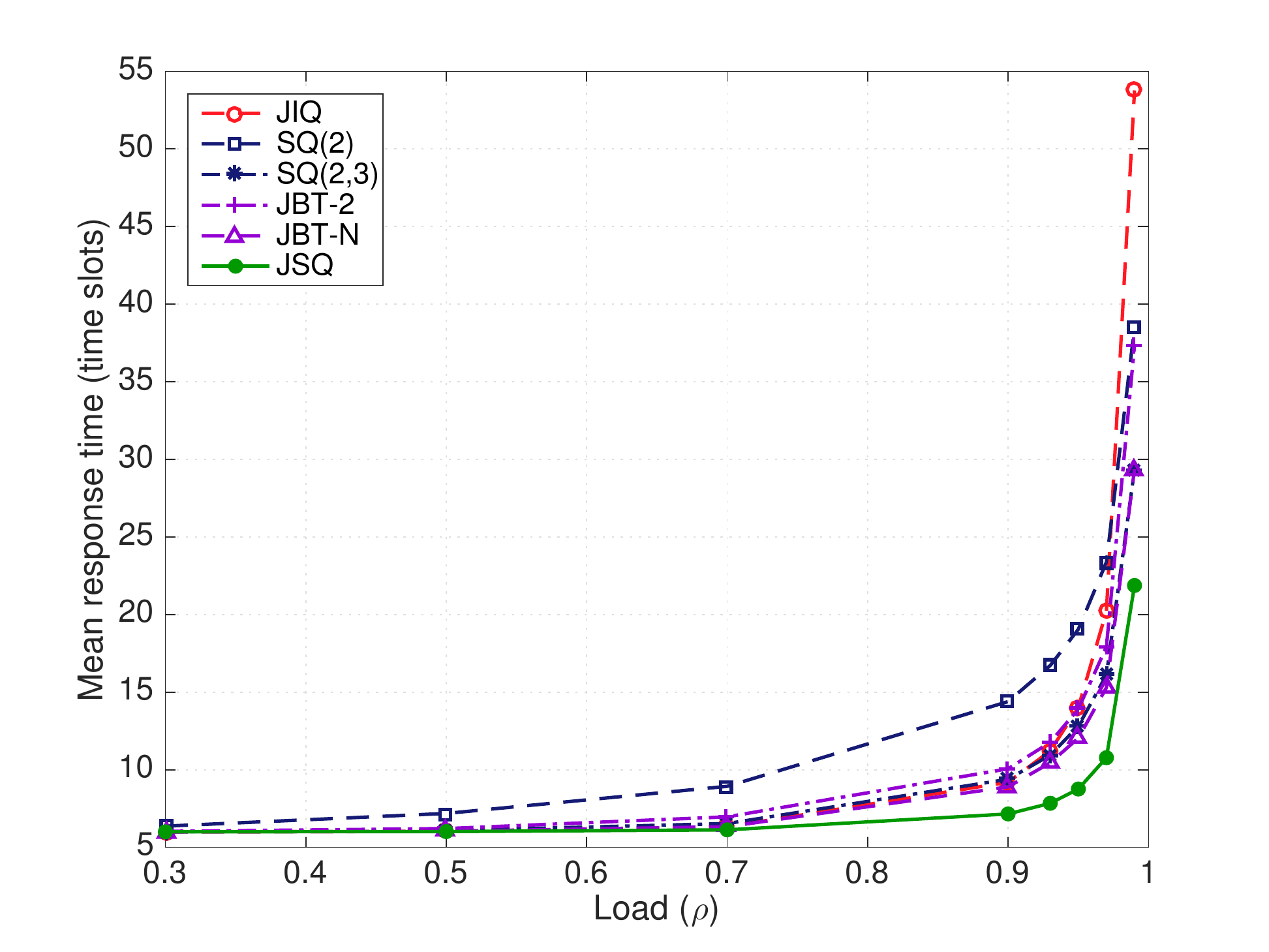}
\caption{Delay performance under $10$ homogeneous servers with bursty arrival and Poisson service.}\label{fig:delay_burst_poiss}
 \end{center}
\end{minipage}
\end{figure}



\begin{figure}[t]
\graphicspath{{./Figures/}}
\begin{minipage}{3.2in}
\begin{center}
\includegraphics[width=3.2in]{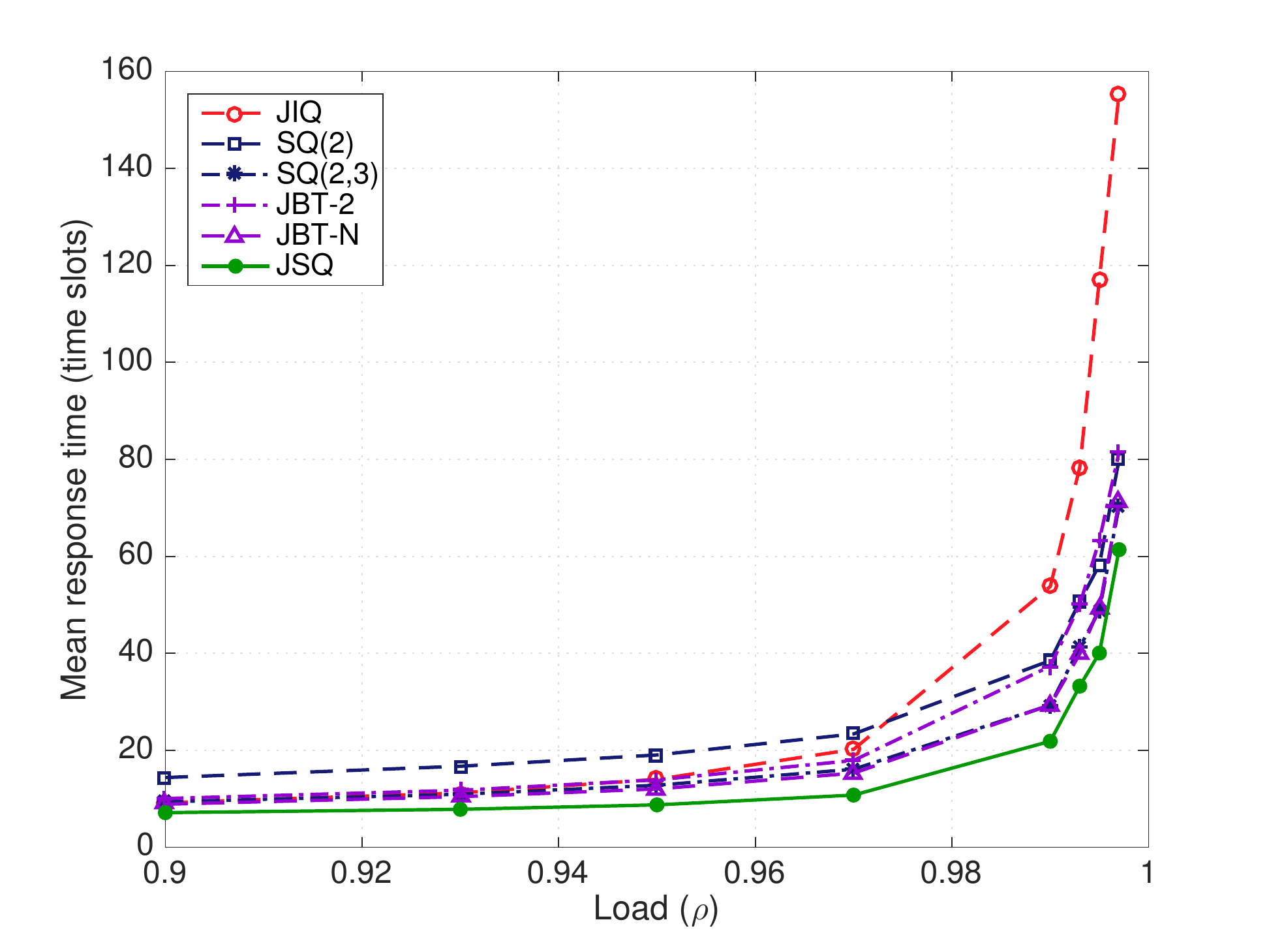}
\caption{Heavy-traffic delay performance under $10$ homogeneous servers with bursty arrival and Poisson service.}\label{fig:heavydelay_burst_poiss}
 \end{center}
\end{minipage}
\hfill
\begin{minipage}{3.2in}
\begin{center}
\includegraphics[width=3.2in]{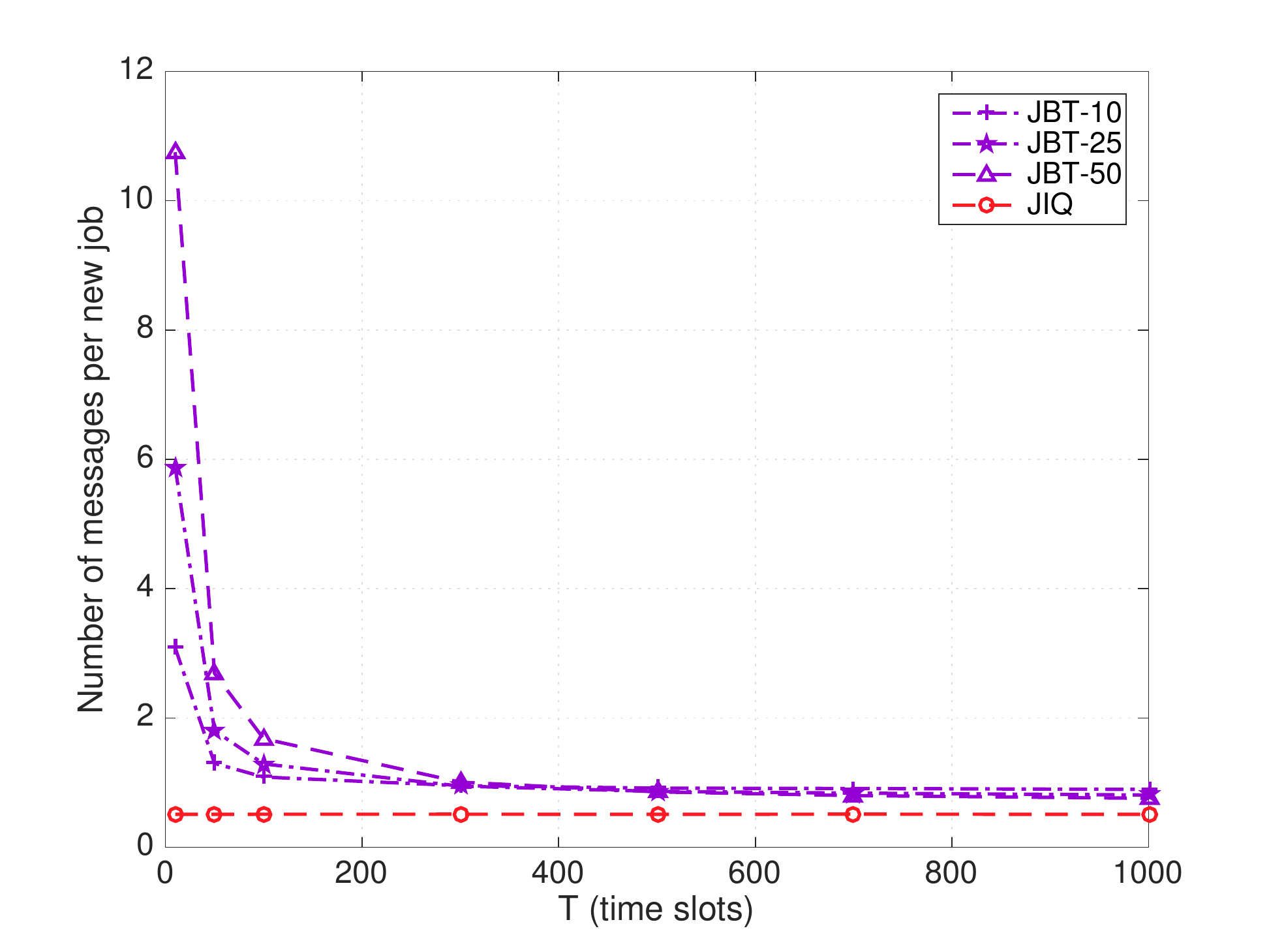}
\caption{Message per new job arrival under $50$ homogeneous servers with respect to $T$.}\label{fig:mess_N50_Td}
 \end{center}
\end{minipage}
\end{figure}



\begin{figure}[t]
\graphicspath{{./Figures/}}
\begin{minipage}{3.2in}
\begin{center}
\includegraphics[width=3.2in]{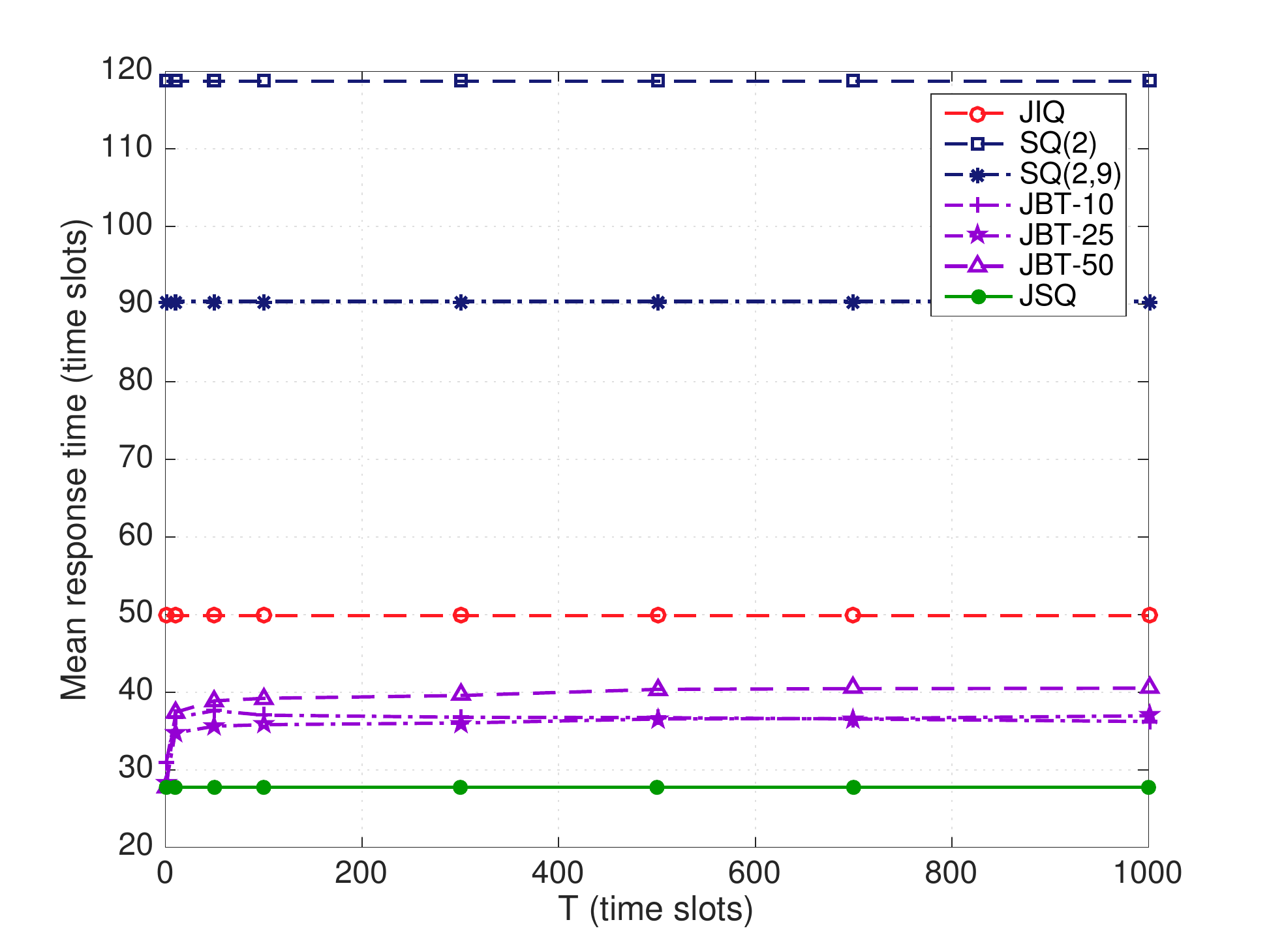}
\caption{Delay performance under $50$ homogeneous servers with respect to $T$.}\label{fig:delay_N50_Td}
 \end{center}
\end{minipage}
\hfill
\begin{minipage}{3.2in}
\begin{center}
\includegraphics[width=3.2in]{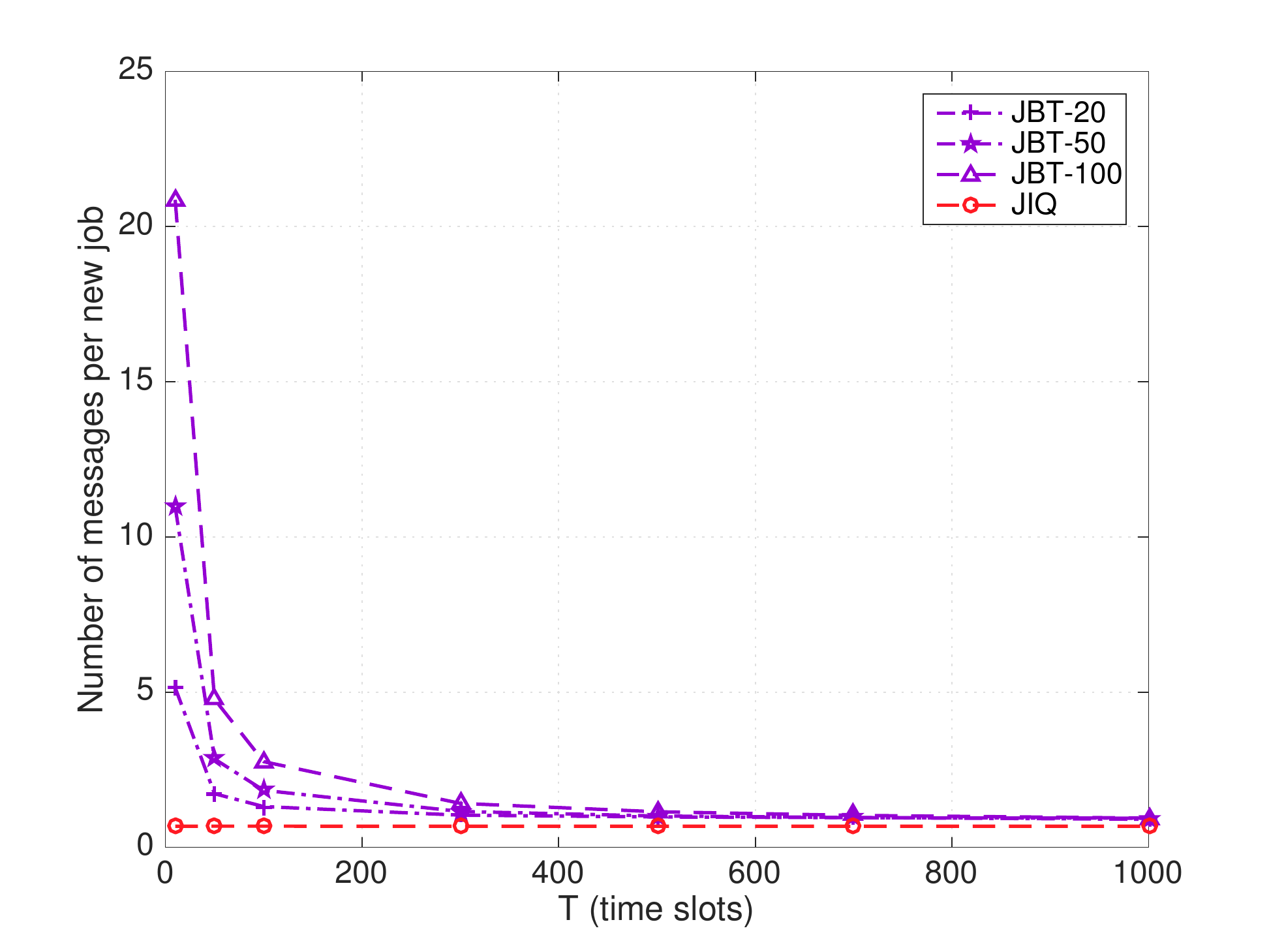}
\caption{Message per new job arrival under $100$ homogeneous servers with respect to $T$.}\label{fig:mess_N100_Td}
 \end{center}
\end{minipage}
\end{figure}



\begin{figure}[t]
\graphicspath{{./Figures/}}
\begin{minipage}{3.2in}
\begin{center}
\includegraphics[width=3.2in]{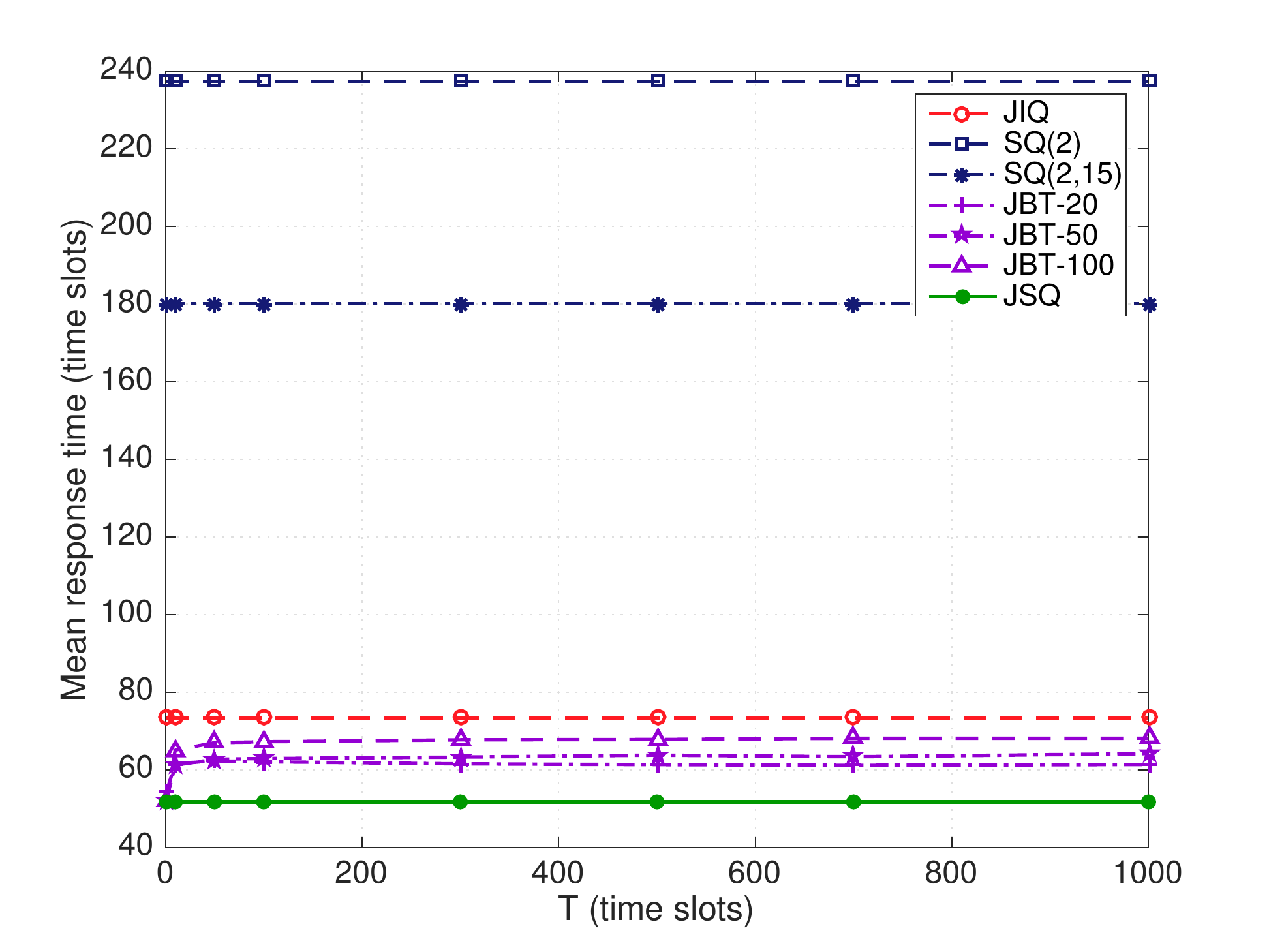}
\caption{Delay performance under $100$ homogeneous servers with respect to $T$.}\label{fig:delay_N100_Td}
 \end{center}
\end{minipage}
\hfill
\begin{minipage}{3.2in}
\begin{center}
\includegraphics[width=3.2in]{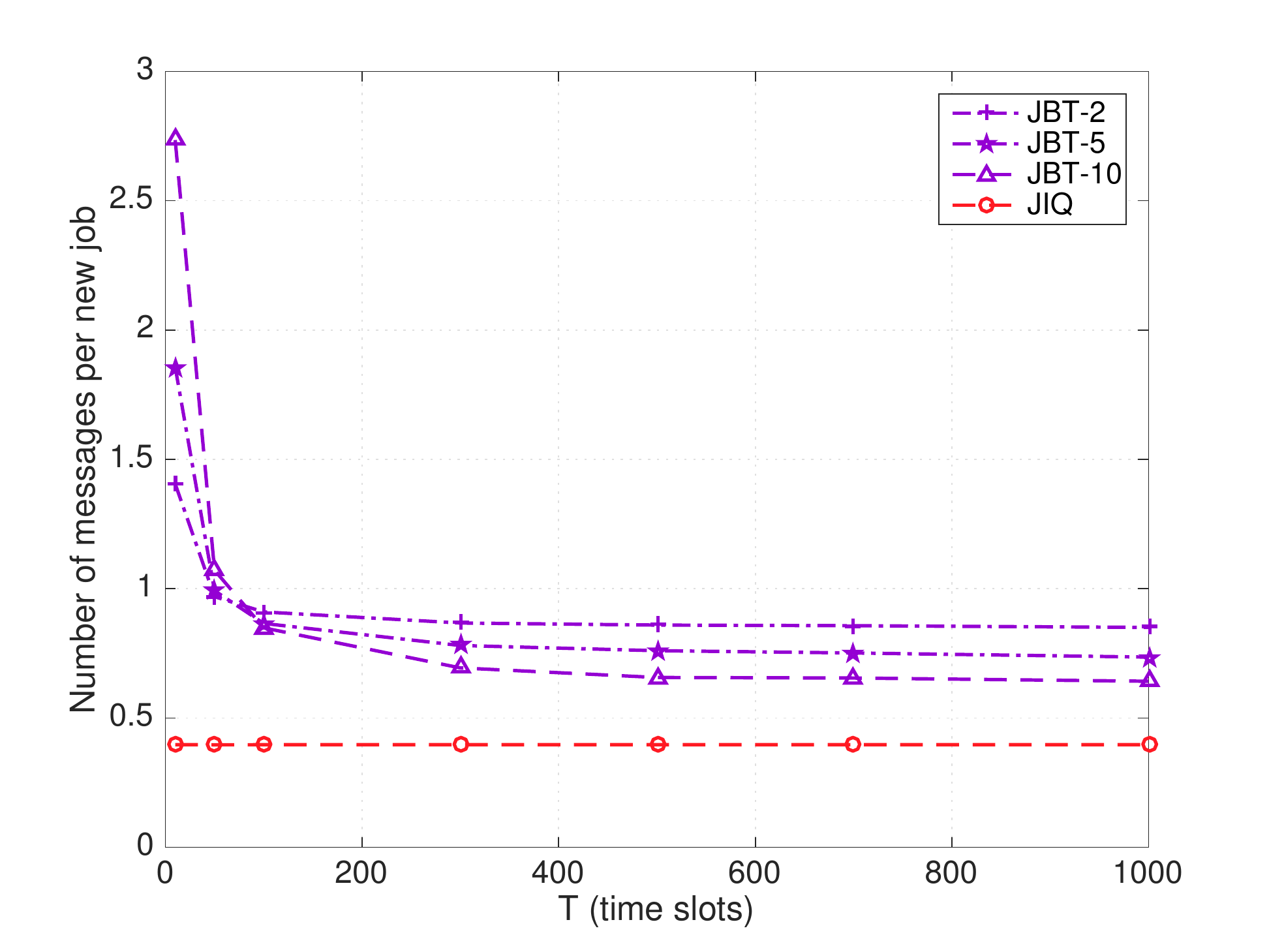}
\caption{Message per new job arrival under $10$ homogeneous servers with Poisson arrival and constant service versus $T$.}\label{fig:mess_N10_Td_poiss_const}
 \end{center}
\end{minipage}
\end{figure}



\begin{figure}[t]
\graphicspath{{./Figures/}}
\begin{minipage}{3.2in}
\begin{center}
\includegraphics[width=3.2in]{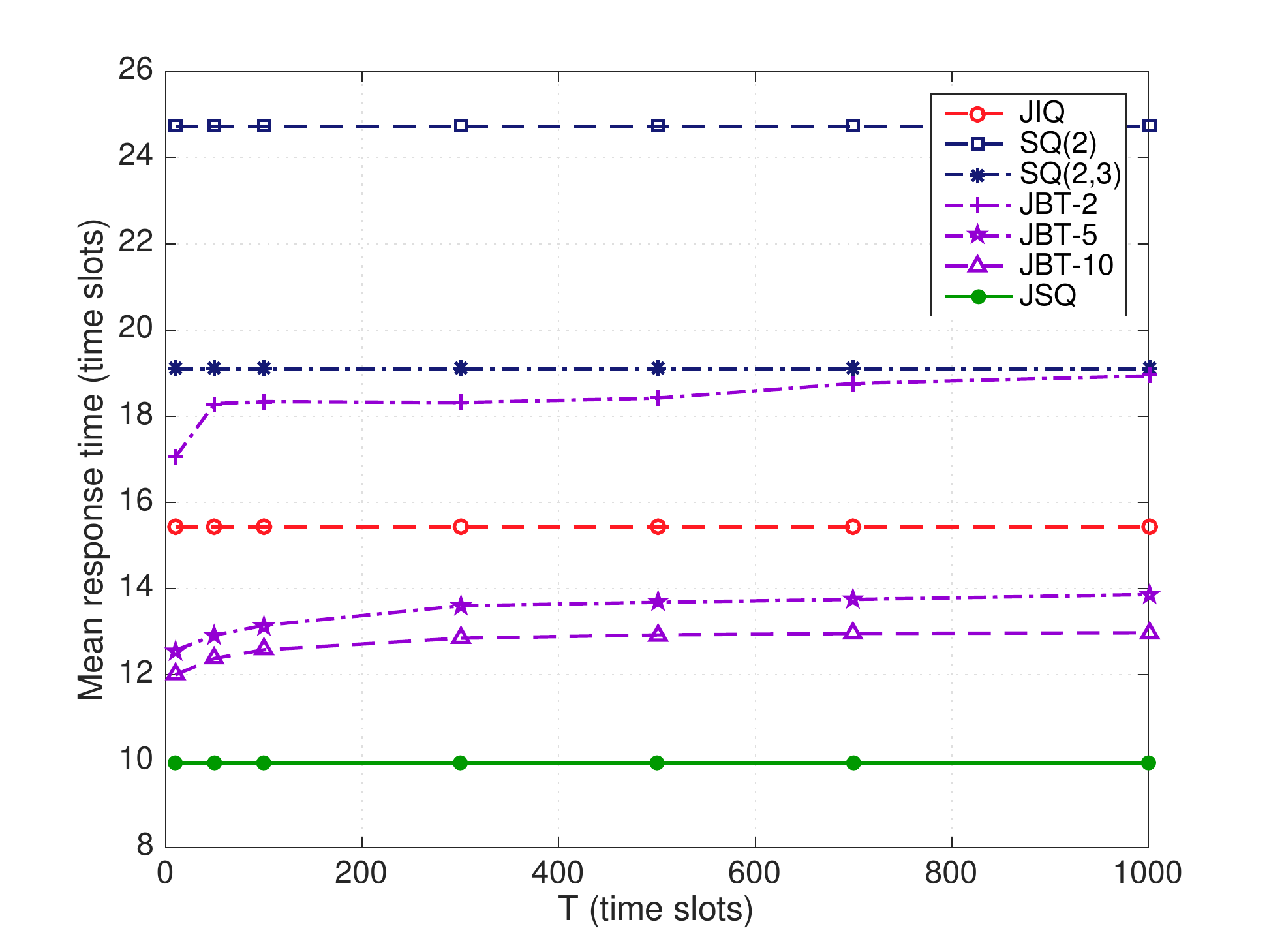}
\caption{Delay performance under $10$ homogeneous servers with Poisson arrival and constant service versus $T$.}\label{fig:delay_N10_Td_poiss_const}
 \end{center}
\end{minipage}
\hfill
\begin{minipage}{3.2in}
\begin{center}
\includegraphics[width=3.2in]{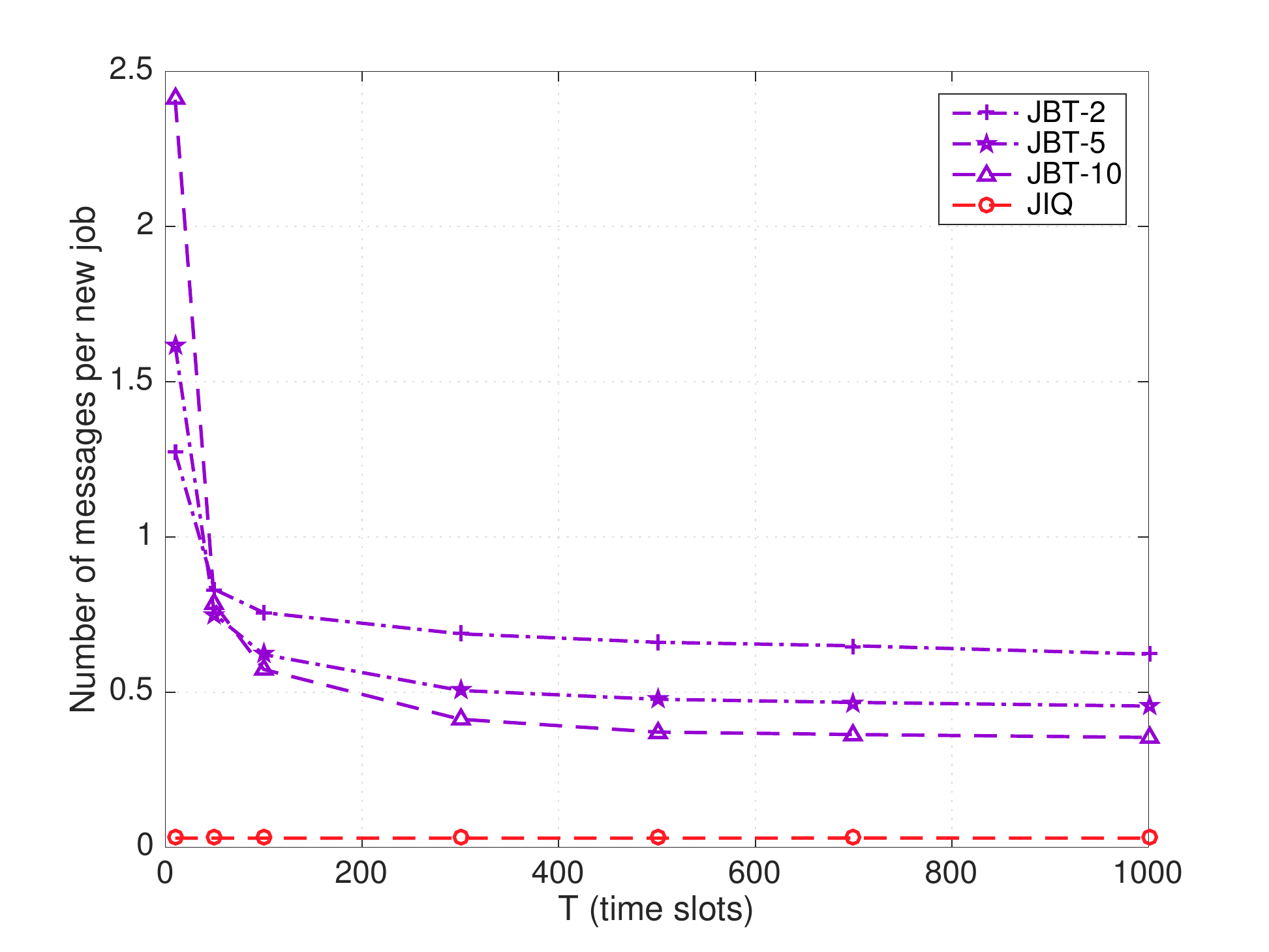}
\caption{Message per new job arrival under $10$ homogeneous servers with Poisson arrival and bursty service versus $T$.}\label{fig:mess_N10_Td_poiss_burst}
 \end{center}
\end{minipage}
\end{figure}



\begin{figure}[t]
\graphicspath{{./Figures/}}
\begin{minipage}{3.2in}
\begin{center}
\includegraphics[width=3.2in]{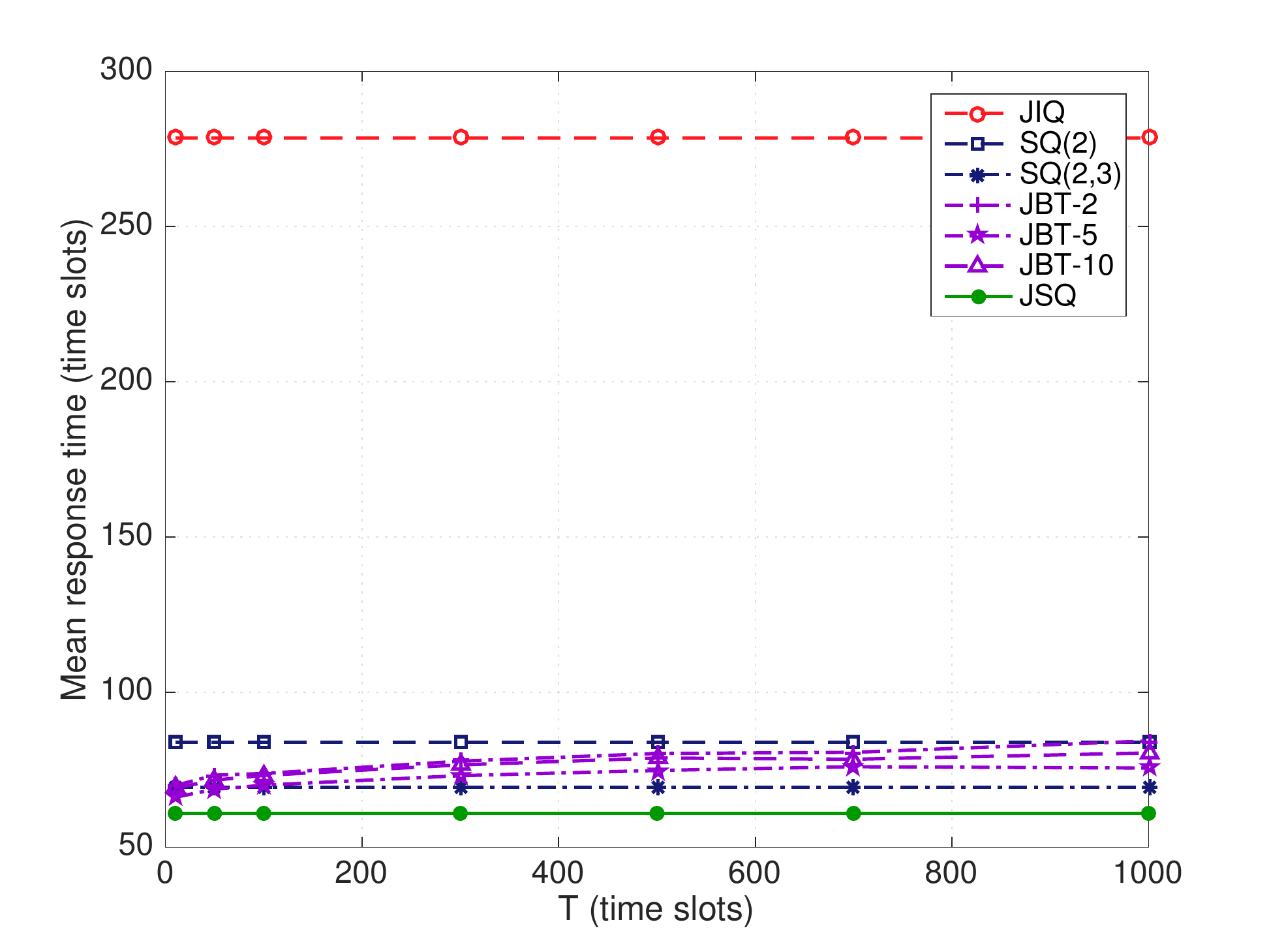}
\caption{Delay performance under $10$ homogeneous servers with Poisson arrival and bursty service versus $T$.}\label{fig:delay_N10_Td_poiss_burst}
 \end{center}
\end{minipage}
\hfill
\begin{minipage}{3.2in}
\begin{center}
\includegraphics[width=3.2in]{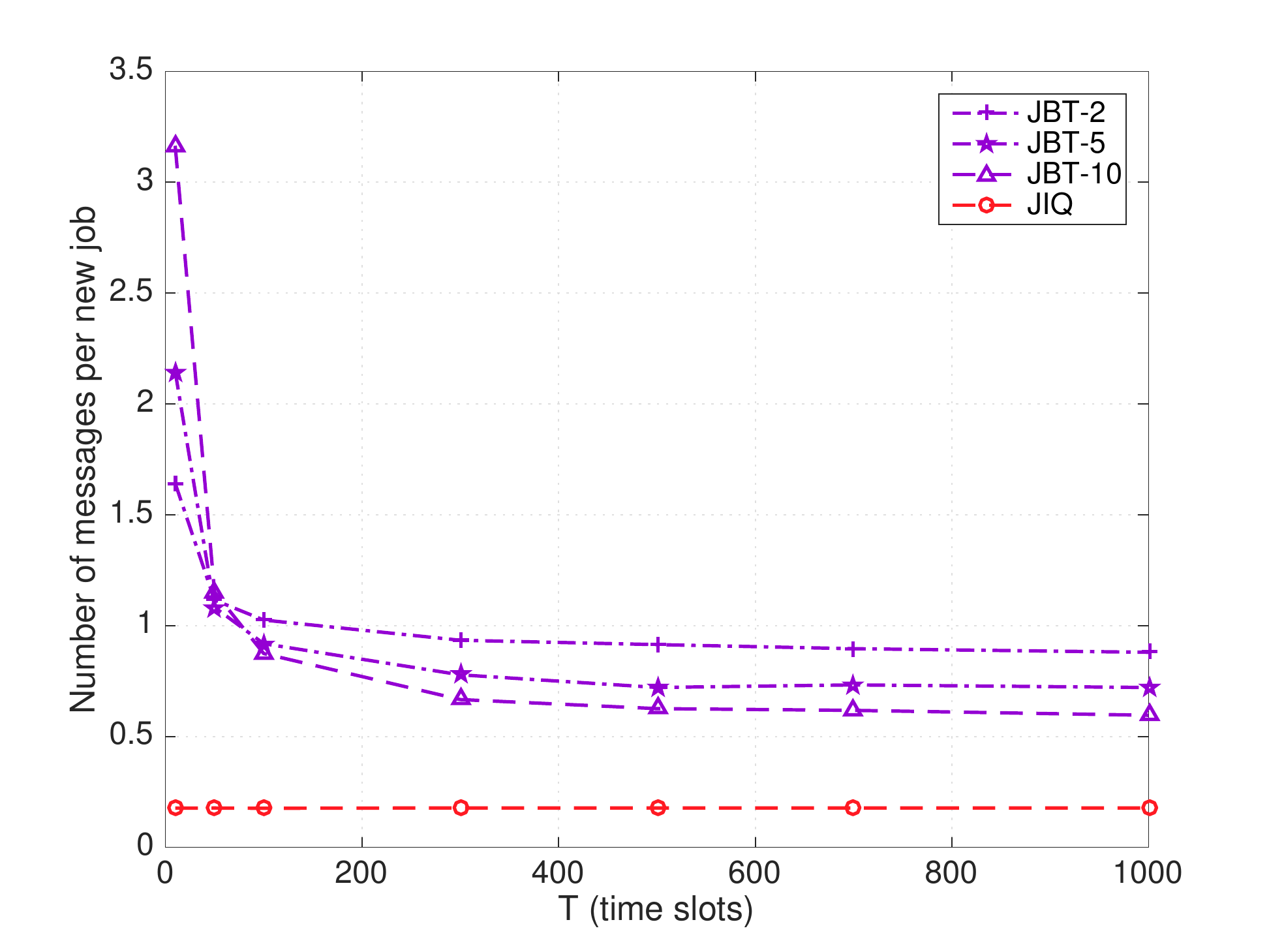}
\caption{Message per new job arrival under $10$ homogeneous servers with bursty arrival and Poisson service versus $T$.}\label{fig:mess_N10_Td_burst_poiss}
 \end{center}
\end{minipage}
\end{figure}



	\begin{figure}[t]
	\graphicspath{{./Figures/}}
	\centering
	\includegraphics[width=3.2in]{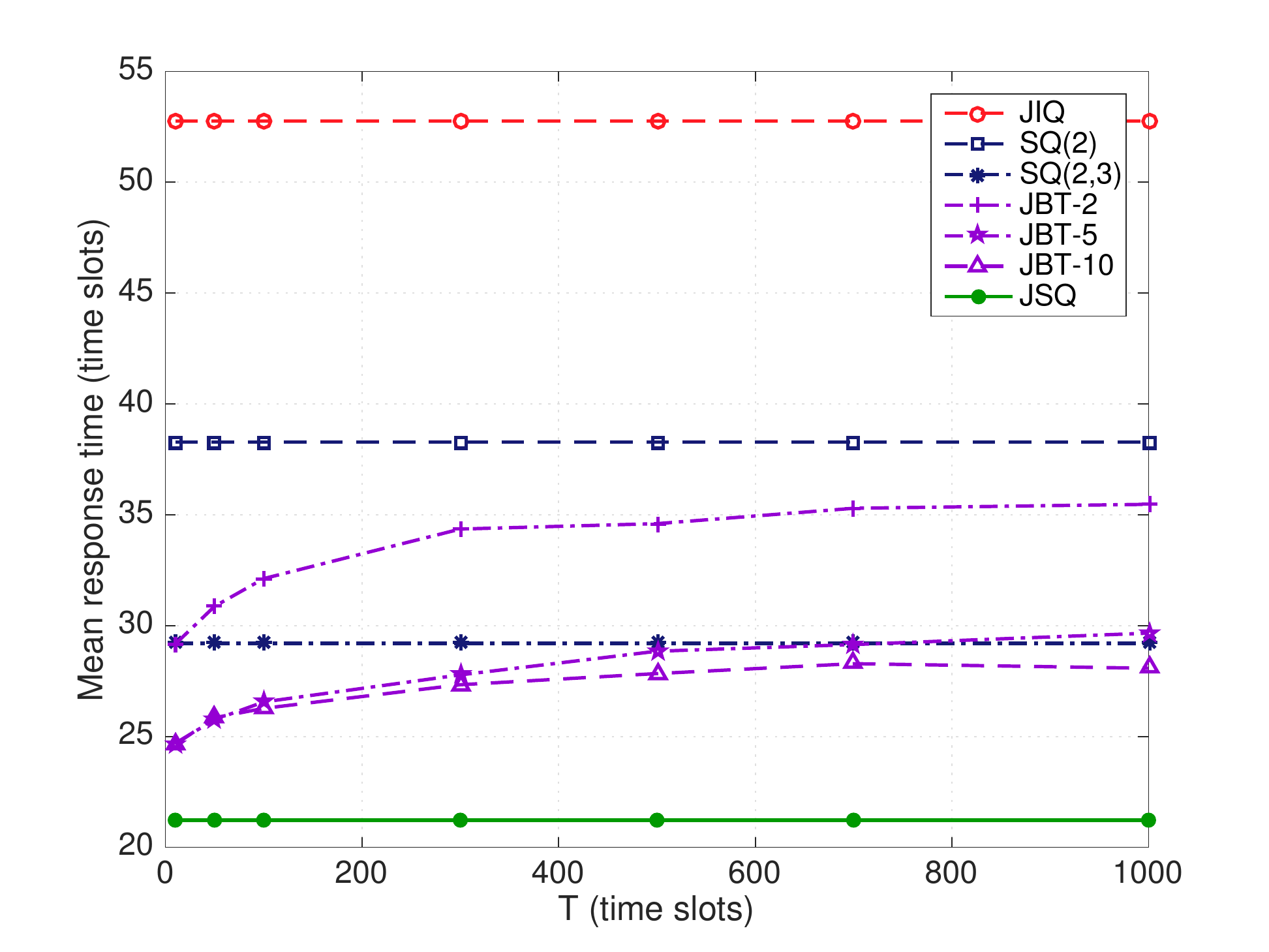}
	\caption{Delay performance under $10$ homogeneous servers with bursty arrival and Poisson service versus $T$.}\label{fig:delay_N10_Td_burst_poiss}
	\end{figure}

	Let us turn to look at the throughput performance for different combinations of arrival and service distribution in $10$ heterogeneous servers. As before, half of them have rate $1$ and half of them have rate $10$. Figure \ref{fig:throughput_poiss_const} shows the case of Poisson arrival and constant service. Similar results as before can be observed in this case. Figure \ref{fig:throughput_poiss_burst} shows the results for Poisson arrival and bursty service. In particular, the exogenous arrival at each time slot is drawn from Poisson distribution. The number of potential service at each time slot for servers with rate $1$ can be either $0$ or $10$, and the number of potential service at each time  for servers with rate $10$ is either $0$ or $15$ in this setting. Again, similar results are observed in this case. Figure \ref{fig:throughput_burst_poiss} illustrates the case of bursty arrival and Poisson service. In particular, the exogenous arrival at each time slot is either $0$ or $60$, and the potential service at each time slot of each server is drawn from a Poisson distribution with the corresponding rate. It can be easily seen that similar results hold in this case too.

\subsection{Delay Performance}

	In this subsection, we  provide additional delay performance for larger system sizes and different combinations of arrival and service process in homogeneous servers. In particular, we present the results for $50$ and $100$ homogeneous servers under Poisson arrival and Poisson service. Moreover, simulation results for $10$ homogeneous servers under Poisson arrival-constant service, Poisson arrival-bursty service and bursty arrival-Poisson arrival are all presented. As can be seen, our proposed policy JBT-$d$ achieves good performance in all these cases. All the servers have rate $1$.

	Let us first look at the delay performance under a load which ranges from light to heavy traffic regime for $50$ and $100$ homogeneous servers, as shown in Figures \ref{fig:delay_N50} and respectively. It can be easily seen that our proposed policy JBT-$d$ achieves a much smaller response time over all the considered loads than that of power-of-$2$ and power-of-$2$ with memory (SQ($2$,$9$), which utilizes almost the same amount of memory as the JBT-$d$). Moreover, the performance under JBT-$d$ is as good as JSQ for an even larger range of loads compared to the $10$ servers case, and the same trend holds for JIQ, which can be explained that  as $N$ becomes larger it is more likely to find an idle queue for a given load. However, as load approaches to $1$, the performance of JIQ degrades badly as discussed in the following.

	Now, let us get a closer look for the heavy-traffic regime, i.e., $\rho > 0.9$, under $50$ and $100$ servers, as shown in Figures \ref{fig:heavydelay_N50} and \ref{fig:heavydelay_N100}, respectively. We can easily observe that, as $\rho$ approaches to $1$, the performance of JIQ degrades substantially due to the lack of idle servers while our proposed policy JBT-$d$ remains quite close to JSQ in both cases. 

	Next, it is time for us to take a look at the delay performance under different combinations of arrival and service distribution in $10$ servers case. Figures \ref{fig:delay_poiss_const} and \ref{fig:heavydelay_poiss_const} shows the same set of results for the case of Poisson arrival and constant service. Same trend can be observed as before. Figures \ref{fig:delay_poiss_burst} and \ref{fig:heavydelay_poiss_burst} demonstrates the same set of results for the case of Poisson arrival and bursty service. In particular, the potential number of jobs served at each time slot is either $0$ or $10$ in this bursty setting. It can be easily seen that JIQ degrades much faster in this case as the load increases than that in the Poisson case. Figures \ref{fig:delay_burst_poiss} and \ref{fig:heavydelay_burst_poiss} illustrates the same set of results for the case of bursty arrival and Poisson service. In particular, the exogenous number of arrivals for each time slot is taken either $0$ or $12$. From the figures, we can observe that burst in arrival process also aggravates the performance of response time, and yet our proposed policy JBT-$d$ still has superior performance over other polices from light traffic to heavy traffic.

\subsection{Message Overhead}
	In this subsection, further results on message overhead for different settings are provided. As before, we present results for larger system sizes, e.g., $50$ and $100$ server, with Poisson arrival-Poisson service. Meanwhile, the message overhead for Poisson arrival-constant service, Poisson arrival-bursty service and bursty arrival-Poisson service in $10$ servers are also given. All these cases demonstrate the low message rate property of our proposed JBT-$d$ policy.

	Figure \ref{fig:mess_N50_Td} illustrates the message overhead with respect to $T$ under $50$ homogeneous servers at load $\rho = 0.99$, and the corresponding delay performance is shown in Figure \ref{fig:delay_N50_Td}. It can be easily seen that the JBT-$d$ policy is able to achieve a much better performance with a much lower message rate in a broad range of $T$ when compared to power-of-$d$ and power-of-$d$ with memory since both of them have message rate of $4$. Moreover, the message overhead of JBT-$d$ for $T>100$ is quite close to that of JIQ, which indicates that JBT-$d$ is capable of guaranteeing heavy-traffic delay optimality with just a slightly more message. In addition, we can observe that the delay performance is almost invariant with respect to $T$, especially when $T > 100$, which implies that we are allowed to use a large $T$ without sacrificing the delay performance. Figures \ref{fig:mess_N100_Td} and \ref{fig:delay_N100_Td} shows the case of $100$ homogeneous servers. We can see that the same results hold in this case as before.

	Now, let us  look at the message overhead under different arrival and service distributions. Figures \ref{fig:mess_N10_Td_poiss_const} and \ref{fig:delay_N10_Td_poiss_const} show the case for Poisson arrival and constant service with rate $1$. Nearly the same results hold in this case as well. Figures \ref{fig:mess_N10_Td_poiss_burst} and \ref{fig:delay_N10_Td_poiss_burst} demonstrates the case of Poisson arrival and bursty service. In particular, the burst in service is simulated by allowing the number of potential service at each time slot to take either $0$ or $10$. First thing to note is that burst in service degrades the performance of JIQ significantly, which agrees with our theoretical results. From the perspective of message rate, we can see that in this case message rate of JIQ is quite close to $0$, which accounts for the poor performance. In contrast, our proposed JBT-$d$ policy still achieve good performance with a much lower performance in this case. It is worth noting that although the delay performance of power-of-$d$ with memory (SQ($2$,$3$)) is slightly better than that of our JBT-$d$ policy, its message rate is four fold as much as the rate of JBT-$d$ when $T > 50$. As before, a larger $d$ does not necessarily mean a larger message rate when $T$ is large. Figures \ref{fig:mess_N10_Td_burst_poiss} and \ref{fig:delay_N10_Td_burst_poiss} illustrate the case of bursty arrival and Poisson service. More specifically, the exogenous arrival for each time slot is either $0$ or $12$, and the potential service for each time slot at each server is drawn from Poisson distribution with mean $1$. The behavior of message overhead is nearly the same as before. That is, for a large range of $T$, the message rate of JBT-$d$ is strictly less than $1$, which implies that our proposed JBT-$d$ is able to achieve heavy-traffic delay optimality by using a message rate that is less than $1$. From the delay performance in Figure \ref{fig:delay_N10_Td_burst_poiss}, we can observe that JBT-$d$ still has good performance in bursty arrival case and it also suggests that a large $d$ would be better for this bursty arrival setting.

\section{Confidence Intervals}
\label{sec:CI}
In this section, we provide the $95\%$ confidence interval for all average-based simulation results. In particular, it is obtained by running the simulation for $9,000,000$ time slots, which are divided into $30$ batches. Then, the standard batch means method is adopted to compute the confidence interval for all the results.
\subsection{Throughput performance}
In this subsection, we present the $95\%$ confidence intervals for all the simulation results of throughput performance, which is summarized in Table \ref{tab:result_CI_throughput}.

\subsection{Delay performance}
In this subsection, we present the $95\%$ confidence intervals for all the simulation results of delay performance, which is summarized in Table \ref{tab:result_CI_delay}.

\subsection{Message overhead}
In this subsection, we present the $95\%$ confidence intervals for all the simulation results of message overhead and the corresponding delay performance, which are summarized in Table \ref{tab:result_CI_mess}.

\begin{table}[!ht]

	\centering

	\caption{The simulation result of throughput performance and its corresponding confidence interval.}\label{tab:result_CI_throughput}

	\begin{tabular}{|c|c|}
	\hline
	Figures  & Confidence interval  table\\
	\hline
	\hline
	Figure \ref{fig:throughput_N10} & Table \ref{tab:throughput_N10}\\
	\hline
	Figure \ref{fig:throughput_N50} & Table \ref{tab:throughput_N50}\\
	\hline
	Figure \ref{fig:throughput_N100} & Table \ref{tab:throughput_N100}\\
	\hline
	Figure \ref{fig:throughput_poiss_const} & Table \ref{tab:throughput_poiss_const}\\
	\hline
	Figure \ref{fig:throughput_poiss_burst} & Table \ref{tab:throughput_poiss_burst}\\
	\hline
	Figure \ref{fig:throughput_burst_poiss} & Table \ref{tab:throughput_burst_poiss}\\
	\hline
	\end{tabular}

\end{table}

	\begin{table}[!ht]

	\centering

	\caption{The simulation result of delay performance and its corresponding confidence interval.}\label{tab:result_CI_delay}

	\begin{tabular}{|c|c|}
	\hline
	Figures  & Confidence interval  table\\
	\hline
	\hline
	Figures \ref{fig:delay_N10} and \ref{fig:heavydelay_N10} & Table \ref{tab:delay_N10}\\
	\hline
	Figures \ref{fig:delay_N50} and \ref{fig:heavydelay_N50_sim}(\ref{fig:heavydelay_N50})& Table \ref{tab:delay_N50}\\
	\hline
	Figures \ref{fig:delay_N100} and \ref{fig:heavydelay_N100}& Table \ref{tab:delay_N100}\\
	\hline
	Figures \ref{fig:delay_poiss_const} and \ref{fig:heavydelay_poiss_const} & Table \ref{tab:delay_poiss_const}\\
	\hline
	Figures \ref{fig:delay_poiss_burst} and \ref{fig:heavydelay_poiss_burst_sim}(\ref{fig:heavydelay_poiss_burst})& Table \ref{tab:delay_poiss_burst}\\
	\hline
	Figures \ref{fig:delay_burst_poiss} and \ref{fig:heavydelay_burst_poiss}& Table \ref{tab:delay_burst_poiss}\\
	\hline
	\end{tabular}

	\end{table}

	\begin{table}[!ht]

	\centering

	\caption{The simulation result of message overhead (delay performance) and its corresponding confidence interval.}\label{tab:result_CI_mess}

	\begin{tabular}{|c|c|}
	\hline
	Figures  & Confidence interval  table\\
	\hline
	\hline
	Figures \ref{fig:mess_N10_Td} and \ref{fig:delay_N10_Td} & Table \ref{tab:mess_N10_Td} and Table \ref{tab:delay_N10_Td}\\
	\hline
	Figures \ref{fig:mess_N50_Td} and \ref{fig:delay_N50_Td}& Table \ref{tab:mess_N50_Td}  and Table \ref{tab:delay_N50_Td}\\
	\hline
	Figures \ref{fig:mess_N100_Td} and \ref{fig:delay_N100_Td}& Table \ref{tab:mess_N100_Td} and Table \ref{tab:delay_N100_Td}\\
	\hline
	Figures \ref{fig:mess_N10_Td_poiss_const} and \ref{fig:delay_N10_Td_poiss_const} & Table \ref{tab:mess_N10_Td_poiss_const} and Table \ref{tab:delay_N10_Td_poiss_const}\\
	\hline
	Figures \ref{fig:mess_N10_Td_poiss_burst} and \ref{fig:delay_N10_Td_poiss_burst}& Table \ref{tab:mess_N10_Td_poiss_burst} and Table \ref{tab:delay_N10_Td_poiss_burst}\\
	\hline
	Figures \ref{fig:mess_N10_Td_burst_poiss} and \ref{fig:delay_N10_Td_burst_poiss}& Table \ref{tab:mess_N10_Td_burst_poiss} and Table \ref{tab:delay_N10_Td_burst_poiss} \\
	\hline
	\end{tabular}

	\end{table}

\begin{table*}

	\centering

	\caption{The $95\%$ confidence interval of response time under $10$ heterogeneous servers with Poisson arrival and Poisson service.}\label{tab:throughput_N10}

	\begin{tabular}{|c|c|c|c|c|c|}
	\hline
	Policy   & $\rho=0.5$ &$\rho=0.7$ & $\rho=0.9$ & $\rho = 0.95$ & $\rho = 0.99$ \\
	\hline
	\hline
	 JSQ &$3.363\pm 0.003$ &$3.894\pm 0.003$ &$4.998\pm0.007$ &$5.804\pm0.010$ & $8.093\pm0.083$\\
	\hline
	JBTG-10  &$2.156\pm 0.005$ &$3.246\pm 0.005$ &$5.732\pm0.027$ &$8.151\pm0.100$ & $13.562\pm0.380$\\
	\hline
	JBTG-2  &$2.192\pm 0.008$ &$3.414\pm 0.034$ &$6.159\pm0.079$ &$8.294\pm0.142$ & $13.902\pm0.412$\\
	\hline
	JBT-2  &$3.715\pm 0.084$ &$5.288\pm 0.113$ &$16.882\pm1.463$ &$28.904\pm4.272$ & $53.830\pm16.384$\\
	\hline
	SQ(2,3) &$3.968\pm 0.008$ &$6.398\pm 0.151$ &$12.942\pm0.052$ &$18.200\pm0.103$ & $31.200\pm0.291$\\
	\hline
	SQ(2) &$1.9\mathrm e 4\pm 0.4\mathrm e 4$ &$4.2\mathrm e 4\pm 0.9 \mathrm e 4$ &$5.5 \mathrm e 4\pm1.1 \mathrm e 4$ &$5.7 \mathrm e 4\pm1.2\mathrm e 4$ & $5.9 \mathrm e 4\pm1.2\mathrm e 4$\\
	\hline
	JIQ &$3.112\pm 0.003$ &$3.896\pm 0.006$ &$1.1 \mathrm e 4\pm2.3 \mathrm e 3$ &$2.7 \mathrm e 4\pm5.6\mathrm e 3$ & $3.9 \mathrm e 4\pm8.1\mathrm e 3$\\
	\hline
	\end{tabular}

\end{table*}

\begin{table*}

	\centering

	\caption{The $95\%$ confidence interval of response time under $50$ heterogeneous servers with Poisson arrival and Poisson service. }\label{tab:throughput_N50}

	\begin{tabular}{|c|c|c|c|c|c|}
	\hline
	Policy   & $\rho=0.5$ &$\rho=0.7$ & $\rho=0.9$ & $\rho = 0.95$ & $\rho = 0.99$ \\
	\hline
	\hline
	 JSQ &$17.841\pm 0.007$ &$20.564\pm 0.008$ &$23.369\pm0.008$ &$24.510\pm0.016$ & $27.290\pm0.051$\\
	\hline
	JBTG-50  &$12.115\pm 0.019$ &$17.134\pm 0.021$ &$22.449\pm0.015$ &$26.091\pm0.069$ & $42.996\pm0.685$\\
	\hline
	JBTG-10  &$12.119\pm 0.019$ &$17.147\pm 0.022$ &$22.837\pm0.085$ &$26.036\pm0.146$ & $35.153\pm0.441$\\
	\hline
	JBT-10  &$16.659\pm 0.011$ &$20.184\pm 0.044$ &$25.544\pm0.159$ &$34.323\pm0.496$ & $64.107\pm5.690$\\
	\hline
	SQ(2,9) &$22.554\pm 0.077$ &$39.086\pm 0.071$ &$77.287\pm0.223$ &$103.829\pm0.511$ & $170.932\pm1.552$\\
	\hline
	SQ(2)  &$2.9\mathrm e 4\pm 0.6\mathrm e 4$ &$5.1\mathrm e 4\pm 1.1 \mathrm e 4$ &$6.5 \mathrm e 4\pm1.4 \mathrm e 4$ &$6.7 \mathrm e 4\pm1.4\mathrm e 4$ & $6.9 \mathrm e 4\pm1.4\mathrm e 4$\\
	\hline
	JIQ &$16.663\pm 0.011$ &$20.031\pm 0.008$ &$24.553 \pm0.044$ &$57.322 \pm1.336$ & $1.2 \mathrm e 4\pm2.6\mathrm e 3$\\
	\hline
	\end{tabular}

\end{table*}

\begin{table*}

	\centering

	\caption{The $95\%$ confidence interval of response time under $100$ heterogeneous servers with Poisson arrival and Poisson service.}\label{tab:throughput_N100}

	\begin{tabular}{|c|c|c|c|c|c|}
	\hline
	Policy   & $\rho=0.5$ &$\rho=0.7$ & $\rho=0.9$ & $\rho = 0.95$ & $\rho = 0.99$ \\
	\hline
	\hline
	 JSQ &$35.964\pm 0.010$ &$41.441\pm 0.011$ &$46.942\pm0.010$ &$48.572\pm0.047$ & $51.896\pm0.138$\\
	\hline
	JBTG-100  &$24.627\pm 0.035$ &$34.615\pm 0.037$ &$44.789\pm0.016$ &$47.963\pm0.039$ & $74.042\pm0.699$\\
	\hline
	JBTG-20  &$24.621\pm 0.037$ &$34.610\pm 0.035$ &$44.845\pm0.025$ &$48.378\pm0.108$ & $61.347\pm0.590$\\
	\hline
	JBT-20  &$33.592\pm 0.0157$ &$40.417\pm 0.017$ &$47.508\pm0.127$ &$54.114\pm0.306$ & $92.365\pm4.815$\\
	\hline
	SQ(2,15) &$47.041\pm 0.217$ &$81.887\pm 0.164$ &$158.471\pm0.638$ &$211.925\pm0.964$ & $347.291\pm3.172$\\
	\hline
	SQ(2)  &$3.0\mathrm e 4\pm 0.6\mathrm e 4$ &$5.3\mathrm e 4\pm 1.1 \mathrm e 4$ &$6.6 \mathrm e 4\pm1.4 \mathrm e 4$ &$6.9 \mathrm e 4\pm1.4\mathrm e 4$ & $7.1 \mathrm e 4\pm1.5\mathrm e 4$\\
	\hline
	JIQ &$33.589\pm 0.019$ &$40.403\pm 0.012$ &$46.843 \pm0.028$ &$54.622 \pm0.203$ & $5.9 \mathrm e 3\pm1.2\mathrm e 3$\\
	\hline
	\end{tabular}

\end{table*}

\newpage

\begin{table*}

	\centering

	\caption{The $95\%$ confidence interval of response time under $10$ heterogeneous servers with Poisson arrival and constant service. }\label{tab:throughput_poiss_const}

	\begin{tabular}{|c|c|c|c|c|c|}
	\hline
	Policy   & $\rho=0.5$ &$\rho=0.7$ & $\rho=0.9$ & $\rho = 0.95$ & $\rho = 0.99$ \\
	\hline
	\hline
	 JSQ &$3.274\pm 0.003$ &$3.802\pm 0.002$ &$4.747\pm0.005$ &$5.447\pm0.008$ & $7.197\pm0.033$\\
	\hline
	JBTG-10  &$2.063\pm 0.004$ &$3.081\pm 0.004$ &$5.344\pm0.020$ &$7.623\pm0.100$ & $13.234\pm0.471$\\
	\hline
	JBTG-2  &$2.085\pm 0.009$ &$3.285\pm 0.027$ &$5.825\pm0.086$ &$7.871\pm0.174$ & $12.782\pm0.407$\\
	\hline
	JBT-2  &$3.575\pm 0.087$ &$5.108\pm 0.149$ &$14.389\pm1.769$ &$25.579\pm4.832$ & $58.508\pm17.599$\\
	\hline
	SQ(2,3) &$3.848\pm 0.007$ &$6.257\pm 0.133$ &$12.742\pm0.055$ &$17.800\pm0.087$ & $30.550\pm0.268$\\
	\hline
	SQ(2) &$1.8\mathrm e 4\pm 0.4\mathrm e 4$ &$4.1\mathrm e 4\pm 0.9 \mathrm e 4$ &$5.5 \mathrm e 4\pm1.1 \mathrm e 4$ &$5.7 \mathrm e 4\pm1.2\mathrm e 4$ & $5.9 \mathrm e 4\pm1.2\mathrm e 4$\\
	\hline
	JIQ &$3.018\pm 0.003$ &$3.714\pm 0.004$ &$7.0 \mathrm e 3\pm1.5 \mathrm e 3$ &$2.3 \mathrm e 4\pm4.9\mathrm e 3$ & $3.6 \mathrm e 4\pm7.5\mathrm e 3$\\
	\hline
	\end{tabular}

\end{table*}

\begin{table*}

	\centering

	\caption{The $95\%$ confidence interval of response time under $10$ heterogeneous servers with Poisson arrival and bursty service. }\label{tab:throughput_poiss_burst}

	\begin{tabular}{|c|c|c|c|c|c|}
	\hline
	Policy   & $\rho=0.5$ &$\rho=0.7$ & $\rho=0.9$ & $\rho = 0.95$ & $\rho = 0.99$ \\
	\hline
	\hline
	 JSQ &$3.274\pm 0.003$ &$3.802\pm 0.002$ &$4.747\pm0.005$ &$5.447\pm0.008$ & $7.197\pm0.033$\\
	\hline
	JBTG-10  &$2.063\pm 0.004$ &$3.081\pm 0.004$ &$5.344\pm0.020$ &$7.623\pm0.100$ & $13.234\pm0.471$\\
	\hline
	JBTG-2  &$2.085\pm 0.009$ &$3.285\pm 0.027$ &$5.825\pm0.086$ &$7.871\pm0.174$ & $12.782\pm0.407$\\
	\hline
	JBT-2  &$3.575\pm 0.087$ &$5.108\pm 0.149$ &$14.389\pm1.769$ &$25.579\pm4.832$ & $58.508\pm17.599$\\
	\hline
	SQ(2,3) &$3.848\pm 0.007$ &$6.257\pm 0.133$ &$12.742\pm0.055$ &$17.800\pm0.087$ & $30.550\pm0.268$\\
	\hline
	SQ(2) &$1.8\mathrm e 4\pm 0.4\mathrm e 4$ &$4.1\mathrm e 4\pm 0.9 \mathrm e 4$ &$5.5 \mathrm e 4\pm1.1 \mathrm e 4$ &$5.7 \mathrm e 4\pm1.2\mathrm e 4$ & $5.9 \mathrm e 4\pm1.2\mathrm e 4$\\
	\hline
	JIQ &$3.018\pm 0.003$ &$3.714\pm 0.004$ &$7.0 \mathrm e 3\pm1.5 \mathrm e 3$ &$2.3 \mathrm e 4\pm4.9\mathrm e 3$ & $3.6 \mathrm e 4\pm7.5\mathrm e 3$\\
	\hline
	\end{tabular}

\end{table*}

\begin{table*}

	\centering

	\caption{The $95\%$ confidence interval of response time under $10$ heterogeneous servers with bursty arrival and Poisson service. }\label{tab:throughput_burst_poiss}

	\begin{tabular}{|c|c|c|c|c|c|}
	\hline
	Policy   & $\rho=0.5$ &$\rho=0.7$ & $\rho=0.9$ & $\rho = 0.95$ & $\rho = 0.99$ \\
	\hline
	\hline
	 JSQ &$3.946\pm 0.007$ &$4.519\pm 0.007$ &$6.036\pm0.014$ &$7.157\pm0.032$ & $12.656\pm0.578$\\
	\hline
	JBTG-10  &$2.746\pm 0.008$ &$4.006\pm 0.009$ &$7.527\pm0.089$ &$10.588\pm0.231$ & $20.027\pm1.026$\\
	\hline
	JBTG-2  &$2.795\pm 0.012$ &$4.182\pm 0.036$ &$7.628\pm0.097$ &$10.319\pm0.231$ & $21.718\pm1.092$\\
	\hline
	JBT-2  &$4.397\pm 0.127$ &$6.203\pm 0.163$ &$30.532\pm5.625$ &$37.714\pm6.651$ & $64.409\pm19.751$\\
	\hline
	SQ(2,3) &$4.705\pm 0.012$ &$7.061\pm 0.023$ &$13.755\pm0.076$ &$19.396\pm0.154$ & $35.491\pm0.959$\\
	\hline
	SQ(2) &$1.8\mathrm e 4\pm 0.4\mathrm e 4$ &$4.2\mathrm e 4\pm 0.9 \mathrm e 4$ &$5.5 \mathrm e 4\pm1.1 \mathrm e 4$ &$5.7 \mathrm e 4\pm1.2\mathrm e 4$ & $5.9 \mathrm e 4\pm1.2\mathrm e 4$\\
	\hline
	JIQ &$3.725\pm 0.007$ &$5.092\pm 0.020$ &$2.1 \mathrm e 4\pm4.4 \mathrm e 3$ &$3.6 \mathrm e 4\pm7.5\mathrm e 3$ & $4.7 \mathrm e 4\pm9.8\mathrm e 3$\\
	\hline
	\end{tabular}

\end{table*}

\begin{table*}

	\centering

	\caption{The $95\%$ confidence interval of response time under $10$ homogeneous servers }\label{tab:delay_N10}

	\begin{tabular}{|c|c|c|c|c|c|}
	\hline
	Policy   & $\rho=0.3$ &$\rho=0.5$ & $\rho=0.7$ & $\rho = 0.9$ & $\rho = 0.95$ \\
	\hline
	\hline
	 JSQ   &$2.027\pm 0.005$  &$3.015\pm 0.006$  &$4.038\pm 0.007$  &$5.608\pm0.018$  &$6.854\pm0.052$ \\
	\hline
	JBT-10 &$2.027\pm 0.006$  &$3.015\pm 0.005$  &$4.073\pm 0.008$  &$6.671\pm0.041$  &$9.255\pm0.127$ \\
	\hline
	JBT-2  &$2.035\pm 0.007$  &$3.103\pm 0.019$  &$4.419\pm 0.047$  &$7.873\pm0.111$   &$10.826\pm0.345$ \\
	\hline
	SQ(2,3) &$2.027\pm 0.005$ &$3.054\pm 0.007$ &$4.389\pm 0.009$ &$8.080\pm0.044$ &$11.407\pm0.098$ \\
	\hline
	SQ(2) &$2.122\pm 0.006$ &$3.516\pm 0.009$ &$5.764\pm 0.017$ &$11.494\pm0.071$ &$15.737\pm0.128$ \\
	\hline
	JIQ &$2.028\pm 0.005$ &$3.017\pm 0.007$ &$4.070\pm 0.008$ &$6.780\pm0.035$ &$10.429\pm0.142$ \\
	\hline
	\hline
	\end{tabular}
	
	\begin{tabular}{|c|c|c|c|c|}
	\hline
	Policy   & $\rho=0.97$ &$\rho=0.99$ & $\rho=0.993$ & $\rho = 0.995$ \\
	\hline
	\hline
	 JSQ &$8.336\pm 0.142$ & $15.838\pm0.341$ &$19.119\pm0.562$ &$27.090\pm0.961$ \\
	\hline
	JBT-10  &$11.696\pm 0.205$& $22.035\pm0.401$ &$25.210\pm0.654$ &$34.208\pm1.092$ \\
	\hline
	JBT-2  &$14.127\pm 0.449$ & $28.686\pm0.532$ &$33.288\pm0.777$ &$44.217\pm1.234$ \\
	\hline
	SQ(2,3) &$14.484\pm 0.221$ & $25.010\pm0.353$ &$28.412\pm0.585$ &$36.990\pm1.013$ \\
	\hline
	SQ(2) &$19.512\pm 0.239$ & $31.469\pm0.365$ &$35.046\pm0.601$ &$44.112\pm1.082$ \\
	\hline
	JIQ &$15.221\pm 0.369$ & $41.456\pm0.832$ &$53.308\pm1.853$ &$74.239\pm3.526$ \\
	\hline
	\end{tabular}

\end{table*}

\newpage

\begin{table*}

	\centering

	\caption{The $95\%$ confidence interval of response time under $50$ homogeneous servers }\label{tab:delay_N50}

	\begin{tabular}{|c|c|c|c|c|c|}
	\hline
	Policy   & $\rho=0.3$ &$\rho=0.5$ & $\rho=0.7$ & $\rho = 0.9$ & $\rho = 0.95$ \\
	\hline
	\hline
	 JSQ   &$8.006\pm 0.008$  &$13.005\pm 0.012$  &$17.992\pm 0.015$  &$23.005\pm0.015$  &$24.489\pm0.020$ \\
	\hline
	JBT-50 &$8.011\pm 0.008$  &$13.005\pm 0.015$  &$17.992\pm 0.015$  &$23.151\pm0.020$  &$26.377\pm0.073$ \\
	\hline
	JBT-10  &$8.006\pm 0.008$  &$13.006\pm 0.015$  &$18.002\pm 0.016$  &$23.538\pm0.055$   &$26.797\pm0.154$ \\
	\hline
	SQ(2,9) &$8.014\pm 0.009$ &$13.005\pm 0.015$ &$18.242\pm 0.019$ &$36.489\pm0.129$ &$51.968\pm0.218$ \\
	\hline
	SQ(2) &$8.468\pm 0.010$ &$15.442\pm 0.023$ &$26.412\pm 0.045$ &$53.305\pm0.147$ &$71.985\pm0.345$ \\
	\hline
	JIQ &$8.012\pm 0.008$ &$13.006\pm 0.015$ &$17.993\pm 0.015$ &$23.158\pm0.0158$ &$26.466\pm0.058$ \\
	\hline
	\hline
	\end{tabular}
	
	\begin{tabular}{|c|c|c|c|c|c|}
	\hline
	Policy   & $\rho=0.97$ &$\rho=0.99$ & $\rho=0.993$ & $\rho = 0.995$  \\
	\hline
	\hline
	 JSQ &$25.423\pm 0.031$ & $27.694\pm0.056$ &$28.775\pm0.093$ &$30.100\pm0.186$ \\
	\hline
	JBT-50  &$30.136\pm 0.139$& $40.453\pm0.237$ &$44.106\pm0.372$ &$46.407\pm0.524$ \\
	\hline
	JBT-10  &$29.914\pm 0.187$ & $36.279\pm0.172$ &$39.181\pm0.214$ &$41.928\pm0.417$ \\
	\hline
	SQ(2,9) &$63.413\pm 0.370$ & $89.872\pm0.472$ &$98.885\pm0.694$ &$108.071\pm1.063$ \\
	\hline
	SQ(2) &$86.820\pm 0.538$ & $118.975\pm0.576$ &$131.737\pm0.821$ &$142.267\pm1.200$ \\
	\hline
	JIQ &$30.715\pm 0.141$ & $49.152\pm0.523$ &$61.600\pm1.094$ &$77.576\pm2.162$ \\
	\hline
	\end{tabular}

\end{table*}

\begin{table*}

	\centering

	\caption{The $95\%$ confidence interval of response time under $100$ homogeneous servers }\label{tab:delay_N100}

	\begin{tabular}{|c|c|c|c|c|c|}
	\hline
	Policy   & $\rho=0.3$ &$\rho=0.5$ & $\rho=0.7$ & $\rho = 0.9$ & $\rho = 0.95$ \\
	\hline
	\hline
	 JSQ   &$15.501\pm 0.011$  &$25.495\pm 0.017$  &$35.516\pm 0.020$  &$45.491\pm0.026$  &$48.034\pm0.023$ \\
	\hline
	JBT-100 &$15.504\pm 0.011$  &$25.501\pm 0.017$  &$35.516\pm 0.020$  &$45.496\pm0.026$  &$48.797\pm0.044$ \\
	\hline
	JBT-20  &$15.502\pm 0.011$  &$25.495\pm 0.019$  &$35.517\pm 0.020$  &$45.624\pm0.038$   &$49.098\pm0.069$ \\
	\hline
	SQ(2,15) &$15.501\pm 0.012$ &$25.497\pm 0.018$ &$35.707\pm 0.024$ &$73.777\pm0.257$ &$104.120\pm0.367$ \\
	\hline
	SQ(2) &$16.427\pm 0.014$ &$30.405\pm 0.041$ &$52.416\pm 0.086$ &$105.759\pm0.337$ &$143.503\pm0.494$ \\
	\hline
	JIQ &$15.502\pm 0.012$ &$25.501\pm 0.018$ &$35.516\pm 0.021$ &$45.500\pm0.025$ &$48.776\pm0.039$ \\
	\hline
	\hline
	\end{tabular}
	
	\begin{tabular}{|c|c|c|c|c|c|}
	\hline
	Policy   & $\rho=0.97$ &$\rho=0.99$ & $\rho=0.993$ & $\rho = 0.995$ & $\rho = 0.999$ \\
	\hline
	\hline
	 JSQ &$49.274\pm 0.028$ & $51.666\pm0.086$ &$52.646\pm0.164$ &$53.481\pm0.287$ & $61.550\pm2.018$\\
	\hline
	JBT-100  &$53.061\pm 0.106$& $67.958\pm0.898$ &$73.861\pm1.067$ &$78.988\pm1.348$ & $98.545\pm4.268$\\
	\hline
	JBT-20  &$52.757\pm 0.140$ & $61.210\pm0.497$ &$64.172\pm0.633$ &$66.531\pm1.023$ & $84.580\pm3.389$\\
	\hline
	SQ(2,15) &$127.063\pm 0.549$ & $180.343\pm1.833$ &$197.691\pm2.465$ &$212.866\pm3.037$ & $297.502\pm8.737$\\
	\hline
	SQ(2) &$172.406\pm 0.869$ & $237.873\pm2.293$ &$261.528\pm3.387$ &$279.946\pm4.296$ & $385.839\pm10.543$\\
	\hline
	JIQ &$53.172\pm 0.125$ & $73.049\pm0.849$ &$85.704\pm1.826$ &$98.982\pm2.802$ & $294.567\pm21.828$\\
	\hline
	\end{tabular}

\end{table*}

\begin{table*}

	\centering

	\caption{The $95\%$ confidence interval of response time under $10$ homogeneous servers with Poisson arrival and constant service.}\label{tab:delay_poiss_const}

	\begin{tabular}{|c|c|c|c|c|c|}
	\hline
	Policy   & $\rho=0.3$ &$\rho=0.5$ & $\rho=0.7$ & $\rho = 0.9$ & $\rho = 0.95$ \\
	\hline
	\hline
	 JSQ   &$1.500\pm 0.004$  &$2.497\pm 0.005$  &$3.497\pm 0.005$  &$4.716\pm0.011$  &$5.458\pm0.022$ \\
	\hline
	JBT-10 &$1.500\pm 0.004$  &$2.497\pm 0.005$  &$3.500\pm 0.005$  &$5.120\pm0.021$  &$6.568\pm0.044$ \\
	\hline
	JBT-2  &$1.506\pm 0.005$  &$2.550\pm 0.014$  &$3.776\pm 0.037$  &$6.240\pm0.147$   &$8.462\pm0.214$ \\
	\hline
	SQ(2,3) &$1.500\pm 0.004$ &$2.511\pm 0.005$ &$3.685\pm 0.007$ &$6.745\pm0.035$ &$9.617\pm0.069$ \\
	\hline
	SQ(2) &$1.552\pm 0.005$ &$2.850\pm 0.006$ &$4.863\pm 0.013$ &$9.830\pm0.050$ &$13.492\pm0.088$ \\
	\hline
	JIQ &$1.500\pm 0.004$ &$2.497\pm 0.005$ &$3.500\pm 0.005$ &$5.099\pm0.018$ &$6.661\pm0.038$ \\
	\hline
	\hline
	\end{tabular}
	
	\begin{tabular}{|c|c|c|c|c|c|}
	\hline
	Policy   & $\rho=0.97$ &$\rho=0.99$ & $\rho=0.993$ & $\rho = 0.995$  \\
	\hline
	\hline
	 JSQ &$6.201\pm 0.042$ & $9.962\pm0.157$ &$11.945\pm0.332$ &$15.422\pm0.467$ \\
	\hline
	JBT-10  &$7.878\pm 0.079$& $12.977\pm0.183$ &$15.302\pm0.405$ &$10.209\pm0.523$ \\
	\hline
	JBT-2  &$10.396\pm 0.290$ & $18.985\pm0.382$ &$22.628\pm0.583$ &$29.774\pm0.741$ \\
	\hline
	SQ(2,3) &$12.014\pm 0.111$ & $19.104\pm0.217$ &$21.985\pm0.397$ &$26.231\pm0.515$ \\
	\hline
	SQ(2) &$16.402\pm 0.145$ & $24.736\pm0.230$ &$27.764\pm0.444$ &$32.610\pm0.521$ \\
	\hline
	JIQ &$8.355\pm 0.087$ & $15.467\pm0.251$ &$18.773\pm0.496$ &$24.277\pm0.637$ \\
	\hline
	\end{tabular}

\end{table*}

\newpage

\begin{table*}

	\centering

	\caption{The $95\%$ confidence interval of response time under $10$ homogeneous servers with Poisson arrival and bursty service}\label{tab:delay_poiss_burst}

	\begin{tabular}{|c|c|c|c|c|c|}
	\hline
	Policy   & $\rho=0.3$ &$\rho=0.5$ & $\rho=0.7$ & $\rho = 0.9$ & $\rho = 0.95$ \\
	\hline
	\hline
	 JSQ   &$9.009\pm 0.027$  &$9.154\pm 0.025$  &$9.974\pm 0.031$  &$14.164\pm0.178$  &$19.512\pm0.587$ \\
	\hline
	JBT-10 &$9.100\pm 0.023$  &$9.773\pm 0.035$  &$12.252\pm 0.072$  &$23.693\pm0.493$  &$33.289\pm1.254$ \\
	\hline
	JBT-2  &$9.062\pm 0.026$  &$9.667\pm 0.037$  &$11.919\pm 0.091$  &$21.243\pm0.638$   &$30.387\pm0.968$ \\
	\hline
	SQ(2,3) &$9.017\pm 0.025$ &$9.398\pm 0.028$ &$11.183\pm 0.043$ &$18.287\pm0.241$ &$25.669\pm0.649$ \\
	\hline
	SQ(2) &$9.149\pm 0.029$ &$10.454\pm 0.045$ &$13.880\pm 0.063$ &$24.829\pm0.336$ &$34.707\pm0.748$ \\
	\hline
	JIQ &$9.080\pm 0.032$ &$9.810\pm 0.036$ &$12.652\pm 0.070$ &$31.385\pm0.035$ &$60.308\pm2.982$ \\
	\hline
	\hline
	\end{tabular}
	
	\begin{tabular}{|c|c|c|c|c|c|}
	\hline
	Policy   & $\rho=0.97$ &$\rho=0.99$ & $\rho=0.993$ & $\rho = 0.995$  \\
	\hline
	\hline
	 JSQ &$26.375\pm 1.204$ &$61.864\pm 2.771$ &$91.184\pm6.933$ &$114.013\pm10.963$ \\
	\hline
	JBT-10  &$43.297\pm 2.287$& $84.148\pm3.274$ &$110.089\pm7.315$ &$140.608\pm11.082$ \\
	\hline
	JBT-2  &$42.166\pm 2.523$ & $87.971\pm3.131$ &$120.281\pm7.352$ &$147.652\pm11.987$ \\
	\hline
	SQ(2,3) &$33.781\pm 1.338$ & $71.298\pm2.815$ &$102.409\pm7.020$ &$124.326\pm10.985$ \\
	\hline
	SQ(2) &$44.433\pm 1.575$ & $84.208\pm2.864$ &$117.421\pm7.138$ &$138.626\pm11.177$ \\
	\hline
	JIQ &$98.679\pm 7.706$ & $280.949\pm16.030$ &$398.908\pm26.882$ &$527.805\pm47.782$ \\
	\hline
	\end{tabular}

\end{table*}

\begin{table*}

	\centering

	\caption{The $95\%$ confidence interval of response time under $10$ homogeneous servers with bursty arrival and Poisson service }\label{tab:delay_burst_poiss}

	\begin{tabular}{|c|c|c|c|c|c|}
	\hline
	Policy   & $\rho=0.3$ &$\rho=0.5$ & $\rho=0.7$ & $\rho = 0.9$ & $\rho = 0.95$ \\
	\hline
	\hline
	 JSQ   &$6.001\pm 0.025$  &$6.015\pm 0.013$  &$6.128\pm 0.016$  &$7.153\pm0.037$  &$8.758\pm0.110$ \\
	\hline
	JBT-10 &$6.001\pm 0.029$  &$6.035\pm 0.015$  &$6.332\pm 0.021$  &$8.896\pm0.067$  &$12.043\pm0.233$ \\
	\hline
	JBT-2  &$6.031\pm 0.026$  &$6.210\pm 0.029$  &$6.966\pm 0.067$  &$10.064\pm0.170$   &$13.914\pm0.383$ \\
	\hline
	SQ(2,3) &$6.001\pm 0.027$ &$6.085\pm 0.015$ &$6.532\pm 0.022$ &$9.388\pm0.060$ &$12.805\pm0.154$ \\
	\hline
	SQ(2) &$6.365\pm 0.031$ &$7.185\pm 0.030$ &$8.931\pm 0.046$ &$14.394\pm0.110$ &$19.050\pm0.222$ \\
	\hline
	JIQ &$6.001\pm 0.028$ &$6.029\pm 0.015$ &$6.327\pm 0.019$ &$9.197\pm0.085$ &$14.023\pm0.262$ \\
	\hline
	\hline
	\end{tabular}
	
	\begin{tabular}{|c|c|c|c|c|c|}
	\hline
	Policy   & $\rho=0.97$ &$\rho=0.99$ & $\rho=0.993$ & $\rho = 0.995$  \\
	\hline
	\hline
	 JSQ &$10.778\pm 0.262$ & $21.872\pm0.698$ &$33.241\pm1.410$ &$40.061\pm8.565$ \\
	\hline
	JBT-10  &$15.293\pm 0.433$& $29.366\pm0.834$ &$39.880\pm1.446$ &$49.440\pm8.605$ \\
	\hline
	JBT-2  &$17.919\pm 0.788$ & $34.347\pm1.070$ &$50.013\pm1.957$ &$63.382\pm10.735$ \\
	\hline
	SQ(2,3) &$16.103\pm 0.332$ & $29.333\pm0.728$ &$41.273\pm1.445$ &$48.821\pm8.520$ \\
	\hline
	SQ(2) &$23.381\pm 0.389$ & $38.513\pm0.791$ &$50.151\pm1.492$ &$58.266\pm8.672$ \\
	\hline
	JIQ &$20.209\pm 0.665$ & $53.776\pm1.853$ &$78.304\pm3.827$ &$117.161\pm20.937$ \\
	\hline
	\end{tabular}

\end{table*}

\begin{table*}

	\centering

	\caption{Number of messages per new job arrival under $10$ homogeneous servers with respect to $T$ }\label{tab:mess_N10_Td}

	\begin{tabular}{|c|c|c|c|c|c|}
	\hline
	Policy   & $T=10$ &$T=50$ & $T=100$ & $T= 500$ & $T = 700$ \\
	\hline
	\hline
	JSQ &$20$ &$20$ &$20$ &$20$ &$20$\\
	\hline
	JBT-10  &$2.633\pm 0.004$ &$0.978\pm 0.005$ &$0.755\pm0.006$ &$0.558\pm0.013$ & $0.545\pm0.014$\\
	\hline
	JBT-5  &$1.796\pm 0.003$ &$0.921\pm 0.004$ &$0.787\pm0.006$ &$0.658\pm0.017$ & $0.655\pm0.016$\\
	\hline
	JBT-2  &$1.377\pm 0.003$ &$0.935\pm 0.004$ &$0.868\pm0.006$ &$0.787\pm0.014$ & $0.795\pm0.014$\\
	\hline
	SQ(2,3) &$4$ &$4$ &$4$ &$4$ &$4$\\
	\hline
	SQ(2)  &$4$ &$4$ &$4$ &$4$ &$4$\\
	\hline
	JIQ &$0.167\pm 0.012$ &$0.167\pm 0.012$ &$0.167\pm 0.012$ &$0.167\pm 0.012$ &$0.167\pm 0.012$\\
	\hline
	\end{tabular}

\end{table*}

\newpage

\begin{table*}

	\centering

	\caption{The $95\%$ confidence interval of response time under $10$ homogeneous servers with respect to $T$ }\label{tab:delay_N10_Td}

	\begin{tabular}{|c|c|c|c|c|c|}
	\hline
	Policy   & $T=10$ &$T=50$ & $T=100$ & $T= 500$ & $T = 700$ \\
	\hline
	\hline
	JSQ &$15.838\pm 0.341$ &$15.838\pm 0.341$ &$15.838\pm 0.341$ &$15.838\pm 0.341$ &$15.838\pm 0.341$\\
	\hline
	JBT-10  &$18.616\pm 0.352$ &$19.394\pm 0.357$ &$19.816\pm0.372$ &$20.848\pm0.383$ & $21.151\pm0.413$\\
	\hline
	JBT-5  &$18.749\pm 0.363$ &$19.434\pm 0.362$ &$19.967\pm0.365$ &$21.162\pm0.408$ & $21.554\pm0.445$\\
	\hline
	JBT-2  &$22.717\pm 0.392$ &$24.212\pm 0.380$ &$24.742\pm0.407$ &$26.645\pm0.483$ & $26.971\pm0.492$\\
	\hline
	SQ(2,3) &$25.010\pm 0.353$ &$25.010\pm 0.353$ &$25.010\pm 0.353$ &$25.010\pm 0.353$ &$25.010\pm 0.353$\\
	\hline
	SQ(2)  &$31.469\pm 0.365$ &$31.469\pm 0.365$ &$31.469\pm 0.365$ &$31.469\pm 0.365$ &$31.469\pm 0.365$\\
	\hline
	JIQ &$41.456\pm 0.832$  &$41.456\pm 0.832$  &$41.456\pm 0.832$  &$41.456\pm 0.832$ &$41.456\pm 0.832$\\
	\hline
	\end{tabular}

\end{table*}

\begin{table*}

	\centering

	\caption{Number of messages per new job arrival under $50$ homogeneous servers with respect to $T$ }\label{tab:mess_N50_Td}

	\begin{tabular}{|c|c|c|c|c|c|}
	\hline
	Policy   & $T=10$ &$T=50$ & $T=100$ & $T= 500$ & $T = 700$ \\
	\hline
	\hline
	JSQ &$100$ &$100$ &$100$ &$100$ &$100$\\
	\hline
	JBT-50  &$10.749\pm 0.005$ &$2.690\pm 0.005$ &$1.679\pm0.007$ &$0.857\pm0.010$ & $0.800\pm0.010$\\
	\hline
	JBT-25  &$5.879\pm 0.004$ &$1.804\pm 0.004$ &$1.290\pm0.004$ &$0.865\pm0.010$ & $0.838\pm0.010$\\
	\hline
	JBT-10  &$3.085\pm 0.258$ &$1.298\pm 0.347$ &$1.086\pm0.326$ &$0.915\pm0.411$ & $0.909\pm0.395$\\
	\hline
	SQ(2,9) &$4$ &$4$ &$4$ &$4$ &$4$\\
	\hline
	SQ(2)  &$4$ &$4$ &$4$ &$4$ &$4$\\
	\hline
	JIQ &$0.511\pm 0.009$ &$0.511\pm 0.009$ &$0.511\pm 0.009$ &$0.511\pm 0.009$ &$0.511\pm 0.009$\\
	\hline
	\end{tabular}

\end{table*}

\begin{table*}

	\centering

	\caption{The $95\%$ confidence interval of response time under $50$ homogeneous servers with respect to $T$ }\label{tab:delay_N50_Td}

	\begin{tabular}{|c|c|c|c|c|c|}
	\hline
	Policy   & $T=10$ &$T=50$ & $T=100$ & $T= 500$ & $T = 700$ \\
	\hline
	\hline
	JSQ &$27.694\pm 0.056$ &$27.694\pm 0.056$ &$27.694\pm 0.056$ &$27.694\pm 0.056$ &$27.694\pm 0.056$\\
	\hline
	JBT-50  &$37.396\pm 0.121$ &$38.885\pm 0.136$ &$39.202\pm0.156$ &$40.357\pm0.183$ & $40.449\pm0.203$\\
	\hline
	JBT-25  &$34.732\pm 0.091$ &$35.611\pm 0.109$ &$35.807\pm0.110$ &$36.541\pm0.155$ & $36.596\pm0.143$\\
	\hline
	JBT-10  &$36.717\pm 0.100$ &$37.613\pm 0.137$ &$37.037\pm0.149$ &$36.714\pm0.158$ & $36.540\pm0.159$\\
	\hline
	SQ(2,9) &$89.872\pm 0.472$ &$89.872\pm 0.472$ &$89.872\pm 0.472$ &$89.872\pm 0.472$ &$89.872\pm 0.472$\\
	\hline
	SQ(2)  &$118.975\pm 0.576$ &$118.975\pm 0.576$ &$118.975\pm 0.576$ &$118.975\pm 0.576$ &$118.975\pm 0.576$\\
	\hline
	JIQ &$49.876\pm 0.523$  &$49.876\pm 0.523$  &$49.876\pm 0.523$  &$49.876\pm 0.523$ &$49.876\pm 0.523$\\
	\hline
	\end{tabular}

\end{table*}

\begin{table*}

	\centering

	\caption{Number of messages per new job arrival under $100$ homogeneous servers with respect to $T$ }\label{tab:mess_N100_Td}

	\begin{tabular}{|c|c|c|c|c|c|}
	\hline
	Policy   & $T=10$ &$T=50$ & $T=100$ & $T= 500$ & $T = 700$ \\
	\hline
	\hline
	JSQ &$200$ &$200$ &$200$ &$200$ &$200$\\
	\hline
	JBT-100  &$20.861\pm 0.006$ &$4.772\pm 0.004$ &$2.761\pm0.005$ &$1.146\pm0.006$ & $1.030\pm0.008$\\
	\hline
	JBT-50  &$10.963\pm 0.004$ &$2.847\pm 0.003$ &$1.837\pm0.004$ &$1.010\pm0.007$ & $0.960\pm0.007$\\
	\hline
	JBT-20  &$5.143\pm 0.004$ &$1.728\pm 0.003$ &$1.313\pm0.003$ &$0.979\pm0.006$ & $0.955\pm0.006$\\
	\hline
	SQ(2,15) &$4$ &$4$ &$4$ &$4$ &$4$\\
	\hline
	SQ(2)  &$4$ &$4$ &$4$ &$4$ &$4$\\
	\hline
	JIQ &$0.682\pm 0.006$ &$0.682\pm 0.006$ &$0.682\pm 0.006$ &$0.682\pm 0.006$ &$0.682\pm 0.006$\\
	\hline
	\end{tabular}

\end{table*}

\begin{table*}

	\centering

	\caption{The $95\%$ confidence interval of response time under $100$ homogeneous servers with respect to $T$ }\label{tab:delay_N100_Td}

	\begin{tabular}{|c|c|c|c|c|c|}
	\hline
	Policy   & $T=10$ &$T=50$ & $T=100$ & $T= 500$ & $T = 700$ \\
	\hline
	\hline
	JSQ &$51.666\pm 0.086$ &$51.666\pm 0.086$ &$51.666\pm 0.086$ &$51.666\pm 0.086$ &$51.666\pm 0.086$\\
	\hline
	JBT-100  &$65.064\pm 0.341$ &$66.965\pm 0.414$ &$67.221\pm0.490$ &$67.803\pm0.518$ & $68.104\pm0.677$\\
	\hline
	JBT-50  &$61.172\pm 0.238$ &$62.704\pm 0.285$ &$62.896\pm0.327$ &$63.771\pm0.524$ & $63.358\pm0.500$\\
	\hline
	JBT-20  &$61.554\pm 0.268$ &$62.301\pm 0.285$ &$62.065\pm0.328$ &$61.318\pm0.407$ & $61.133\pm0.415$\\
	\hline
	SQ(2,15) &$180.343\pm 1.833$ &$180.343\pm 1.833$ &$180.343\pm 1.833$ &$180.343\pm 1.833$ &$180.343\pm 1.833$\\
	\hline
	SQ(2)  &$237.873\pm 2.293$ &$237.873\pm 2.293$ &$237.873\pm 2.293$ &$237.873\pm 2.293$ &$237.873\pm 2.293$\\
	\hline
	JIQ &$73.049\pm 0.849$  &$73.049\pm 0.849$  &$73.049\pm 0.849$  &$73.049\pm 0.849$ &$73.049\pm 0.849$\\
	\hline
	\end{tabular}

\end{table*}

\newpage

\begin{table*}

	\centering

	\caption{Number of messages per new job arrival under $10$ homogeneous servers with Poisson arrival and constant service versus $T$ }\label{tab:mess_N10_Td_poiss_const}

	\begin{tabular}{|c|c|c|c|c|c|}
	\hline
	Policy   & $T=10$ &$T=50$ & $T=100$ & $T= 500$ & $T = 700$ \\
	\hline
	\hline
	JSQ &$20$ &$20$ &$20$ &$20$ &$20$\\
	\hline
	JBT-10  &$2.734\pm 0.003$ &$1.071\pm 0.005$ &$0.848\pm0.006$ &$0.656\pm0.013$ & $0.654\pm0.013$\\
	\hline
	JBT-5  &$1.856\pm 0.003$ &$0.992\pm 0.004$ &$0.865\pm0.005$ &$0.760\pm0.013$ & $0.751\pm0.010$\\
	\hline
	JBT-2  &$1.405\pm 0.003$ &$0.969\pm 0.003$ &$0.910\pm0.005$ &$0.859\pm0.009$ & $0.856\pm0.009$\\
	\hline
	SQ(2,3) &$4$ &$4$ &$4$ &$4$ &$4$\\
	\hline
	SQ(2)  &$4$ &$4$ &$4$ &$4$ &$4$\\
	\hline
	JIQ &$0.397\pm 0.07$ &$0.397\pm 0.07$ &$0.397\pm 0.07$ &$0.397\pm 0.07$ &$0.397\pm 0.07$\\
	\hline
	\end{tabular}

\end{table*}

\begin{table*}

	\centering

	\caption{The $95\%$ confidence interval of response time under $10$ homogeneous servers with Poisson arrival and constant service versus $T$ }\label{tab:delay_N10_Td_poiss_const}

	\begin{tabular}{|c|c|c|c|c|c|}
	\hline
	Policy   & $T=10$ &$T=50$ & $T=100$ & $T= 500$ & $T = 700$ \\
	\hline
	\hline
	JSQ &$9.962\pm 0.157$ &$9.962\pm 0.157$ &$9.962\pm 0.157$ &$9.962\pm 0.157$ &$9.962\pm 0.157$\\
	\hline
	JBT-10  &$12.000\pm 0.168$ &$12.381\pm 0.170$ &$12.581\pm0.172$ &$12.926\pm0.182$ & $12.960\pm0.183$\\
	\hline
	JBT-5  &$12.569\pm 0.168$ &$12.920\pm 0.168$ &$13.148\pm0.177$ &$13.683\pm0.199$ & $13.751\pm0.202$\\
	\hline
	JBT-2  &$17.043\pm 0.183$ &$18.294\pm 0.209$ &$18.338\pm0.243$ &$18.418\pm0.264$ & $18.761\pm0.282$\\
	\hline
	SQ(2,3) &$19.104\pm 0.217$ &$19.104\pm 0.217$ &$19.104\pm 0.217$ &$19.104\pm 0.217$ &$19.104\pm 0.217$\\
	\hline
	SQ(2)  &$24.736\pm 0.230$ &$24.736\pm 0.230$ &$24.736\pm 0.230$ &$24.736\pm 0.230$ &$24.736\pm 0.230$\\
	\hline
	JIQ &$15.467\pm 0.251$  &$15.467\pm 0.251$   &$15.467\pm 0.251$   &$15.467\pm 0.251$  &$15.467\pm 0.251$ \\
	\hline
	\end{tabular}

\end{table*}

\begin{table*}

	\centering

	\caption{Number of messages per new job arrival under $10$ homogeneous servers with Poisson arrival and bursty service versus $T$ }\label{tab:mess_N10_Td_poiss_burst}

	\begin{tabular}{|c|c|c|c|c|c|}
	\hline
	Policy   & $T=10$ &$T=50$ & $T=100$ & $T= 500$ & $T = 700$ \\
	\hline
	\hline
	JSQ &$20$ &$20$ &$20$ &$20$ &$20$\\
	\hline
	JBT-10  &$2.408\pm 0.003$ &$0.783\pm 0.004$ &$0.573\pm0.005$ &$0.371\pm0.013$ & $0.363\pm0.017$\\
	\hline
	JBT-5  &$1.615\pm 0.003$ &$0.750\pm 0.005$ &$0.621\pm0.006$ &$0.477\pm0.016$ & $0.466\pm0.017$\\
	\hline
	JBT-2  &$1.271\pm 0.003$ &$0.830\pm 0.005$ &$0.756\pm0.007$ &$0.660\pm0.016$ & $0.649\pm0.018$\\
	\hline
	SQ(2,3) &$4$ &$4$ &$4$ &$4$ &$4$\\
	\hline
	SQ(2)  &$4$ &$4$ &$4$ &$4$ &$4$\\
	\hline
	JIQ &$0.029\pm 0.003$ &$0.029\pm 0.003$ &$0.029\pm 0.003$ &$0.029\pm 0.003$ &$0.029\pm 0.003$\\
	\hline
	\end{tabular}

\end{table*}

\begin{table*}

	\centering

	\caption{The $95\%$ confidence interval of response time under $10$ homogeneous servers with Poisson arrival and bursty service versus $T$ }\label{tab:delay_N10_Td_poiss_burst}

	\begin{tabular}{|c|c|c|c|c|c|}
	\hline
	Policy   & $T=10$ &$T=50$ & $T=100$ & $T= 500$ & $T = 700$ \\
	\hline
	\hline
	JSQ &$61.864\pm 2.771$ &$61.864\pm 2.771$ &$61.864\pm 2.771$ &$61.864\pm 2.771$ &$61.864\pm 2.771$\\
	\hline
	JBT-10  &$69.047\pm 4.879$ &$71.659\pm 4.796$ &$73.187\pm4.802$ &$78.771\pm5.082$ & $78.389\pm5.543$\\
	\hline
	JBT-5  &$66.133\pm 4.973$ &$68.687\pm 5.118$ &$70.028\pm5.021$ &$74.772\pm5.171$ & $75.977\pm5.296$\\
	\hline
	JBT-2  &$69.910\pm 5.054$ &$73.245\pm 5.499$ &$73.796\pm4.987$ &$80.246\pm5.456$ & $80.611\pm6.044$\\
	\hline
	SQ(2,3) &$71.298\pm2.815$ &$71.298\pm2.815$ &$71.298\pm2.815$ &$71.298\pm2.815$ &$71.298\pm2.815$\\
	\hline
	SQ(2)  & $84.208\pm2.864$ & $84.208\pm2.864$ & $84.208\pm2.864$ & $84.208\pm2.864$ & $84.208\pm2.864$\\
	\hline
	JIQ & $280.949\pm16.030$  & $280.949\pm16.030$   & $280.949\pm16.030$  & $280.949\pm16.030$  & $280.949\pm16.030$\\
	\hline
	\end{tabular}

\end{table*}

\newpage

\begin{table*}

	\centering

	\caption{Number of messages per new job arrival under $10$ homogeneous servers with bursty arrival and Poisson service versus $T$ }\label{tab:mess_N10_Td_burst_poiss}

	\begin{tabular}{|c|c|c|c|c|c|}
	\hline
	Policy   & $T=10$ &$T=50$ & $T=100$ & $T= 500$ & $T = 700$ \\
	\hline
	\hline
	JSQ &$20$ &$20$ &$20$ &$20$ &$20$\\
	\hline
	JBT-10  &$3.164\pm 0.005$ &$1.159\pm 0.008$ &$0.874\pm0.009$ &$0.625\pm0.016$ & $0.618\pm0.021$\\
	\hline
	JBT-5  &$2.145\pm 0.005$ &$1.082\pm 0.007$ &$0.918\pm0.010$ &$0.721\pm0.020$ & $0.732\pm0.023$\\
	\hline
	JBT-2  &$1.641\pm 0.004$ &$1.115\pm 0.006$ &$1.025\pm0.008$ &$0.914\pm0.021$ & $0.896\pm0.027$\\
	\hline
	SQ(2,3) &$4$ &$4$ &$4$ &$4$ &$4$\\
	\hline
	SQ(2)  &$4$ &$4$ &$4$ &$4$ &$4$\\
	\hline
	JIQ &$0.177\pm 0.017$ &$0.177\pm 0.017$ &$0.177\pm 0.017$ &$0.177\pm 0.017$ &$0.177\pm 0.017$\\
	\hline
	\end{tabular}

\end{table*}

\begin{table*}

	\centering

	\caption{The $95\%$ confidence interval of response time under $10$ homogeneous servers with bursty arrival and Poisson service versus $T$ }\label{tab:delay_N10_Td_burst_poiss}

	\begin{tabular}{|c|c|c|c|c|c|}
	\hline
	Policy   & $T=10$ &$T=50$ & $T=100$ & $T= 500$ & $T = 700$ \\
	\hline
	\hline
	JSQ &$21.872\pm 0.698$ &$21.872\pm 0.698$ &$21.872\pm 0.698$ &$21.872\pm 0.698$ &$21.872\pm 0.698$\\
	\hline
	JBT-10  &$24.636\pm 0.720$ &$25.845\pm 0.741$ &$26.269\pm0.743$ &$27.838\pm0.745$ & $28.284\pm0.747$\\
	\hline
	JBT-5  &$24.694\pm 0.728$ &$25.755\pm 0.732$ &$26.558\pm0.763$ &$28.846\pm0.783$ & $29.147\pm0.806$\\
	\hline
	JBT-2  &$29.164\pm 0.739$ &$30.869\pm 0.780$ &$32.121\pm0.805$ &$34.599\pm0.959$ & $35.293\pm0.989$\\
	\hline
	SQ(2,3) &$29.333\pm0.728$ &$29.333\pm0.728$ &$29.333\pm0.728$ &$29.333\pm0.728$ &$29.333\pm0.728$\\
	\hline
	SQ(2)  & $38.513\pm0.791$ & $38.513\pm0.791$ & $38.513\pm0.791$ & $38.513\pm0.791$ & $38.513\pm0.791$\\
	\hline
	JIQ & $53.776\pm1.853$  & $53.776\pm1.853$   & $53.776\pm1.853$ & $53.776\pm1.853$ & $53.776\pm1.853$\\
	\hline
	\end{tabular}

\end{table*}

\end{document}